\documentclass[11pt,letter]{article}

\usepackage[T1]{fontenc}

\usepackage{fullpage}
\usepackage{graphicx} 
\usepackage{dcolumn}
\usepackage{bm}
\usepackage{amsmath}
\usepackage{amsfonts}
\usepackage{amssymb}
\usepackage{appendix}
\usepackage{comment}
\usepackage{cite}
\usepackage{verbatim}
\usepackage{cases}
\usepackage{bbm}
\usepackage{caption,mdframed,setspace}
\usepackage{alltt}
\usepackage{amsmath,amsfonts,amsthm,alltt,amssymb}

\usepackage[usenames,dvipsnames,svgnames,table]{xcolor}

\usepackage[inline]{enumitem}

\usepackage[procnumbered,ruled,vlined,linesnumbered]{algorithm2e}
\DontPrintSemicolon

\makeatletter
\newcommand{\AlgoResetCount}{\renewcommand{\@ResetCounterIfNeeded}{\setcounter{AlgoLine}{0}}}
\newcommand{\AlgoNoResetCount}{\renewcommand{\@ResetCounterIfNeeded}{}}
\newcounter{AlgoSavedLineCount}

\makeatother

\usepackage{mathrsfs}

\usepackage[colorlinks=true,
citecolor=Green]{hyperref}
\usepackage{url}

\usepackage{aliascnt}

\newtheorem{theorem}{Theorem}[section]

\newaliascnt{corollary}{theorem}

\aliascntresetthe{corollary}

\newaliascnt{lemma}{theorem}
\newtheorem{lemma}[lemma]{Lemma}
\aliascntresetthe{lemma}

\newaliascnt{definition}{theorem}
\newtheorem{definition}[definition]{Definition}
\aliascntresetthe{definition}

\newtheorem{claim}[theorem]{Claim}
\newtheorem{fact}[theorem]{Fact}

\theoremstyle{definition}

\def\equationautorefname~#1\null{Equation~(#1)\null}	

\usepackage{thmtools}
\usepackage{thm-restate}
\usepackage{thm-autoref}

\global\long\def\boldVar#1{\mathbf{#1}}
\global\long\def\mvar#1{\boldVar{#1}}

\newcommand{\mxhat}{\widehat{\mvar{X}}}

\newcommand{\prob}[2]{\mathbb{P}_{#1}\left[ #2 \right]}

\def\expec#1#2{{\mathbb{E}}_{#1}\left[ #2 \right]}

\def\abs#1{\left|#1  \right|}
\newcommand{\norm}[1]{{\left\| #1 \right\|}}
\newcommand{\normInline}[1]{{\| #1 \|}}

\def\defeq{\stackrel{\mathrm{def}}{=}}
\def\setof#1{\left\{#1  \right\}}

\def\pr#1{\left ( #1 \right )}

\DeclareMathOperator*{\argmin}{arg\,min}

\newcommand\Otil{\widetilde{O}}
\newcommand\rea{\mathbb R}

\newcommand{\zero}{\mathbf{0}}
\DeclareMathOperator{\nnz}{nnz}
\newcommand{\vecind}{\boldsymbol{\mathit{\chi}}}

\newcommand{\vecone}{\bm{1}}
\newcommand{\veczero}{\bm{0}}
\newcommand{\matzero}{\bm{0}}

\newcommand{\vzero}{\veczero}
\newcommand{\vones}{\vecone}
\newcommand{\ma}{\mvar{A}}

\newcommand{\md}{\mvar{D}}

\newcommand{\mI}{\mvar{I}}
\newcommand{\mm}{\mvar{M}}
\newcommand{\ms}{\mvar{S}}
\newcommand{\mU}{\mvar{U}}

\newcommand{\mlap}{{\mvar{{{L}}}}}
\newcommand{\mzero}{\mvar{0}}
\newcommand{\E}{\mathbb{E}}
\newcommand{\R}{\rea}

\newcommand{\otilde}{\widetilde{O}}

\newcommand{\mApxNorm}{\mvar{F}}

\newcommand{\Ltil}{\widetilde{\mvar{L}}}
\newcommand{\Fhat}{\widehat{\mvar{F}}}
\newcommand{\pseudoOp}{\dagger}
\newcommand{\pseudoRoot}{\dagger/2}

\renewcommand{\sc}[2]{\textsc{Sc}\left(#1, #2\right)}

\renewcommand{\u}[1]{\mvar{U}_{#1}}
\renewcommand{\c}[1]{\mvar{C}_{#1}}

\newcommand{\lapid}{\mvar{\Pi}}

\newcommand{\eulLU}{\textsc{EulerianLU}}

\newcommand{\matlow}{\boldsymbol{\mathit{{\mathfrak{L}}}}}

\newcommand{\matup}{\boldsymbol{\mathit{{\mathfrak{U}}}}}

\def\aa{\pmb{\mathit{a}}}
\newcommand\bb{\boldsymbol{\mathit{b}}}
\newcommand\cc{\boldsymbol{\mathit{c}}}

\renewcommand\ll{\boldsymbol{\mathit{l}}}

\newcommand\rr{\boldsymbol{\mathit{r}}}

\newcommand\vv{\boldsymbol{\mathit{v}}}

\newcommand\yy{\boldsymbol{\mathit{y}}}
\newcommand\zz{\boldsymbol{\mathit{z}}}
\newcommand\xx{\boldsymbol{\mathit{x}}}

\newcommand\xxhat{\boldsymbol{\mathit{\widehat{x}}}}
\newcommand\xxtil{\boldsymbol{\mathit{\tilde{x}}}}

\renewcommand\AA{\mvar{A}}
\newcommand\BB{\mvar{B}}
\newcommand\CC{\mvar{C}}
\newcommand\DD{\mvar{D}}

\newcommand\II{\mvar{I}}

\newcommand\NN{\mvar{N}}
\newcommand\MM{\mvar{M}}
\newcommand\LL{\mvar{L}}
\newcommand{\PP}{\mvar{P}}
\newcommand{\QQ}{\mvar{Q}}
\newcommand\RR{\mvar{R}}
\renewcommand\SS{\mvar{S}}
\newcommand\TT{\mvar{T}}
\newcommand\UU{\mvar{U}}
\newcommand\WW{\mvar{W}}
\newcommand\VV{\mvar{V}}
\newcommand\XX{\mvar{X}}

\newcommand\ZZ{\mvar{Z}}

\newcommand{\vstar}[2]{\textsc{St}\!\left[#1\right]_{#2}}

\newcommand{\dlap}{\LL}

\newcommand{\mdir}{\NN}

\newcommand\wprob{\delta}
\newcommand{\eps}{\epsilon}

\DeclareMathOperator*{\im}{im}

\DeclareMathOperator*{\expct}{\mathbb{E}}

\newcommand{\trp}{\top}

\newcommand{\todo}[1]{{\bf \color{red} TODO: #1}}
\newcommand{\mc}[1]{{\bf \color{green} Michael: #1}}
\newcommand{\jk}[1]{{\bf \color{green} Jon: #1}}
\newcommand{\jp}[1]{{\bf \color{green} John: #1}}
\newcommand{\rp}[1]{{\bf \color{green} Richard: #1}}
\newcommand{\as}[1]{{\bf \color{green} Aaron : #1}}
\newcommand{\sidford}[1]{\as{#1}}
\newcommand{\rk}[1]{{\bf \color{green} Rasmus: #1}}

\newcommand{\todolow}[1]{{\bf \color{magenta} todolow: #1}}
\newcommand{\todolowest}[1]{{\bf \color{gray} todolowest: #1}}
\newcommand{\mhelp}[1]{{\reversesmarginpar{\color{cyan} #1}}}

\renewcommand{\todo}[1]{}
\renewcommand{\mc}[1]{}
\renewcommand{\jk}[1]{}
\renewcommand{\jp}[1]{}
\renewcommand{\rp}[1]{}
\renewcommand{\as}[1]{}
\renewcommand{\rk}[1]{}
\renewcommand{\todolow}[1]{}
\renewcommand{\todolowest}[1]{}
\renewcommand{\mhelp}[1]{}

\newcommand{\SinglePhase}{\textsc{SinglePhase}}
\newcommand{\SparsifyEulerian}{\textsc{SparsifyEulerian}}

\renewcommand{\Re}{\R}

\begin{document}

\title{
Solving Directed Laplacian Systems in Nearly-Linear Time through Sparse LU Factorizations
}
\date{\today}
\author{Michael B. Cohen\thanks{This material is based upon work supported by the National Science Foundation under Grant No. 1111109.}\\
MIT\\
micohen@mit.edu\and Jonathan Kelner\footnotemark[1] \\
MIT\\
kelner@mit.edu\and Rasmus Kyng\thanks{This material is based upon work supported by ONR Award N00014-16-1-2374.}\\
Yale University\\
rjkyng@gmail.com \and  John Peebles\thanks{This material is based upon work supported by the National Science
Foundation Graduate Research Fellowship under Grant No. 1122374 and
by the National Science Foundation under Grant No. 1065125.}\\
MIT\\
jpeebles@mit.edu\and  Richard Peng\thanks{This material is based upon work supported by the National Science
	Foundation under Grant No. 1637566.}\\
Georgia Tech \\
rpeng@cc.gatech.edu\and  Anup B. Rao\\
Georgia Tech\\
anup.rao@gatech.edu\and  Aaron Sidford\\
Stanford University\\
sidford@stanford.edu }

\maketitle

\begin{abstract}
In this paper, we show how to solve directed Laplacian systems in nearly-linear time. Given a linear system in an $n \times n$ Eulerian directed Laplacian with $m$ nonzero entries, we show how to compute an $\epsilon$-approximate solution in time $O(m \log^{O(1)} (n) \log (1/\epsilon))$. Through reductions from [Cohen et al. FOCS'16], this gives the first nearly-linear time algorithms for computing $\epsilon$-approximate solutions to row or column diagonally dominant linear systems (including arbitrary directed Laplacians) and computing $\epsilon$-approximations to various properties of random walks on  directed graphs, including stationary distributions, personalized PageRank vectors, hitting times, and escape probabilities. These bounds improve upon the recent almost-linear algorithms of [Cohen et al. STOC'17],
which gave an algorithm to solve Eulerian Laplacian systems in time  
$O((m+n2^{O(\sqrt{\log n \log \log n})})\log^{O(1)}(n  \epsilon^{-1}))$.

To achieve our results, we provide a structural result that we believe is of independent interest. We show that Eulerian Laplacians (and therefore the Laplacians of all strongly connected directed graphs) have sparse approximate LU-factorizations. That is, for every such directed Laplacian $\mlap$, there is a lower triangular matrix $\matlow$ and an upper triangular matrix $\matup$, each with at most $\tilde{O}(n)$ nonzero entries, such that their product $\matlow \matup$ spectrally approximates $\mlap$ in an appropriate norm. This claim can be viewed as an analogue of recent work on sparse Cholesky factorizations of Laplacians of undirected graphs. We show how to construct such factorizations in nearly-linear time and prove that, once constructed, they yield nearly-linear time algorithms for solving directed Laplacian systems.
\end{abstract}


\newpage
\section{Introduction}
\label{sec:introduction}

A matrix $\mm$ is \emph{(column) diagonally dominant} if  $\mm_{ii}\geq \sum_{j\neq i}|\mm_{ji}|$ for all $i$.  Such matrices, which notably include Laplacians of graphs, are ubiquitous in computer science, with applications spanning scientific computing, graph theory, machine learning, and the analysis of random processes, among others.
For \emph{symmetric} diagonally dominant matrices, which include the Laplacians of  \emph{undirected} graphs, Spielman and Teng gave an algorithm in 2004~\cite{ST04} to solve the corresponding linear systems in nearly-linear time. Since then, the ability solve such linear systems has emerged as a powerful algorithmic primitive, serving as a crucial subroutine in the design of faster algorithms for a long list of problems (e.g.,\cite{LeeS14,Madry16,Madry13,CKMST,madry2015fast,CohenMSV17,KelnerMillerPeng,KelnerMadry,spielman2011graph,DaitchSpielman,ShermanBreaking,KelnerMaymounkov,OSV}).

Moreover, the pursuit of faster, simpler, and more parallelizable Laplacian system solvers has driven numerous algorithmic advances, comprising both improvements in the underlying linear algebraic machinery~\cite{SpielmanTengSolver:journal,KoutisMP10,KoutisMP11,KelnerOSZ13,lee2013efficient,CohenKMPPRX14,PengS14,KyngLPSS16,KyngS16} and applications of their ideas and techniques to problems in other domains (e.g.,\cite{KLOS14,DurfeeKPRS16,Madry10,lee2013efficient,AbrahamDKKP16:arxiv}).
 
However, while this approach has been incredibly successful for symmetric linear systems and undirected graph optimization problems, comparable results for their asymmetric or directed counterparts have proven quite elusive. 
In particular, the techniques for solving Laplacian systems seemed to rely intrinsically on multiple properties of undirected graphs, and, until recently, the best algorithms in the directed case simply treated the Laplacians as unstructured matrices and applied general linear algebraic routines, leading to super-quadratic running times.

Two recent papers~\cite{cohen2016faster,CohenKPPRSV16} suggested that it may be possible to close this gap, potentially laying the foundation for a new class of nearly-linear time algorithms for directed graphs and asymmetric linear systems.
The first paper~\cite{cohen2016faster} showed that linear systems involving several natural classes of asymmetric matrices, including Laplacians of directed graphs, general square column diagonally dominant matrices, and their transposes (called \emph{row diagonally dominant} matrices), 
could be reduced 
 with only polylogarithmic overhead to solving linear systems in the Laplacians of Eulerian graphs.
They further showed how to use these solvers, with only polylogarithmic overhead, to compute a wide range of fundamental quantities associated with random walks on directed graphs, including the stationary distribution, personalized PageRank vectors, hitting times, and escape probabilities.
The paper combined these reductions with an algorithm to solve Eulerian Laplacian systems in time~$\otilde(m^{3/4}n + mn^{2/3})$ to achieve faster (but still significantly super-linear) algorithms for all of these problems.%
\footnote{ Following the notation and terminology of the previous papers, we use $\otilde$ notation to suppress terms that are polylogarithmic in $n$, the natural condition number of the problem $\kappa$, and the desired accuracy~$\epsilon$. We use the term ``nearly-linear'' for algorithms that run in time $\tilde{O}(m)=O(m) \log^{O(1)}(n \kappa \epsilon)$, and ``almost linear'' for algorithms that run in time $O(m(n\kappa \epsilon^{-1})^{o(1)})$. 
}
The second paper~\cite{CohenKPPRSV16} gave an improved solver for Eulerian systems that runs in almost-linear time $\otilde (m+n2^{O(\sqrt{\log n \log \log n})})$, 
providing almost-linear time algorithms for all of the problems reduced to such a solver in~\cite{cohen2016faster}.

In this paper, we close the algorithmic gap between the directed and undirected cases (up to polylogarithmic factors) by providing an algorithm to solve Eulerian Laplacian systems in time $\tilde{O}(m)$.  Combining this with the reductions from~\cite{cohen2016faster} yields the first nearly-linear time algorithms for all of the problems listed above.


To achieve our results, we prove a structural result that we believe to be of independent interest.
 We show that Laplacians of strongly connected directed graphs have sparse approximate LU-factorizations. 
 More precisely, we show that for every such directed Laplacian $\mlap \in \R^{n \times n}$, there is a lower triangular matrix $\matlow \in \R^{n \times n}$ and an upper triangular matrix $\matup \in \R^{n  \times n}$ such that  both matrices have at most $\otilde(n)$ nonzero entries, and $\matlow \matup$ spectrally approximates $\mlap$ in an appropriate norm. We show how to construct such factorizations in nearly-linear time for the special case of Laplacians of strongly connected directed graphs where weighted in-degree equals weighted out-degree for every vertex.
Such Laplacians are known as Eulerian Laplacians.
We prove that once constructed, these factorizations yield nearly-linear time algorithms for solving directed Laplacian systems, and also yield sparse LU-factorizations of all Laplacians of strongly connected directed graphs.

This claim is analogous to the result first obtained by Lee, Peng, and Spielman in~\cite{LeePS15}, which showed that undirected Laplacians have sparse Cholesky factorizations, and our algorithm builds on a combination of the techniques in~\cite{LeePS15, KyngS16} as well as the sparsification machinery developed in~\cite{CohenKPPRSV16}.
Unfortunately, as we will discuss below (\autoref{sub:main_ideas}), there were several aspects of the approach in~\cite{KyngS16} that relied strongly on properties of symmetric Laplacians that are not present in the asymmetric setting, and obtaining a similar result for directed Laplacians required the development of new algorithmic and analytical techniques.



\subsection{Our Results}

The main technical result in this paper is the following theorem, which asserts that we can compute sparse approximate LU factorizations of an Eulerian Laplacian in nearly-linear time. Eulerian Laplacians have a special property that a standard symmetrization of these ($\mlap + \mlap^T$) are PSD matrices and define useful norms. Our main algorithmic result for Eulerian Laplacians is:




\begin{restatable}[Sparse LU]{theorem}{thmeulLU}
\label{thm:eulLU}
Given an Eulerian Laplacian $\mlap \in \R^{n \times n}$ with $m$
nonzero entries and any $\epsilon \in (0, 1/2)$, and $\delta < 1/n$,
in ${\Otil(m + n \eps^{-8} \log^{O(1)}(1/\delta) )} $ time,
with probability at least $1-O(\delta)$
the algorithm $\eulLU(\mlap,\delta,\eps)$,
produces lower and upper triangular matrices $\matlow \in \R^{n \times n}$ and $\matup \in \R^{n \times n}$ such that for some symmetric PSD matrix $\mApxNorm \approx_{\text{poly}(n)} (\mlap + \mlap^\intercal)/2$,
 $(\matlow \matup)^\top \mApxNorm^{\pseudoOp} (\matlow \matup) \succeq 1/O(\log^2 n) \cdot \mApxNorm$, $\| \mApxNorm^{\pseudoOp / 2} (\mlap - \matlow \matup) \mApxNorm^{\pseudoOp / 2} \|_2 \leq \epsilon$ and $\max\{\nnz(\matlow),\nnz(\matup)\} \leq n \log^{O(1)}(1/\delta)  \cdot \epsilon^{-6}$.
\end{restatable}

\rp{Can we just set $\delta = 1 / poly(n)$,
aka. do w.h.p.?
That way we don't have any $O(1)$s in exponents.}

\sidford{One reviewer asked for the version of this theorem for general directed Laplacians. It might be worth stating this (or maybe even something about RCDD more generally in this section) however I do not have a strong view.}


For simplicity, we assume all real number computations
are exact throughout this paper.
However, we believe a crude numerical stability analysis
for constant length inputs in the fixed point precision
model similar to the one sketched in Appendix A
of~\cite{cohen2016faster} is possible.  \jk{Saying that ``we believe'' a stability analysis to be possible feels little weak.  Should we strengthen the statement?  Is there any part of the algorithm where this isn't totally straightforward?} \sidford{Agreed, I am for strengthening the statement.}
\rp{My worry is that because we pivot $O(\log{n})$ rounds, if there is a $poly(n)$ blowup of error at each round, it may compound to an overhead in bit-length of $\log{n}$. I'm not sure if we have a statement like `if we perturb L, its Schur complement stays close in the RCDD sense'. In general I feel like there are strictly more moving pieces here because it's no longer `multiply a bunch of constant conditioned matrices together'.}

The conditions of this theorem can easily be shown to yield that the pseudoinverse of $\matlow \matup$ is a good preconditioner for solving $\mlap \xx = \bb$ for Eulerian $\mlap$. Consequently, we show in \autoref{sec:algorithmOverview} as a corollary of this main theorem that we can solve Eulerian Laplacian systems in nearly-linear time.

\begin{restatable}[Nearly-Linear Time Solver for Eulerian Laplacians]{corollary}{eulsolve}
\label{cor:solver}
Given an Eulerian Laplacian $\mlap \in \R^{n \times n}$ and a vector
$\bb \in \R^n$ with $\bb \perp \vecone$, and $\epsilon \in (0, 1/2)$,
in $O(m \log^{O(1)} n \log(1/\epsilon))$ time we can w.h.p. compute an
$\epsilon$-approximate solution $\xxtil$ to $\mlap \xx = \bb$ in the sense that $\norm{\xxtil- \mlap^{\pseudoOp} \bb}_{\u{\mvar{L}}} \leq \epsilon \norm{\mlap^\pseudoOp \bb}_{\u{\mvar{L}}}$ where $\u{\mlap} = (\mlap + \mlap^\trp)/2$.
\end{restatable}


Combining this result with the reductions from general directed Laplacians to Eulerian Laplacians from~\cite{cohen2016faster} then leads to nearly-linear time algorithms for computing quantities related to directed random walks when the mixing time is $poly(n)$. Such a reduction is identical to the incorporation of the previous almost-linear time Eulerian solver outlined in Appendix D of~\cite{CohenKPPRSV16}. It results in running times that's $O(m \log^{O(1)}n \log^{O(1)}t_{mix} \log(1 / \epsilon)$: linear up to polylogarithmic factors in $n$, $m$, and a natural condition number-like quantity related to the mixing time fo the random walks, $t_{mix}$. The problems for which we readily obtain such a running time include:
\begin{itemize}
\item solving row diagonally dominant (or column diagonally dominant) linear systems including arbitrary (non-Eulerian) directed Laplacian systems
\item computing the stationary distribution of a Markov chain
\item personalized PageRank
\item obtaining polynomially good estimates of the mixing time of a Markov chain
\item computing the hitting time from one vertex to another
\item computing escape probabilities for any triple of vertices
\item approximately computing all-pairs commute times\footnote{If one wishes to compute commute times for a number of pairs greater than the number of edges in the graph, the runtime will be nearly-linear in the output size instead of the number of the number of edges.} 
\end{itemize}




\subsection{Overview of Approach}
\label{sub:main_ideas}%

Here we present the broad algorithmic and analytical approach we take to provably solving Eulerian Laplacian systems in nearly-linear time. Our algorithm is broadly inspired by recent nearly-linear time (undirected) Laplacian system solvers based on performing repeated vertex elimination and Schur complement sparsification to compute a sparse approximate Cholesky factorization \cite{KyngLPSS16,KyngS16}. 

Initially, one might hope to obtain our results by simply adapting the algorithm in \cite{KyngS16}, which is arguably the simplest known for provably solving undirected Laplacian systems in nearly-linear time. However, there are multiple substantial barriers to adapting any undirected Laplacian system solver based on Cholesky factorization to solve arbitrary Eulerian Laplacians. We obtain our results by providing new algorithmic and analytical tools to systematically overcome these issues. 

%
%
In the remainder of this subsection we outline the high level-ideas that underlie these tools and connect them to the later sections that present them rigorously. For the details of the mathematical notation we use see \autoref{sec:preliminaries}.

\paragraph{Vertex Elimination and Approximate LU-Factorization (\autoref{cha:dirlap})}
To solve a system of equations in an Eulerian Laplacian $\mlap \in \R^{V \times V}$,  we leverage the well-known fact that it suffices to compute a matrix $\ZZ$ that is a good \emph{preconditioner} or \emph{approximate pseudoinverse} (see \autoref{defn:approxInv}). As was shown in \cite{CohenKPPRSV16}, if $\ZZ$ can be applied in nearly-linear time and $\ZZ \mlap \approx \mI_{\im(\ZZ)}$ in a suitable norm, then this suffices to solve Laplacian systems in $\mlap$ in nearly-linear time through an iterative method known as \emph{preconditioned Richardson iteration} (see \autoref{lem:precond_richardson}). 

Whereas \cite{CohenKPPRSV16} constructed $\ZZ$ through repeated sparsification and squaring, here we take an approach inspired by~\cite{KyngLPSS16,KyngS16} based on  vertex elimination and sparsification of Schur complements. 
For any square matrix $\ma \in \R^{n \times n}$ where $F,C$ partition $[n]$ and $\ma_{FF}$ is invertible, it holds that
\[
\ma
=
\left[
\begin{array}{cc}
\ma_{FF} & \ma_{FC} \\
\ma_{CF} & \ma_{CC}
\end{array}
\right]
= 
\left[
\begin{array}{cc}
\mI & \mzero \\
\ma_{CF} \ma_{FF}^{-1} & \mI
\end{array}
\right]
\left[
\begin{array}{cc}
\ma_{FF} & \mzero \\
\mzero & \ma_{CC} - \ma_{CF} \ma_{FF}^{-1} \ma_{FC}
\end{array}
\right]
\left[
\begin{array}{cc}
\mI & \ma_{FF}^{-1} \ma_{FC} \\
\mzero & \mI
\end{array}
\right].
\]
We can easily invert the upper and lower triangular matrices in this factorization: 
\[
\left[
\begin{array}{cc}
\mI & \mzero \\
\ma_{CF} \ma_{FF}^{-1} & \mI
\end{array}
\right]^{-1}
=
\left[
\begin{array}{cc}
\mI & \mzero \\
-\ma_{CF} \ma_{FF}^{-1} & \mI
\end{array}
\right]
\text{ and }
\left[
\begin{array}{cc}
\mI & \ma_{FF}^{-1} \ma_{FC} \\
\mzero & \mI
\end{array}
\right]^{-1}
=
\left[
\begin{array}{cc}
\mI & -\ma_{FF}^{-1} \ma_{FC} \\
\mzero & \mI
\end{array}
\right]
,
\]
This reduces solving linear systems in $\ma$ to solving smaller linear systems in $\ma_{FF}$ and the \emph{Schur complement} $\ma_{CC} - \ma_{CF} \ma_{FF}^{-1} \ma_{FC}$.

The recent work of \cite{KyngLPSS16,KyngS16}
leverages the fact that, when $\ma$ is a undirected Laplacian, its Schur complement is as well. By cleverly choosing the block $F$ of vertices to eliminate, sparsifying the Schur complement, and recursing, these papers compute efficient preconditioners. 
In the work of \cite{KyngS16}, 
each $F$ is simply a single coordinate or vertex of the associated graph. In this case eliminating this vertex induces a Schur complement with the vertex removed and an appropriately weighted clique added to its neighbors. Repeated elimination and recursion then yields a $LU$ factorization of the original matrix. In the setting of  \cite{KyngS16} this matrix is symmetric and therefore the lower and upper triangular matrices are transposes of each other and the factorization is known as a \emph{Cholesky factorization}.
 By picking the vertex to eliminate randomly and sparsifying the cliques directly \cite{KyngS16} showed that a sparse approximate Cholesky factorization of the Laplacian can be obtained in nearly-linear time.

It is easy to see that for Eulerian Laplacians, it is also the case that its Schur complements are Eulerian Laplacians. Eliminating a single vertex corresponds to deleting that vertex and adding a weighted complete bipartite graph, or \emph{biclique}, from the vertices that yield incoming edges to the eliminated vertex to the vertices that yield outgoing edges from the eliminated vertex. Consequently, \cite{KyngS16} suggests a natural approach for solving directed Laplacian systems and producing sparse $LU$-factorizations of Eulerian Laplacians: repeatedly eliminate random vertices and directly sparsify the bicliques that eliminating these vertices induces.

Unfortunately, there are multiple issues with this approach. First, sparsifying directed Eulerian Laplacians is more delicate than sparsifying undirected Laplacians and even showing that sparsifiers exist was a major contribution of \cite{CohenKPPRSV16}. Second, analyzing the error induced by sparsification is much more difficult for directed Laplacians and reasoning about this error was the major contributor to the almost linear rather than nearly-linear running time of \cite{CohenKPPRSV16}. Third, the reasoning of \cite{KyngS16} required rather tight bounds on the sparsity induced by sparsification, and  showing that their algorithm works without this seems impossible. 

Nevertheless, we show how to overcome these issues by proving a biclique sparsification technique with favorable properties for our analysis, providing new results on the error induced by Schur complement sparsification as it relates to the quality of a preconditioner, and modifying the \cite{KyngS16} algorithm to only eliminate carefully chosen sets of vertices and alternate with full sparsification of the resulting graph. We discuss each of these further below. 

\paragraph{Unbiased Degree Preserving Vertex Elimination (\autoref{sec:SingleVertexElimination}):}

The first immediate issue in leveraging the insights from \cite{KyngS16} to develop a a single vertex elimination algorithm capable of producing sparse LU factorizations of an Eulerian Laplacians is to determine how to sparsify the bicliques that elimination creates. On the one hand, the analysis of \cite{KyngS16} crucially leverages that the sparsification procedure is unbiased, i.e. it produces a Laplacian that is in expectation the sparsified graph, with low variance so that matrix martingale arguments can be used to bound the error induced by the entire procedure. On the other the known sparsification results for Eulerian graphs \cite{CohenKPPRSV16} require that the sparsified graph is Eulerian and typically work by exactly preserving the in-degrees and out-degrees of vertices. 

Unfortunately, these two constraints on a sparsification procedure seem at odds with each other as independent sampling to create an unbiased estimator likely does not preserves degrees. Moreover, it is difficult to see how to completely drop either constraint. However, a different algorithmic primitive suggests optimism. Consider the edges of the vertex we eliminate.
If we treat these edges as flows coming into and out of the vertex, then running a standard flow path decomposition on the these flows will result in a collection of flow paths -- and if we treat these flow paths as edges that bypass the middle vertex, we get a graph on the neighbors of the eliminated vertex.
This graph will be sparse and have the same in- and out-degree for every neighbor as the biclique created by elimination. We show that it is possible to randomize this flow decomposition in such a way that the sparse graph becomes an unbiased estimator of the dense biclique.

Formally, in \autoref{sec:SingleVertexElimination} we remedy this issue by providing an efficient biclique sparsification procedure that is unbiased, with low variance, that exactly preserves the in and out degrees of vertices. This procedure works by careful sampling edges from the biclique so that no vertex has an excess of in-degree or out-degree and then carefully sampling from the remaining graph so that the conditional expectation is preserved. Crucially this procedure is not independent sampling from the entire graph. Moreover, we show that this procedure has low variance as desired. We hope the ideas of \autoref{sec:SingleVertexElimination} could be useful for sparsification more broadly.

%

\paragraph{Schur Complement Blow Up and Bounding Error Increase (\autoref{sec:SchurBounds} and \autoref{sec:ErrorAccumulation}):}

While our biclique sparsification procedure of \autoref{sec:SingleVertexElimination} provides hope of developing a single vertex elimination algorithm, it  leaves significant issues in bounding the error induced by sparsification.

First, a key fact used in the previous analysis~\cite{KyngS16} is that, when recursing, the Schur complement matrices that result are always spectrally dominated by the original matrix. This is used to show that the overall error of the algorithm is small as long as the error is small in every recursive phase. It is not even obvious how to define the right notion of ``spectral domination,'' in our context since the typical notion only applies to positive semidefinite matrices.
Moreover, once it is appropriately defined (using notions from~\cite{CohenKPPRSV16}), this spectral domination property fails even for simple graphs like the directed cycle, which can increase by a factor that is linear in the number of vertices. 

Nevertheless, in \autoref{sec:SchurBounds}
we show that it is possible to efficiently find a large fraction of the vertices such that eliminating them only causes a constant multiplicative growth in the relevant norms. 
We show that (highly) row-column diagonally dominant sets are well-behaved in this sense and thus we can simply perform random vertex elimination restricted to these large sets.
We note that though an analogous notion for undirected graphs was used in \cite{LeePS15,KyngLPSS15:arxiv}, it played a different role there. They show that Jacobi iterations for the Laplacian minor defined by a diagonally dominant set converges quickly and also, that a sparse approximation of Schur complement can be constructed efficiently. On the other hand, we use the row-column diagonal dominance  to show that an appropriately defined spectral norm of the Schur complement does not blow up. 

Unfortunately, if the multiplicative errors from eliminating the highly row-column diagonally dominant sets compounded, the overall error would still be too large for our elimination procedure to yield a good preconditioner. The previous solvers \cite{cohen2016faster, CohenKPPRSV16} essentially measured the quality of a preconditioner by bounding error in the norm induced by the undirected Laplacian obtained from an Eulerian Laplacian by removing the direction on every edge. However the compounding error from repeated Schur complement sparsification seems prohibitively large in this norm. In some sense, this was the critical issue which prevented \cite{CohenKPPRSV16} from achieving a nearly-linear running time. 


To circumvent this we develop a deeper understanding of the norms we can use to prove convergence of iterative solvers for directed Laplacian systems. We show that there is good deal of flexibility in the family of norms that can be used to prove convergence of iterative solvers for directed Laplacian systems (See \autoref{sec:ErrorAccumulation}).
%
%
%
Formally, we prove that based on our vertex elimination sequence, we can construct a new matrix (\autoref{eq:F}) that is both sufficient for proving the convergence of preconditioned Richardson iteration and for which the errors induced by repeated Schur complement sparsification is small. This matrix is a careful combination of the undirected graph norms induced by the sequence of Schur complement matrices we encountered.  We show this matrix spectrally dominates the ``large'' matrices created by vertex elimination and therefore allows us to  convert the local errors induced by one round of sparsification into a global bound on the quality of a preconditioner (\autoref{lem:CumulativeErrorNew}).
We ultimately prove that our solver converges quickly in this new norm, before converting the final error guarantee back to the norms we care about in \autoref{cor:solver} for application purposes. We believe that these tools for analyzing the error induced by sparsifying Schur complements could be useful for reasoning about asymmetric matrices more broadly.

\paragraph{Sparsification, Algorithm Design, and Analysis (\autoref{cha:dirlap} and \autoref{sec:singlePhase})}

The above 
gives the key building blocks for developing a nearly-linear time algorithm for solving Eulerian Laplacian systems. It suggests that picking a block of row-column diagonally dominant vertices, eliminating them with our unbiased degree-preserving vertex elimination procedure, and then analyzing this in a carefully chosen norm induced by the recursion could work. However, there is still one more issue that arises---the sparsity of the resulting Schur complements this procedure induces grows much faster then in \cite{KyngS16}. To overcome this, we simply leverage black box previous work from \cite{CohenKPPRSV16} that Eulerian Laplacians have sparsifiers which can be efficiently computed. Consequently, by occasionally sparsifying the resulting Schur complement we can fix this final issue.

Putting together these pieces yields our algorithm. \autoref{alg:sketch} gives a rough sketch of this algorithm, which alternates between vertex elimination on carefully chosen subsets and sparsification. 

\begin{algorithm}[H]							
\caption{Overall Sketch of the Algorithm}
\label{alg:sketch}
\SetAlgoVlined
					

\For{$i\gets 1$ \KwTo $O(\log n)$ }{
    Choose a highly row-column diagonally dominant set $F \subset V$ of size proportional to $|V|$. \label{algstep:ddSet}
    
    \For{$v \in F$  }{
    
    Vertex Elimination $v:$ Use the routine $\textsc{SingleVertexElim}$ to eliminate vertex $v$ and to add a sparse approximation of its Schur complement \label{algstep:vertexElim}
	
    Sparsify Graph: If the number of edges in the resulting graph is above a threshold, then use the routine  $\textsc{SparsifyEulerian}$ to sparsify the whole graph
    }
    $V \leftarrow V \backslash F$
    }
        
\end{algorithm}

Although this algorithm only describes how to eliminate vertices, as in Gaussian elimination, this also gives an LU factor decomposition. For a more detailed description of the algorithm, see \autoref{alg:multiPhase} and routines it calls. Details of \autoref{algstep:vertexElim} of \autoref{alg:sketch} can be found in \autoref{alg:rcdd}. See \autoref{alg:selim} for a description of $\textsc{SingleVertexElim}$. 

We provide the analysis of this procedure (assuming the pieces of the rest of the paper) in \autoref{cha:dirlap} and in \autoref{sec:singlePhase} we provide the careful martingale analysis of elimination and sparsification within a single highly row column diagonally dominant block. Though it has several pieces, ultimately, we believe this framework for solving Eulerian solvers is simpler than that in \cite{CohenKPPRSV16} and we hope that these pieces may find further use and possibly lead to even simpler algorithms.

\newcommand\AAcal{\boldsymbol{\mathcal{A}}}

\subsection{Paper Outline}

The presentation of these results is split into several pieces. In \autoref{sec:algorithmOverview} we give the main algorithm that computes an approximate LU-factorization of an Eulerian Laplacian and prove its correctness as well as \autoref{thm:eulLU} and \autoref{cor:solver}, assuming the analysis of later parts of the paper. In \autoref{sec:SingleVertexElimination}, we analyze our single vertex elimination or biclique sparsification procedure. In \autoref{sec:SchurBounds} we analyze the error induced by sparsifying Schur complements of highly diagonally dominant sets. In \autoref{sec:singlePhase} we perform the matrix martingale analysis of the error incurred by single vertex elimination and graph sparsification. In \autoref{sec:ErrorAccumulation} we provide proofs of facts regarding the norms we use to analyze the overall error of our approximate LU factorization procedure. In \autoref{sec:rcdd-analysis} we show how to find RCDD subsets and in \autoref{sec:mat_facts} we provide matrix facts we use throughout. 


\subsection{Preliminaries}
\label{sec:preliminaries}

\paragraph{Matrices:} For a square matrix $\mvar{M}$, we denote its symmetric part by 
$
\u{\mvar{M}} \defeq (1 /2 ) ( \mvar{M} + \mvar{M}^\top )
$. Typically, $\u{\mvar{M}} \succeq 0$ and the kernel of $\mvar{M}$ and $\mvar{M}^\top$ are same, i.e. $\ker(\mvar{M}) = \ker(\mvar{M}^\top)$.

For matrix $\ma \in \R^{n\times n}$ and subsets $F,C\subseteq[n]$
we let $\ma_{FC}\in\R^{F\times C}$ denote the sub-matrix of $\ma$
corresponding to the $F$ and $C$ entries. 
Furthermore, for $v\in\R^{N}$
and $F\subseteq[n]$ we let $v_{F}$ denote the restriction of $v$
to the coordinates of $v$. Consequently, if $F,C\subseteq[n]$ partition
$[n]$ then we have that 
\[
[\ma v]_{F}=\ma_{FF}v_{F}+\ma_{FC}v_{C}
\text{ and }
v^{\top}\ma v=v_{F}^{\top}\ma_{FF}v_{F}+v_{F}^{\top}\ma_{FC}v_{C}+v_{C}^{\top}\ma_{CF}v_{F}+v_{C}^{\top}\ma v_{C}\,.
\]
Such a partition naturally leads to the notion of Schur complements.
We will use $\sc{\mvar{A}}{C}$ to denote $n$-by-$n$ matrix where the only
non-zeros are on $C$, and these variables corresponding to the result
of eliminating all variables in $F$.
Formally, we let $\mzero_{FC}\in\R^{F\times C}$ be
the all zero matrix, from which we can write this Schur complement as: 
\[
\sc{ \mvar{A}}{F} \defeq\left[\begin{array}{cc}
\mzero_{FF} & \mzero_{FC}\\
\mzero_{CF} & \ma_{CC}-\ma_{CF}\ma_{FF}^{-1}\ma_{FC} 
\end{array}\right] ~.
\]

\paragraph{Norms:} 
Given PSD $\mvar{H} \in \R^{n \times n}$, we define a semi-norm on
vector $\norm{\cdot}_{\mvar{H}}$ by $\norm{x}_\mvar{H} = \sqrt{  x^\top
  \mvar{H} x }$.
For any norm
$\norm{\cdot}$ defined on $\R^{n}$ we define the \emph{seminorm} it induces on $\R^{n\times n}$ for all $\mvar{A}\in\R^{n\times n}$
by $\norm{\mvar{A}}=\max_{x\neq0} \norm{\mvar{A} x} / \norm{x}$.
When we wish to make clear that we are considering this ratio we use the $\rightarrow$ symbol; e.g., 
$\norm{\mvar{A}}_{\mvar{H} \to \mvar{H}}=\max_{x\neq 0} \norm{\mvar{A} x}_\mvar{H} / \norm{x}_\mvar{H}$ 
but we may also simply write $\norm{\mvar{A}}_{\mvar{H}} \defeq \norm{\mvar{A}}_{\mvar{H} \rightarrow \mvar{H}}$. 


\paragraph{Pseudoinverse and Square Roots:} For symmetric PSD matrix $\ma$ we use $\ma^\pseudoOp$ to denote its Moore-Penrose pseudoinverse of $\ma$, we use $\ma^{1/2}$ to denote its square root, i.e. the unique PSD matrix such that $[\ma^{1/2}]^2 = \ma$, and we use $\ma^{\pseudoRoot}$ to denote $[\ma^\pseudoOp]^{1/2} = [\ma^{1/2}]^\pseudoOp$.

\paragraph{Upper and Lower Triangular Matrices:} We say a square matrix $\matup$
  is upper triangular if it has non-zero entries $\matup(i,j) \neq 0$
  only for $i \leq j$ (i.e. above the diagonal). 
  Similarly, we say a square matrix $\matlow$
  is lower triangular if it has non-zero entries $\matup(i,j) \neq 0$
  only for $i \geq j$ (i.e. below the diagonal). 
  Often, we work with matrices that are not upper or lower
  triangular, but for which we know a permutation matrix $\PP$ s.t. 
  $\PP \matup \PP^{\trp}$ is upper (respectively lower) triangular.
  For computational purposes, this is essentially equivalent to having
  a upper or lower triangular matrix, and we refer to such
    matrices as upper (or lower) triangular. 
  The algorithms we develop for factorization will always compute the
  necessary permutation.

\paragraph{Asymmetric Matrix Approximation:} We use the asymmetric matrix approximation definition of \cite{CohenKPPRSV16} and say that 
matrix $\BB$ is said to be an \emph{$\epsilon$-approximation of matrix $\AA$} if and only if $\UU_\AA$ is symmetric PSD with 
$\ker(\UU_\AA) \subseteq \ker(\AA - \BB) \cap \ker((\AA - \BB)^\top)$ and       $\normInline{\u{\AA}^{\pseudoRoot} (\AA -
            \BB) \u{\AA}^{\pseudoRoot}}_2 \leq \eps$.
In this paper we only use this definition in the restricted setting where $\ker(\ma) = \ker(\ma^\top)$.

\paragraph{Row Column Diagonal Dominant (RCDD):} For square matrix $\mvar{A}$ we say that a subset of the coordinates $F$ is $\alpha$-RCDD if $\sum_{j\in F,j\neq i}|\mlap_{ij}|\leq\frac{1}{1+\alpha}|\mlap_{ii}|$ and $\sum_{j\in F,j\neq i}|\mlap_{ji}|\leq\frac{1}{1+\alpha}|\mlap_{ii}|$. 

\paragraph{Directed Laplacians:} We follow the conventions of \cite{cohen2016faster} and say matrix $\mlap\in\R^{n\times n}$ is a \emph{directed Laplacian} if its off diagonal entries are non-positive, i.e. $\mlap_{i,j} \leq 0$ for all $i \neq j$, and $\vones^\top \mlap = \vzero$, i.e.  $\mlap_{ii}=-\sum_{j\neq i}\mlap_{ji}$ for all $i$. For every directed Laplacian $\mlap \in \R^{n \times n}$ we associate a
graph $G_{\mlap}=(V,E,w)$ with vertices $V = [n]$, and edges $(i,j)$
of weight $w_{ij} =-\mlap_{ji}$, for all $i\neq j\in[n]$ with $\mlap_{ji}\neq0$.
Occasionally we write $\mlap=\md-\ma^{\top}$ to denote that we decompose
$\mlap$ into the diagonal matrix $\md$ (where $\md_{ii}=\mlap_{ii}$
is the out degree of vertex $i$ in $G_{\mlap}$) and non-negative matrix $\ma$ (which is weighted
adjacency matrix of $G_{\mlap}$, with $\ma_{ij}=w_{ij}$ if $(i,j)\in E$,
and $\ma_{ij}=0$ otherwise).
Letting $\vecind_v \in \rea^{n}$ be the vector whose $v$'th coordinate is set to one and all others to be zero and $\bb_{(u,v)} = \vecind_v - \vecind_u$, a directed Laplacian can be written as
$
\mlap = \sum_{(v,u) \in E} w_{(v,u)}\bb_{(u,v)}\vecind_{v}^\top$.

Finally, we call a matrix $\mlap$ is an
\emph{Eulerian Laplacian} if it is a directed Laplacian with
$\mlap\vecone = \veczero$. Note that $\mlap$ is an \emph{Eulerian
  Laplacian} if and only if its associated graph is Eulerian.

\section{Main Algorithm}
\label{cha:dirlap}


\label{sec:dirlap}
\label{sec:algorithmOverview}

In this section, we give an overview of the key components of our
algorithm for LU factorization and show how to use an LU factorization
it generates to solve Eulerian Laplacians. We also give proofs of
our two main statements on Eulerian Laplacians (\autoref{thm:eulLU}
and \autoref{cor:solver}) assuming auxiliary statements proven in this
section and later in the chapter.
\paragraph{Standard LU Factorization.}
The core of our algorithm is a sparsified version of LU factorization.
The starting point for standard LU factorization is the formula for
eliminating a single variable.
Letting $C = V \setminus \setof{v}$, the formula for eliminating a
single variable in a matrix $\mlap$ is
\begin{align*}
\mlap
=
\left[
\begin{array}{cc}
\mlap_{vv} & \mlap_{vC} \\
\mlap_{Cv}& \mlap_{CC}
\end{array}
\right]
&= 
\left[
\begin{array}{cc}
1 & \mzero \\
\frac{1}{\mlap_{vv}^{1/2}}\mlap_{Cv} & \mI
\end{array}
\right]
\left[
\begin{array}{cc}
1 & \mzero \\
  \mzero & \mlap_{CC}  -  \frac{1}{\mlap_{vv}} \mlap_{Cv} \mlap_{vC}
\end{array}
\right]
\left[
\begin{array}{cc}
1 & \frac{1}{\mlap_{vv}^{1/2}} \mlap_{vC} \\
\mzero & \mI
\end{array}
\right]
\end{align*}
We have decomposed the matrix into a lower triangular matrix, a block
diagonal matrix, and an upper triangular matrix.
Recursively applying the same elimination procedure to the matrix
$\mlap_{CC}  -  \frac{1}{\mlap_{vv}} \mlap_{Cv} \mlap_{vC}$ on the
block diagonal, we can then get a decomposition into a product of a
sequence of lower triangular matrices times a sequence of upper
triangular matrices.
Since the product of lower triangular matrices is lower triangular,
and similarly the product of upper triangular matrices is upper
triangular, we can then collect these into one lower triangular factor
$\mathfrak{L}$ 
and an upper triangular factor $\mathfrak{U}$ and write $\mlap =
\mathfrak{L} \mathfrak{U}$.
However, one can check that in fact the factors $\mathfrak{L}$ and
$\mathfrak{U}$ have a simple structure: 
If $d_i$ is the diagonal entry corresponding to the variable being
eliminated in round $i$, and $\cc_i$ is the corresponding column
scaled by  $d_i^{-1/2}$, and $\rr_i$ the row scaled scaled by $d_i^{-1/2}$
(e.g. so that in the
first round $d_1 = \mlap_{vv}$,
$\cc_1 = d_1^{-1/2} \mlap_{vV}$, and $\rr_1 = d_1^{-1/2} \mlap_{Vv}$),
then
\[
\mathfrak{L} = 
\begin{bmatrix}
\cc_1  & 0     &  0       & \\
   \vline &    \cc_2 &  0       & \cdots \\
    \vline &         \vline       & \cc_3 &\\
\end{bmatrix}
\text{ and }
\mathfrak{U} = 
\begin{bmatrix}
   \rr_1 & \hspace{-0.5em} \rule[0.5ex]{2em}{0.4pt} & \hspace{-1.5em}\rule[0.5ex]{2em}{0.4pt}\\
    0 &   \rr_2 & \hspace{-0.5em}\rule[0.5ex]{2em}{0.4pt} \\
    0 & 0 &  \rr_3 \\
    & \vdots & 
\end{bmatrix}
.
\]
For us, it will be more convenient to consider all rows and columns to be of
length $n$, and abusing notation by redefining  $\cc_2 \leftarrow
\begin{pmatrix}
 0 \\
 \cc_2
\end{pmatrix}
$ and so on, we can then write
\[
\mathfrak{L} = 
\begin{bmatrix}
\cc_1 &  \cc_2  & \cc_3 & \cdots \\
 \vline &   \vline  & \vline&
\end{bmatrix}
\text{ and }
\mathfrak{U} = 
\begin{bmatrix}
  \rr_1 & \hspace{-0.5em}\rule[0.5ex]{2em}{0.4pt} \\
  \rr_2 & \hspace{-0.5em}\rule[0.5ex]{2em}{0.4pt} \\
  \rr_3 & \hspace{-0.5em}\rule[0.5ex]{2em}{0.4pt} \\
  \vdots
\end{bmatrix}
\text{ and } 
\mlap = \mathfrak{L} \mathfrak{U}  = \sum_{i} \cc_i \rr_i.
\]

\rp{is this sentence/block needed?}
The fact that $C = V \setminus \{v\}$ allows us to write the top-left
block as the scalar entry $\LL_{vv}$, giving
\[
\sc{ \mvar{L}}{ C } = \left[
\begin{array}{cc}
0 & \mzero \\
  \mzero & \mlap_{CC}  -  \frac{1}{\mlap_{vv}} \mlap_{Cv} \mlap_{vC}
\end{array}
\right]
.
\]
and after simple algebraic manipulations,
\begin{align}
\label{eq:additiveschur}
\mlap
=
\frac{1}{\mlap_{vv}} \mlap_{Vv} \mlap_{vV}
+
\sc{ \mvar{L}}{ C }
=
\cc_1 \rr_1
+
\sc{ \mvar{L}}{ C }
\end{align}

\paragraph{LU Factorization of Eulerian Laplacians.}
Let us restrict our attention to LU factorization of a matrix $\mvar{L}$
which is a strongly connected Eulerian Laplacian.
In general, LU factorization may run into trouble if one tries to
eliminate a variable with a zero diagonal entry. 
But, one can show that this will never happen if the input matrix is
an Eulerian Laplacian -- except that the final diagonal entry will
always be zero.

We introduce notation for the directed Laplacian corresponding to
the star graph of edges coming into and going out of a vertex $v$ of
in the graph corresponding to $\mlap$, which we write as
$$ \vstar{\mvar{L}}{v} = 
\sum_{(v,u) \in E}
w_{(v,u)}\bb_{(u,v)}\vecind_{v}^T 
+ 
\sum_{(u,v) \in E} w_{(u,v)}\bb_{(v,u)}\vecind_{u}^T ,$$
Rearranging Equation~\eqref{eq:additiveschur}, we have
\begin{align*}
\label{eq:additiveschur}
\sc{ \mvar{L}}{ V \setminus \setof{v} }
=
\mlap
-
\frac{1}{\mlap_{vv}} \mlap_{Vv} \mlap_{vV}
=
\mlap
-
\vstar{\mvar{L}}{v} 
+
\vstar{\mvar{L}}{v} 
-
\frac{1}{\mlap_{vv}} \mlap_{Vv} \mlap_{vV}
\end{align*}
Importantly, when $\mvar{L}$ is a strongly connected Eulerian Laplacian, then so is $\sc{
  \mvar{L}}{ V \backslash \{v \} } $, in fact Eulerian Laplacians are
closed under taking Schur complements, and strong connectivity is preserved.
Furthermore, $\mlap
-
\vstar{\mvar{L}}{v}$ is simply the directed Laplacian corresponding to the
graph of $\mlap$ with $v$ and the edges incident on it removed,
and $\vstar{\mvar{L}}{v} 
-
\frac{1}{\mlap_{vv}} \mlap_{Vv} \mlap_{vV}$, is a directed
Laplacian of a weighted biclique on the neighbors of $v$.
In general, neither of these parts of the Schur complement are
individually Eulerian, but together they are.

\paragraph{Our algorithm.}
As stated in \autoref{thm:eulLU}, we assume that the input matrix, $\mvar{L}$, is a strongly
connected Eulerian Laplacian.
Our algorithm is based on standard LU factorization as outlined above.
Since each elimination creates a biclique on the neighbors of the
vertex being eliminated, the graph corresponding to the remaining
Schur complement matrix can quickly grow dense.
To rememdy this, we develop a method for computing sparse
approximations of the Schur complement, without ever having to write
down the dense biclique.
This throws up several difficulties: The sparse approximation must
accumulate very little error over many iterations and it must preserve
Eulerianness. See
\autoref{sub:main_ideas} for an overview of the challenges and the
techniques we use to resolve these difficulties.

The main algorithm for LU factorization,
\autoref{alg:multiPhase}, has $p_{\max}$ phases or iterations. The
routine for a single phase, given in \autoref{alg:singlePhase},
iteratively eliminates vertices belonging to a random set selected by
\autoref{alg:rcdd}.
We need to perform the eliminations in phases because eliminating
from the sets selected by \autoref{alg:rcdd} helps us constrain the
growth in the norm of certain matrices, which is necessary for
controlling overall error accumulation.

In every iteration within any phase of the
algorithm, a vertex, say $v,$ is eliminated.
We can afford to store the row and column that this creates in our LU
factorization, but we cannot afford to compute the dense biclique that
is created in the graph corresponding to the Schur complement.

Since the biclique directed Laplacian
 $\vstar{\mvar{L}}{v} 
-
\frac{1}{\mlap_{vv}} \mlap_{Vv} \mlap_{vV}$
 is dense in
general, we use the routine $\textsc{SingleVertexElim}(\cdot)$,
given in \autoref{sec:SingleVertexElimination}, to sparsify it during
every elimination, while ensuring we always output an Eulerian Laplacian
as our overall sparse approximation of $\sc{ \mvar{L}}{ V \backslash
  \{v \} }$.
To increase the accuracy of our approximation, we average the result
of $\Otil(1)$ calls to
$\textsc{SingleVertexElim}(\cdot)$ at every elimination step, and this
results in a slow increase in the the total number of edges at every
iteration. We will use $\textsc{SparsifyEulerian}(\cdot)$ from
\cite{CohenKPPRSV16} to sparsify the current graph every once in a
while to keep the total edge count low enough. This is a routine to sparsify any
Eulerian graph,  and can be found in  \cite{CohenKPPRSV16}. Recalling
the notion of asymmetric matrix approximation from 
\autoref{sec:preliminaries}, the
guarantees stated in Theorem~3.16 of \cite{CohenKPPRSV16} can be
stated as follows:

\begin{theorem}[Eulerian Spectral Sparsification - Theorem~3.16 of \cite{CohenKPPRSV16} Rephrased]
\label{thm:order_n_sparsifier} For Eulerian Laplacian $\mlap\in\R^{n\times n}$
and $\epsilon,\wprob,\in(0,1)$ with probability at least $1 - \wprob$ the $\textsc{SparsifyEulerian}(\mlap, \wprob, \epsilon)$ computes 
in $\otilde(\nnz(\mlap)+n\epsilon^{-2}\log(1/\wprob))$
time an Eulerian Laplacian $\widetilde{\mlap}\in\R^{n\times n}$ such that
\begin{enumerate}
\item $\widetilde{\mlap}$ is an
$\epsilon$-asymmetric spectral approximation 
of $\mlap$.
\item $\widetilde{\mlap}$ has $\Otil(n\epsilon^{-2}\log(1/\wprob)) $ non-zeros.
\item the weighted in and out degrees of the graphs associated with $\mlap$ and $\widetilde{\mlap}$ are identical.
\end{enumerate}
\end{theorem}

Let the vertices be labeled in the order in which we eliminate them so
that in phase $p$ of the algorithm, we eliminate vertices $(i_{p}+1) \ldots
i_{p + 1}$. The phases are numbered from $0$ to $p_{\max}-1$, while
vertices are numbered from $1$ to $n$.
Note, this implies that at the start of phase $p$, we
have a graph on $n - i_p$ vertices, starting initially with $i_0 = 0$,
before any eliminations have taken place.
We then index the elimination steps using the superscript
$^{(i)}$ to denote the state of the algorithm
\textbf{just after} we make
	the $i\textsuperscript{th}$ elimination step.

We denote the intermediate matrices produced by our approximate elimination
steps using $\mvar{S}^{(i)}$.
Each $\mvar{S}^{(i)}$ as an $n \times n$ matrix that's non-zero only on
entries that correspond to pairs of un-eliminated variables,
specifically in the block $[i + 1, n] \times [i + 1, n]$.
It is a sparse approximation of the Schur
complement of $\mlap$ onto the remaining variables,
with errors coming from two sources:
\begin{enumerate}
\item The randomized single vertex elimination procedure,
$\textsc{SingleVertexElim}(\cdot)$.
\item The global sparsification procedure
$\textsc{SparsifyEulerian}(\cdot)$.
\end{enumerate}
We then use $\mvar{L}^{(i)}$ to denote the original matrix perturbed by
the perturbations introduced by our elimination steps:
\[
\mvar{L}^{(i + 1)} \defeq \mvar{L}^{(i)}
+ \left( \mvar{S}^{(i+1)} - \sc{\mvar{L}^{(i)}}{\left[i+1, n\right] }\right).
\]

\rp{I thought we introduce $\mvar{S}^{(i)}$ in the single
phase algorithm, Algorithm~\ref{alg:singlePhase}.}
The quantity $\mvar{S}^{(i)}$ is formally defined in
\autoref{alg:multiPhase}, along with a set of indices $i_j$
(sometimes denoted $i_p$) for $j=0\ldots p_{\max}-1$,
where $p_{\max}$ is the total number of phases.
Recall the index $i_j$ is used to denote the state after the $i_j$th
elimination step, and the index of the initial state is $i_0 = 0$.

\begin{algorithm}	
\caption{$\textsc{SinglePhase}(\mvar{L},\wprob,\epsilon)$
\todolow{I'm using notation
    that's good for martingales, but has a too-dense matrix at
    intermediate steps, which one never actually needs to
    compute (although because we work with very sparse matrices, it
    wouldn't kill us). Still, better to clarify the right way to
    implement this somewhere. }}
\label{alg:singlePhase}
\SetAlgoVlined

\KwIn{an Eulerian Laplacian $\mvar{L}$  on vertex set $V$, and an error parameter $\epsilon$}
\KwOut{
Set of vertices eliminated, $F$,
$\mvar{S}^{\abs{F}}$ an Eulerian Laplacian on $V \setminus F$,
matrices $\matup^{(\abs{F})}$, $\matlow^{(\abs{F})}$ with
non-zeros only in the rows/columns corresponding to $F$ respectively
that are upper/lower triangle upon rearranging the vertices in $F$,
and $\mvar{L} \approx \mvar{S}^{(\abs{F})} + \matlow^{(\abs{F})} \matup^{(\abs{F})}$
\rp{I imagine we can eventually get rid of the super/sub scripts on the factorizations,
	but let's keep them here for now...}
}

$P \gets \Theta(\log^2(1/\delta)/\epsilon^2)$

$\epsilon' \gets \Theta(\epsilon/P)$ \;

Compute a sparsifier of $\mvar{L}$,
$\mvar{S}^{(0)} \gets \textsc{SparsifyEulerian}(\mvar{L} , \wprob/P,
\epsilon') $ \;

$T \gets \Otil(n\epsilon^{-6}\log^5(1/\wprob))  $  (the upper bound on
$\nnz(\mvar{S}^{(0)})$ from \autoref{thm:order_n_sparsifier})\;

Pick a $0.1$-RCDD subset (See \autoref{alg:rcdd} and \autoref{sec:rcdd-analysis}) of vertices $F^{(0)}$ from $\mvar{S}^{(0)}$ \;		

Initialize $\matup^{(0)}, \matlow^{(0)} \leftarrow 0$\;


Set $k_{\max} \leftarrow  \lfloor |F^{(0)}| / 2 \rfloor$\;

\For{$k = 1 \ldots  k_{\max}$ }  	{ 
 
 Among vertices in $F^{(k - 1)}$ with degree at most twice
 the average, pick a random $v_{k}$.
 \todolow{discuss how to do the DS for this.} \;
 $F^{(k)} \gets  F^{(k - 1)} \setminus \setof{ v_{k} }$\;

Set $d^{(k)} \leftarrow \mvar{S}^{(k - 1)}(v_{k} , v_{k} )$\;

Update the factorization:
set $\matup^{(k)}$ to $\matup^{(k - 1)}$ with row $v_{k} $ replaced
by $\frac{1}{d^{(k)}} \mvar{S}^{(k - 1)}(v_{k} , :)$ and 
 $\matlow^{(k)}$ to $\matlow^{(k - 1)}$ with column $v_{k} $ replaced
 by $\mvar{S}^{(k - 1)}(v_{k} , :)$.
 \rp{I imagine we eventually can get rid of the superscripts on the
 	LU factorization. But let's keep them for now as it allows us to
 	be maximally precise, and it's easy to modify this code to get rid of them.}
 \;
 
Set $\ll^{(k)}$, $\rr^{(k)}$ to length $n$ vectors containing
the the off-diagonal non-zeros in the column and row of $v_k$ in
$\mvar{S}^{(k - 1)}$ respectively.
\rp{This means that single vertex elim needs to take sparse vectors
	and manipulate them.}\;

Initialize the first matrix of the inner loop to be
the exact Schur complement of pivoting out $v_{k} $ from $\mvar{S}^{(k - 1)}$:
$\mvar{S}^{(k, 0)}
\gets \mvar{S}^{(k - 1)} - \frac{1}{d^{(k)}} \ll^{(k)} \rr^{(k) \top}$\;

\For{$t = 1 \ldots P$}{
$
\mvar{S}^{\left(k, t \right)} \leftarrow 	\mvar{S}^{\left( k, t-1 \right)}
-\frac{1}{P} \left(
\textsc{SingleVertexElim}\left(d^{(k)}, \ll^{(k)},\rr^{(k)} \right)
- \frac{1}{d^{(k)}} \ll^{(k)} \rr^{(k) \top}\right),
$ (see \autoref{alg:selim})
}      

\eIf{ $\nnz(\mvar{S}^{\left(k, P \right)}) \geq
  2 T$}{
$\mvar{S}^{(k)} \gets \textsc{SparsifyEulerian}(\mvar{S}^{\left(k, P \right)} , \wprob/P, \eps') $
}{
$\mvar{S}^{(k)} \gets  \mvar{S}^{\left(k, P \right)}$
}
}
Return $\mvar{S}^{(k_{\max})}$, $\matup^{(k_{\max})}$, $\matlow^{(k_{\max})}$, and $k_{\max}$ 
\end{algorithm}

\begin{algorithm}								
\caption{$\eulLU(\mvar{L},\wprob,\epsilon)$}
\label{alg:multiPhase}
\SetAlgoVlined
					
\KwIn{an Eulerian Laplacian $\mvar{L}$ and error parameter $0 < \epsilon < 1/2$}
\KwOut{lower and upper triangular matrices $\matlow, \matup$ whose product approximates $\mvar{L}$}
$\mvar{L} \gets \textsc{SparsifyEulerian}(\mvar{L} , \wprob/2, O(\epsilon/\log n)) $ \\
$\mvar{S}^{(0)} \gets \mvar{L}$ and set $\matlow, \matup$ to be empty matrices.\\
$j \gets 0$, $i_j \gets 0$\\
\While{$i_j < n$ \quad (i.e., for $p_{\max} = O(\log n)$ iterations)
  \quad}{
$\pr{ \mvar{T}, \matlow',\matup', k_{\max} } \gets
\textsc{SinglePhase}\left(\mvar{S}^{(i_j)},O(\frac{\wprob}{n}),O(\epsilon/\log n)\right)$ \\
$j \gets j+1$\\
$i_j \gets i_{j-1} + k_{\max},$ $\mvar{S}^{(i_j)} \gets \mvar{T}$ \\
Insert the nonzero vectors from the partial LU factorization $\matlow', \matup'$ into their corresponding locations $\matlow$ and $\matup$, respectively.\\
} 
\KwRet{$\matlow, \matup$}
\end{algorithm}


Using a matrix martingale concentration inequality, we prove the following statement about the distortion bounded in each phase. (See \autoref{sec:singlePhase} for the proof.)
\begin{theorem}
\label{thm:InPhaseErrorAccumulation}
Given an $n \times n$ Eulerian matrix $\mvar{L}$ with $m$ non-zeros and a $0.1$-RCDD subset $J$, and an
error parameter $\epsilon \leq 1/2$, and probability bound $\wprob < 1/n$, 
$\textsc{SinglePhase}(\mvar{L},\wprob,\epsilon)$
(\autoref{alg:singlePhase}) creates with probability $1-O(\wprob)$ matrices
$\widetilde{\mvar{S}}, \matlow', \matup'$,
where $\widetilde{\mvar{S}}$ is an Eulerian Laplacian, and $\matlow',
\matup'$ are upper and lower triangular respectively.
The algorithm also finds a subset $\widehat{J}$ such that 
$
\widetilde{\mvar{S}} = \sc{\Ltil}{V \setminus \widehat{J}}
$
for the matrix $\Ltil = \widetilde{\mvar{S}}  + \matlow'\matup' $
\todolow{clarify zero padding?}
such that $\widetilde{\mvar{L}}$ is an $\epsilon$-approximation of $\mvar{L}$ 
and $|\widehat{J}| \geq \frac{1}{2} |J|$.
Furthermore, the number of non-zeroes in $\widetilde{\mvar{S}}$, $ \matlow',$ and $\matup'$
is at most
$
\Otil( n \epsilon^{-6}\log^5(1/\wprob) )
$
and the runtime is at most $
\Otil(m 
+ n\epsilon^{-8}\log^7(1/\wprob) 
+ \epsilon^{-10} \log^{9}(1/\delta))
)
$
where  the $\otilde$ notation additionally hides $O(\log \log 1/\wprob)$ factors. 
\todolow{Add names for factors, fix/clarify SparsifierSize}
\end{theorem}

This means that with high probability we can bound the distortion within 	each phase by:
\[
\norm{\u{\mvar{S}^{\left(i_p\right)}}^{\pseudoRoot}
	\left(\mvar{L}^{\left(i_p\right)} - \mvar{L}^{\left(i_{p + 1}\right)} \right)
	\u{\mvar{S}^{\left(i_p\right)}}^{\pseudoRoot}}
\leq \theta_p \epsilon,
\]
while we also have $(n - i_{p + 1}) \leq 0.99 (n - i_p)$
%
%

We now define a PSD matrix which we use for measuring the
accumulation of errors:
\[
\mvar{F}^{(p)}
\defeq  \sum_{p' \leq p} \theta_{p'} \u{\mvar{S}^{\left(i_p \right)}} ~.
\]
For convenience, we denote $\mvar{F}^{(p_{\max})}$ by $\mvar{F}$:
\begin{equation}
\mvar{F}
\defeq \mvar{F}^{(p_{\max})}
= \sum_{0 \leq p < p_{\max}} \theta_p \u{\mvar{S}^{\left(i_p \right)}}.
\label{eq:F}
\end{equation}
We will set $\theta_p = \frac{1}{p_{\max}}$ so that $\sum_{p = 0}^{p_{\max}-1} \theta_p = 1.$
We will bound the final error in terms of $\mvar{F}$
as stated in the following lemma, which we state
for general matrices.
\begin{lemma}
	\label{lem:CumulativeErrorNew}
\rp{This is my attempt at a self-contained statement: sanity check?}
Consider a sequence of $n$-by-$n$ matrices
$\mvar{S}^{(0)}, \mvar{S}^{(1)}, \ldots, \mvar{S}^{(n)}$ such that
\begin{enumerate}
\item $\mvar{S}^{(0)}$ has non-zero entries only on the indices $[i + 1, n]$,
\item The left/right kernels of $\mvar{S}^{(i)}$ are equal, and the
  after restricting $\mvar{S}^{(i)}$ to the indices $[i + 1, n]$, the
  kernel of the resulting matrix equals the coordinate restriction
  of the vectors in the kernel of $\mvar{S}$. Formally, 
  $\ker(\mvar{S}^{(i)}_{[i + 1, n], [i + 1, n]}) = \setof{ \bb_{[i + 1,
      n]} : \bb \in \ker(\mvar{S}^{(0)})}$.
\item The undirectification of each $\mvar{S}^{(i)}$,
$\u{\mvar{S}^{(i)}} = \frac{1}{2} (\mvar{S}^{(i)} + (\mvar{S}^{(i)})^{\top})$
is positive semi-definite.
\end{enumerate}
Let $\mvar{M} = \mvar{M}^{(0)} = \mvar{S}^{(0)}$, and
define matrices $\mvar{M}^{(1)}, \mvar{M}^{(2)}, \ldots, \mvar{M}^{(n)}$
iteratively by
\[
\mvar{M}^{\left( i + 1 \right)}
\defeq
\mvar{M}^{\left( i \right)}
+
\left( \mvar{S}^{\left( i + 1 \right)}
- \sc{\mvar{M}^{\left( i \right)}}{\left[i + 1, n\right]} \right)
\qquad
\forall~0 \leq i < n.
\]
If for a subsequence of indices
$1 = i_0 < i_1 < i_2 < \ldots < i_{p_{\max}} = n$
associated scaling parameters
$0 < \theta_{0}, \theta_{1}, \ldots , \theta_{p_{\max} - 1} < 1/2$,
and some global error $0 < \epsilon < 1/2$,
we have for every $0 \leq p < p_{\max}$:
\[
\norm{\u{\mvar{S}^{\left(i_p\right)}}^{\pseudoRoot}
	\left(\mvar{M}^{\left(i_p\right)} - \mvar{M}^{\left(i_{p + 1}\right)} \right)
	\u{\mvar{S}^{\left(i_p\right)}}^{\pseudoRoot}}
\leq
\theta_{p} \epsilon,
\]
then for a matrix-norm defined from the symmetrization of the
$\mvar{S}^{(i_p)}$ matrices and the scaling parameters:
\[
\mvar{F}
=
\sum_{0 \leq p < p_{\max}}
\theta_{p} \mvar{U}_{\mvar{S}^{\left(i_{p}\right)}},
\]
we have:
\begin{enumerate}
\item \label{part:CumulativeError}
for each $0 \leq i \leq n$,
\[
\norm{\mvar{F}^{\pseudoRoot}
	\left(\mvar{M} - \mvar{M}^{\left(i\right)} \right)
	\mvar{F}^{\pseudoRoot}}_2
\leq \epsilon,
\]
\item \label{part:FNormBounds}
The final matrix $\MM^{(n)}$ satisfies
\[
\mvar{M}^{\left(n \right) \top}
\mvar{F}^{\dag}
\mvar{M}^{\left(n \right)}
\succeq \frac{1}{10 p^2} \cdot \mvar{F}.
\]
\end{enumerate}
\end{lemma}



%

We will provel \autoref{lem:CumulativeErrorNew} in
\autoref{sec:ErrorAccumulation}, after first bounding the errors
within each phase.
The bounds on the new norm from Part~\ref{part:FNormBounds}
of \autoref{lem:CumulativeErrorNew} and the ability to solve linear
systems in $\mvar{L}^{(n)}$
enable us to solve linear systems in $\mvar{L}$.
To formalize this, we need to draw upon the definition of approximate pseudoinverses
from~\cite{CohenKPPRSV16}.

\begin{definition}[Approximate Pseudoinverse]
\label{defn:approxInv}
Matrix $\mvar{Z}$ is an \emph{$\epsilon$-approximate pseudoinverse of
	matrix $\mvar{M}$ with respect to a symmetric positive semidefinite matrix $\mvar{F}$}, if
$\ker(\mvar{F}) \subseteq \ker(\mvar{M})=\ker(\mvar{M}^\intercal)=\ker(\mvar{Z})=\ker(\mvar{Z}^\intercal)$, and 
$
\normInline{ \mvar{I}_{\text{im}(\mvar{M})} - \mvar{Z} \mvar{M} }_{\mvar{F} \rightarrow \mvar{F}} \leq \epsilon
$.\footnote{Note that the ordering of $\mvar{Z}$ and $\mvar{M}$ is crucial: this definition is not equivalent to $\norm{\mvar{I}_{\im{\mvar{M}}} - \mvar{M}\mvar{Z}}_{\mvar{F} \rightarrow \mvar{F}}$
being small.}
\end{definition}

The reason why approximate pseudoinverses are useful is that if one preconditions with a solver for an approximate pseudoinverse, one can quickly solve the original system.

{
\newcommand*{\normFull}[1]{\norm{#1}}
\renewcommand*{\mU}{\mvar{F}}
\begin{lemma}[Preconditioned Richardson, \cite{CohenKPPRSV16} Lemma~4.2, pg. 30]
\label{lem:precond_richardson}
Let $\bb \in \R^{n}$ and $\MM, \ZZ,  \mU \in \R^{n \times n}$ such that
$\mU$ is symmetric positive semidefinite, $\ker(\mU) \subseteq
\ker(\MM)=\ker(\MM^\intercal) = \ker(\ZZ)=\ker(\ZZ^\intercal)$, and $b
\in \im{(\mm)}$. Then if one performs $t \geq 0$ iterative refinement
steps with step size $\eta > 0$, one obtains a vector $\xx_t =\textsc{PreconRichardson}(\MM, \ZZ, \bb, \eta, t)$ such that 
$$\norm{\xx_t - \mm^{\dag} \bb}_{\mU} \leq \norm{\mI_{\im{(\MM)}} -\eta \ZZ \MM}_{\mU \rightarrow \mU}^{t} 
\norm{\mm^\dag \bb}_{\mU}\,{.}$$ 
Furthermore, preconditioned Richardson implements a linear operator, in the sense that $\xx_t = \ZZ_{t} \bb$, for some matrix $\ZZ_t$ only depending on $\ZZ$, $\MM$, $\eta$ and $t$.
\end{lemma}
}

We now argue that the properties of the approximate LU factorization produced by our algorithm imply that a solver for systems in it is an approximate pseudoinverse of the original Laplacian.
 
\begin{lemma}
	\label{lem:ApproxPinv}
	Suppose we are given matrices $\mvar{L},$ $\Ltil$
	and a positive semi-definite matrix $\mvar{F}$ such that $\ker(\mvar{F}) \subseteq \ker(\mvar{L})=\ker(\mvar{L}^\intercal)=\ker(\Ltil)=\ker(\Ltil^\intercal)$ and
	\begin{enumerate}
		\item \label{cond:OperatorError}
		$\norm{\mvar{F}^{\pseudoRoot} (\mvar{L} - \widetilde{\mvar{L}})
			\mvar{F}^{\pseudoRoot}}_2 \leq \epsilon$,
		\item \label{cond:NormLower}
		$\Ltil^T \mvar{F}^{\dag} \Ltil \succeq \gamma \mvar{F}$.
		\todolow{say something about ignoring roundoff errors} 
	\end{enumerate}
	Then $\Ltil^{\dag}$ is an $\sqrt{\epsilon^2 \gamma^{-1}}$-approximate
	pseudoinverse for $\mvar{L}$ w.r.t. the norm $\mvar{F}$.
\end{lemma}

\begin{proof}
With a slight abuse of notation for notational convenience we let $\mvar{I}$ denote$\mvar{I}_{\text{im}(\mvar{\dlap})} = \lapid$ throughout this proof. The condition we need to show
$
	\norm{\mvar{I} -  \Ltil^{\dag} \mvar{L} }_{\mvar{F} \rightarrow \mvar{F}}
	\leq \sqrt{\epsilon^2 \gamma^{-1}}
$
	is equivalent to
	\[
	\left(\mvar{I} - \Ltil^{\dag} \mvar{L}  \right)^{\top}
	\mvar{F}
	\left(\mvar{I} - \Ltil^{\dag} \mvar{L}  \right)
	\preceq \epsilon^2 \gamma^{-1} \mvar{F}.
	\]
	\rp{somehow expanding the difference terms further out is a bad idea :-(}
	By rearranging a factor of $\Ltil$ on the LHS, we get
	\begin{align}
	\left(\mvar{I} - \Ltil^{\dag} \mvar{L}  \right)^{\top}
	\mvar{F}
	\left(\mvar{I} - \Ltil^{\dag} \mvar{L}  \right)
	&=
	\left(\Ltil -  \mvar{L}  \right)^{\top}
	\Ltil^{\dag \top} \mvar{F} \Ltil^{\dag}
	\left(\Ltil -  \mvar{L}  \right)\\
	\label{eq:ineq1} & \preceq 
	\gamma^{-1}
	\left(\mvar{L}^{(n)} -  \mvar{L}  \right)^{\top}
	\mvar{F}^{\dag}
	\left(\mvar{L}^{(n)} -  \mvar{L}  \right),
	\end{align}
where in the last inequality we used
$ \Ltil^{\dag \top} \mvar{F} \Ltil^{\dag}
	\preceq \gamma^{-1} { \mvar{F}^{\dag}}. $
	Condition~\ref{cond:OperatorError}, i.e., 
	$\norm{\mvar{F}^{\pseudoRoot}
		\left(\mvar{L}^{(n)} -  \mvar{L}  \right)
		\mvar{F}^{\pseudoRoot}}_2 \leq \epsilon$ is equivalent to
	\begin{equation}
	\label{eq:ineq2}
	\left(\mvar{L}^{(n)} -  \mvar{L}  \right)^{\top}
	\mvar{F}^{\dag}
	\left(\mvar{L}^{(n)} -  \mvar{L}  \right)
	\preceq \epsilon^{2} \mvar{F}.
	\end{equation}
	Combining \autoref{eq:ineq1} and \autoref{eq:ineq2}, we get
$
	\left(\mvar{I} - \Ltil^{\dag} \mvar{L}  \right)^{\top}
	\mvar{F}
	\left(\mvar{I} - \Ltil^{\dag} \mvar{L}  \right)
	\preceq \epsilon^{2} \gamma^{-1} \mvar{F}.
$
\end{proof}

We would like for the error guarantees of our solver to be in terms of $\u{\mvar{L}}$. In order to provide such guarantees, we need to relate this matrix to $\mvar{F}.$ 

\begin{lemma}\label{lem:relate-u-and-f}
$\u{\mvar{L}} / O(\log(n)) \preceq \mvar{F} \preceq O(n^2 \log^5 n) \cdot \u{\mvar{L}}
$
\end{lemma}

\begin{proof}

Since $\mvar{F} = \sum_{p' \leq p} \theta_{p'} \u{\mvar{S}^{\left(i_p \right)}} $ and $\theta_{p'} = \frac{1}{O(\log n)}$, we have 
$
\mvar{F} \succeq \frac{1}{O(\log n)} \u{\mvar{L}}
$.
We have by Lemma~\ref{lem:CumulativeErrorNew}, Part~\ref{part:FNormBounds} that
with probability $1 - O(\delta)$,
\begin{align*}
\mvar{F} 
\preceq O(\log^2 n) \cdot \widetilde{\mvar{L}}^{T}  \mvar{F}^{\dag} \widetilde{\mvar{L}} 
\preceq O(\log^4 n) \cdot \mvar{L}^{T}  \mvar{F}^{\dag} \mvar{L} 
\preceq O(\log^5 n) \cdot \mvar{L}^{T}  \u{\mvar{L}}^{\dag} \mvar{L} 
\preceq  O(n^2 \log^5 n) \cdot \u{\mvar{L}},
\end{align*}
where we used \autoref{lem:approx-implies-lfl} and
Lemma~\ref{lem:CumulativeErrorNew} Part~\ref{part:CumulativeError}
for the second step and Lemma~{13} from \cite{cohen2016faster} pg. 19
for the last step.
\end{proof}

We have now stated the key theorems and lemmas needed to analyze
correctness.
With these tools, we can obtain the main theorem statement about finding sparse LU factorizations (\autoref{thm:eulLU}) as follows.

\begin{proof}[Proof of \autoref{thm:eulLU}]

It is clear from the statement of the algorithm and the guarantees of
\autoref{thm:InPhaseErrorAccumulation} \sidford{I changed the theorem slightly to have a low order term. Can we either change the algorithm to check if $\epsilon < n^{-1/2}$ and if it is run Gaussian elimination on the sparsifier or just make the overall bound $\epsilon^{-10}$? I don't have a strong preference.} that \eulLU
(\autoref{alg:multiPhase}) outputs an LU factorization with the
sparsity claimed and with the claimed bound on running time and error probability.
The remaining correctness guarantees were proven as
\autoref{lem:relate-u-and-f}, and both parts of
\autoref{lem:CumulativeErrorNew}, respectively.
\end{proof}

We now have all the tools we need to obtain a fast solver for strongly connected Eulerian Laplacian systems.

\begin{proof}[Proof of \autoref{cor:solver}]
Suppose we have an Eulerian Laplacian $\mvar{L}$ and find a $1 / O(\log^2(n))$-approximate LU factorization in nearly-linear time using \autoref{thm:eulLU} in the sense that $\| \mApxNorm^{\pseudoOp / 2} (\mlap - \matlow \matup) \mApxNorm^{\pseudoOp / 2} \|_2 \leq 1/O(\log^2 n)$. Because it is an LU factorization, we can solve systems in it in linear time. By \autoref{lem:ApproxPinv}, such a solver is an $0.1$-approximate pseudoinverse of $\mvar{L}$ with respect to $\mvar{F}$, provided we pick an appropriately small constant in the error guarantee we invoke our LU factorization algorithm \eulLU (\autoref{alg:multiPhase}) with. By \autoref{lem:precond_richardson}, if we precondition the original system with this solver, we can find a solution $x$ to the original system with $\epsilon/\text{poly}(n)$ error in the sense that $\norm{x - \mlap^{\pseudoOp} b}_\mvar{F} \leq \frac{\epsilon}{\text{poly}(n)} \cdot \norm{\mlap^\pseudoOp b}_\mvar{F}$ in nearly-linear time. Since $\mvar{F} \approx_{\text{poly}(n)} \u{\mvar{L}}$, this implies $\norm{x - \mlap^{\pseudoOp} b}_{\u{\mvar{L}}} \leq \epsilon \cdot \norm{\mlap^\pseudoOp b}_{\u{\mvar{L}}}$. \todolow{consider adding more details, especially about how we do Richardson}
\end{proof}




\newcommand{\SingleVertexElim}{\textsc{SingleVertexElim}}
\newcommand{\ulocal}{\u{local}}

\section{Unbiased Degree Preserving Vertex Elimination}
\label{sec:SingleVertexElimination}

In this section we provide and analyze $\SingleVertexElim$, see
\autoref{alg:selim}, which produces a sparse approximation of the clique
created by Gaussian Elimination on an
Eulerian directed Laplacian.
It can be implemented to run in $O(\deg(v) \log \deg(v))$ time where $\deg$ is the
combinatorial degree of the vertex $v$ being eliminated, i.e. the number of vertices incident to it.

The algorithm has three key features. When including self-loops, it
preserves the weighted in and out degree of each vertex.
This ensures the graph created by replacing the clique with
the sparse approximation is still Eulerian. Note that we may get
self-loops which will cancel out and change the degree of vertices,
but it won't change the fact that each vertex still has in-degree
equal to out-degree. 
Secondly, it produces a \emph{sparse} approximation of the biclique
created by elimination. Thirdly, it achieves these guarantees while being
an unbiased estimate of the biclique.

\begin{algorithm}								
\caption{$\SingleVertexElim(\ll,\rr)$}
\label{alg:selim}
\SetAlgoVlined
					
\KwIn{$\ll, \rr \in \rea^{n}_{\geq 0} $ such that 
$\vecone^{\top} \ll = \vecone^{\top} \rr$.  }
\KwOut{$\mdir \in \rea^{n \times n}$ \sidford{add english description of what this is as we do for other algorithms?}}

$s \leftarrow \vecone^{\top} \ll$ \;

\uIf {$ s = 0$} { 
	\KwRet{ $\mvar{\zero}$ } \;
}
\uElseIf{ $\min(\ll) \leq \min(\rr)$ \sidford{Where do we clarify that this is the smallest non-zero entry?}}
{
	$i \leftarrow \argmin(\ll)$\;
	Pick index $j \leftarrow k$ with probability $\rr(k) / s$\;
	\KwRet{ $\ll(i) \vecind_{i} \vecind_{j}^{\top}  
		+
		\textsc{SingleVertexElim}(\ll - \ll(i) \vecind_{i} ,\rr -\ll(i) \vecind_{j})$
  }\;
}
\uElse
{
$i \leftarrow \argmin(\rr)$\;
Pick index $j \leftarrow k$ with probability $\ll(k) / s$\;
 \KwRet{
	 $\rr(i) \vecind_{j} \vecind_{i}^{\top}  
   +
   \textsc{SingleVertexElim}(\ll - \rr(i) \vecind_{j} ,\rr -\rr(i) \vecind_{i})$
  }\;
}
\end{algorithm}

\begin{lemma}
The matrix $\ma$ returned by 
$\textsc{SingleVertexElim}(\ll,\rr)$
has $\nnz(\ma) \leq \nnz(\ll) + \nnz(\rr)$. Furthermore, the algorithm makes at most $\nnz(\ll) + \nnz(\rr)$
recursive calls to itself each of which is made to a 
vectors $\ll$ and $\rr$ with non-negative entries satisfying $\vecone^{\top} \ll = \vecone^{\top} \rr$. 
\end{lemma}

\begin{proof}
We prove this by induction on $\nnz(\ll) + \nnz(\rr)$.
Base case: $\nnz(\ll) + \nnz(\rr) = 0$, then $\ll,\rr = \vec{\zero}$, so
$\ma = \mvar{\zero}$, and $\nnz(\ma)=0$.
This proves the base case.
For the inductive step, we suppose $\nnz(\ll) + \nnz(\rr) = k+1$ and
that the lemma holds whenever $\nnz(\ll) + \nnz(\rr) \leq k$.
Without loss of generality consider the case of $\min(\ll) \leq \min(\rr)$.
Note that $\ll - \ll(i) \vecind_{i} ,\rr -\ll(i) \vecind_{j}\geq
\vec{\zero}$, and that
\[ 
\nnz(\ll - \ll(i) \vecind_{i}) + \nnz(\rr -\ll(i) \vecind_{j}) 
\leq
k
\]
so by the induction hypothesis with $\ma' = \textsc{SingleVertexElim}(\ll - \ll(i)
\vecind_{i} ,\rr -\ll(i) \vecind_{j})$ we have $\nnz(\ma') \leq k$, and so $\nnz(\ma) \leq
k + 1$. This proves the lemma by induction.
The number of recursive calls can be bounded in the same way.
\end{proof}

\begin{lemma}
\label{clm:se:spaces}
The matrix $\ma$ returned by 
$\textsc{SingleVertexElim}(\ll,\rr)$
has only non-negative entries and
satisfies
$\ma \vecone = \ll$, and $\vecone^{\top} \ma = \rr^{\top}$.
\end{lemma}

\begin{proof}
We prove the lemma by induction on $\nnz(\ll) + \nnz(\rr)$. 
It is true in the base case $\nnz(\ll) + \nnz(\rr) = 0$,
where $\ll,\rr = \vec{\zero}$, so
$\ma = \mvar{\zero}$, and $\ma \vecone = \vec{\zero} = \ll$, 
and $\vecone^{\top} \ma = \vec{\zero}^{\top} = \rr^{\top}$.

For the inductive step, we suppose $\nnz(\ll) + \nnz(\rr) = k+1$ and
that the lemma holds whenever $\nnz(\ll) + \nnz(\rr) \leq k$.
W.l.o.g. consider the case of $\min(\ll) \leq \min(\rr)$.

Let
$\ma' = \textsc{SingleVertexElim}(\ll - \ll(i) \vecind_{i} ,\rr -\ll(i)
\vecind_{j})$.
By the induction hypothesis,
\[
\ma  \vecone
=  \ll(i) \vecind_{i} \vecind_{j}^{\top} \vecone 
+  \ma' \vecone = \ll(i) \vecind_{i}  + \ll - \ll(i) \vecind_{i} = \ll
.
\]
and similarly
\[
\vecone^{\top} \ma  
=  \vecone^{\top} \ll(i) \vecind_{i} \vecind_{j}^{\top}
+ 
\vecone^{\top}  \ma'
 = 
\ll(i) \vecind_{j}^{\top}
 + \rr^{\top} - \ll(i) \vecind_{j}^{\top} = \rr
.
\]
\end{proof}

\begin{lemma}
\label{clm:se:expec}
Given $\ll, \rr \in \rea^{n} $ s.t. both have non-negative entries and
$\vecone^{\top} \ll = \vecone^{\top} \rr = s$,
let $\ma = \textsc{SingleVertexElim}(\ll,\rr)$.
Then $\expec{}{\ma} = \ll \rr^{\top} / s$.
\end{lemma}

\begin{proof}
 We prove the lemma by induction on $\nnz(\ll) + \nnz(\rr)$. 
It is true in the base case $\nnz(\ll) + \nnz(\rr) = 0$,
where $\ll,\rr = \vec{\zero}$, so
$\ma = \mvar{\zero}$, so $\expec{}{\ma} = \mvar{\zero}$.
For the inductive step, we suppose $\nnz(\ll) + \nnz(\rr) = k+1$ and
that the lemma holds whenever $\nnz(\ll) + \nnz(\rr) \leq k$.
Without loss of generality consider the case of $\min(\ll) \leq \min(\rr)$. In this case we have
\begin{align*}
\expec{}{\ma}
& =
\sum_{j} \frac{\rr(j)}{s}
\left(
\ll(i) \vecind_{i} \vecind_{j}^{\top}  
+
\frac{1}{s-\ll(i)}
(\ll - \ll(i) \vecind_{i} ) (\rr -\ll(i) \vecind_{j})^{\top}
\right)
\\
& =
\sum_{j} 
\frac{\ll(i)}{s}
\vecind_{i} \rr(j) \vecind_{j}^{\top}  
+
\frac{\rr(j)}{s}
\frac{1}{s-\ll(i)}
(\ll - \ll(i) \vecind_{i} ) \rr^{\top}
-
\frac{\ll(i)}{s}
\frac{1}{s-\ll(i)}
(\ll - \ll(i) \vecind_{i} ) \rr(j) \vecind_{j}^{\top}
\\
& =
\frac{\ll(i)}{s}
\vecind_{i} \rr^{\top}  
+
\frac{1}{s-\ll(i)}
(\ll - \ll(i) \vecind_{i} ) \rr^{\top}
-
\frac{\ll(i)}{s}
\frac{1}{s-\ll(i)}
(\ll - \ll(i) \vecind_{i} ) \rr^{\top}
\\
& =
\ll(i)\vecind_{i} \rr^{\top}  
\left(
\frac{1}{s}
-
\frac{1}{s-\ll(i)}
+
\frac{\ll(i)}{s}
\frac{1}{s-\ll(i)}
\right)
+
\ll\rr^{\top}
\frac{1}{s-\ll(i)}
\left(
1
-
\frac{\ll(i)}{s}
\right)
\\
&=
\ll \rr^{\top} / s
.
\end{align*}
\end{proof}
A crucial matrix used in analyzing the elimination of a single
vertex is the Schur complement of the star incident on the
eliminated vertex in the undirectification of the whole matrix.
\begin{definition}
	\label{def:local}
        Given an Eulerian Laplacian $\mlap$ and a vertex $v$, let 
        $$\u{local} = \vstar{\u{\mlap}}{v}  - \frac{1}{\u{\mlap}(v, v)}\u{\mlap}(:, v) \u{\mlap}(v, :).$$
\end{definition}
Note that
\[
\u{local} \preceq \vstar{\u{\mlap}}{v}.
\]
And hence for a random choice of vertex $v$ in a graph with $n$
remaining vertices 
\begin{equation}
\expec{v}{\u{local}} \preceq\frac{2}{n} \u{\mvar{L}}
\label{eq:exptlocalglobal}
\end{equation}

\begin{lemma}[Single Vertex Elimination Routine]
	\label{lem:SElim}
	There is a routine \textsc{SingleVertexElim} that takes
        the in and out adjacency list vectors $\ll$ and $\rr$
        of a vertex $u$ in an Eulerian Laplacian
	with $d$ non-zeros,
	and produces a matrix $\ma$ with at most $d$ 
	non-zeros such that the error matrix
	\[
	\mvar{X} = \frac{1}{\rr^{\trp}\vecone} \ll\rr^{\trp}  - \ma
	\]
	satisfies
\begin{enumerate}
\item \label{part:SElimNullSpace}
$\mvar{X} \vecone = 0$, $\mvar{X}^{\top} \vecone = 0$, and
\item \label{part:SElimExpectation}
$\expec{}{\mvar{X}} = 0$, and
\item \label{part:SElimError}
For the local undirectification,
$\mvar{U}_{local}$ as given in \autoref{def:local}, we have
$
\normInline{\mvar{U}_{local}^{\pseudoRoot} \mvar{X} \mvar{U}_{local}^{\pseudoRoot}}_2 \leq 4 ~.
$
	\end{enumerate}
\end{lemma}
\begin{proof}
\todolow{Not sure if refactoring will move this around.}
We already have Parts \ref{part:SElimNullSpace} and \ref{part:SElimExpectation} from \autoref{clm:se:spaces} and \autoref{clm:se:expec}.

For Part~\ref{part:SElimError}, we note that by $\ma$ (by \autoref{clm:se:spaces}) has the sum of
the absolute value of its entries in the $i$th row and $i$th column of at most $\ll(i)+\rr(i)$.
This similarly applies to the the expectation $\frac{1}{\rr^{\trp}\vecone} \ll\rr^{\trp}$.  Thus,
the sums for the error matrix $\mvar{X}$ are at most double this: $2(\ll(i)+\rr(i))$.

Now, we define a diagonal matrix $\mvar{D}_{local}$, whose $i$th
diagonal entry is $\frac{\ll(i)+\rr(i)}{2}$. Because the sums of the absolute values of the $i$th row and column are at most $4 (\mvar{D}_{local})_{ii}$, we have 
\[
\| \mvar{D}_{local}^{\pseudoRoot} \mvar{X} \mvar{D}_{local}^{\pseudoRoot} \| \leq 4.
\]
Now, $\| \mvar{U}_{local}^{\dag 1/2} \mvar{D}_{local}^{1/2} \|_2 \leq 1$ by \autoref{lem:dlocalulocal} (proved in
\autoref{sec:mat_facts},  $\mvar{U}_{local}^{\dag} \preceq \mvar{D}_{local}^{-1} $) and consequently
\[
\| \mvar{U}_{local}^{\pseudoRoot} \mvar{X} \mvar{U}_{local}^{\pseudoRoot} \|_2 = \| (\mvar{U}_{local}^{\dag 1/2} \mvar{D}_{local}^{1/2}) (\mvar{D}_{local}^{\pseudoRoot} \mvar{X} \mvar{D}_{local}^{\pseudoRoot}) (\mvar{U}_{local}^{\dag 1/2} \mvar{D}_{local}^{1/2})^{\trp} \|_2 \leq 4.
\]
\end{proof}

\section{Robustly Bounded Schur Complement Sets}
\label{sec:SchurBounds}

As we have discussed, one of the key difficulties in applying repeated vertex elimination to solve Eulerian Laplacian systems is that unlike with symmetric Laplacian systems the Schur complement of an Eulerian Laplacian may be much larger than that of the original Laplacian. Consequently, if we simply eliminate an arbitrary set of vertices the error we incur my too large to ensure we compute an effective preconditioner. 

To circumvent this we eliminate vertices in phases where in each phase we only eliminate vertices that do not "blow up" the Schur complement, i.e. induce Schur complements that are not spectrally dominated by a small multiple of the current Laplacian. Formally, we only eliminate vertices from what we call \emph{robustly bounded schur complement sets} defined below. These are sets, where even under a small amount of spectral error, the Schur complement is not too much larger than the original graph.


\begin{definition}[Robustly Bounded Schur Complement Set]
\label{def:robustSchurBound}
Given an Eulerian Laplacian  $\mvar{L}$, a 
\emph{robustly bounded Schur complement set} vertex set $J$,
is a subset of the vertices of $\mvar{L}$, such that  
for any $\widehat{J} \subseteq J$ and any $\Ltil$ that $1/2$-approximates $\mvar{L}$
we have 
$\mU_{\sc{\Ltil}{\widehat{J}}}\preceq O(1) \mU_{\mvar{L}}$. 
\end{definition}

In this section, we formally show that $\alpha$-RCDD subsets of the vertices are robustly bounded Schur complement sets and therefore we can easily find such sets. The main result of this section is the following formalization of this claim.

\begin{lemma}[$\alpha$-RCDD Sets are Robustly Bounded]
\label{lem:RCDDRobustSchurBound}
Given an Eulerian Laplacian $\mvar{L}$,
for any fixed constant $\alpha$, 
an $\alpha$-RCDD subset $J$
of the vertices of $\mvar{L}$
is robustly bounded.
\end{lemma}

We prove this lemma in several pieces. First in \autoref{sec:schur_complemt_gen} we provide a general lemma about how much a Schur complement of an arbitrary asymmetric matrix can increase. Then in \autoref{sec:schur_complement_eulerian} we show how to apply this to Eulerian Laplacians to bound show how much a Schur complement of an RCDD subset of an Eulerian Laplacian can increase. Finally, in \autoref{sec:schur_stability} we show that this analysis is robust to asymmetric approximation and with this prove \autoref{lem:RCDDRobustSchurBound}.

\subsection{Bounding General Schur Complements}
\label{sec:schur_complemt_gen}

Here we provide a general lemma that upper bounds
the symmetrization of the Schur complement of a general matrix
spectrally via its symmetrization. This lemma is the main tool we use to reason about the Schur complement of subsets of a $\alpha$-RCDD subset of a Eulerian Laplacian.

\begin{lemma}
	\label{lem:gen_schur_complement} If $\mdir\in\R^{n\times n}$ satisfies
	$\mU\defeq\mU_{\mdir}\succeq\mzero$ and $F,C\subseteq[n]$ is a partition
	of $[n]$ where $\mU_{FF}\succ\mzero$ then $\mU_{\sc{\mdir}{F}}\preceq(1+\alpha)\mU$ provided the following condition holds
	\[
	\mm\defeq\left[\begin{array}{cc}
	[\mdir_{FF}^{-1}]^{\top}\mU_{FF}\mdir_{FF}^{-1} & \mzero_{FC}\\
	\mzero_{CF} & \mzero_{CC}
	\end{array}\right]\preceq\alpha[\mdir^{\dagger}]^{\top}\mU\mdir^{\dagger}
	~.
	\]
\end{lemma}

\begin{proof}
Let $\zz\in\R^{n}$ be arbitrary and let $\xx,\yy\in\R^{n}$ be defined
so that $\xx_{C}=\yy_{C}=\zz_{C}$, $\xx_{F}\defeq-\mdir_{FF}^{-1}\mdir_{FC}\zz_{C}$,
and $\yy_{F}=-\mU_{FF}^{-1}\mU_{FC}\zz_{C}$. (Note $\mdir_{FF}$ is invertible by \autoref{lem:psd_invert} as $\mU_{FF} \succ \mzero$.) Now, this definition was chosen so that
	\[
	\mdir \xx=\left[\begin{array}{cc}
	\mdir_{FF} & \mdir_{FC}\\
	\mdir_{CF} & \mdir_{CC}
	\end{array}\right]\left(\begin{array}{c}
	-\mdir_{FF}^{-1}\mdir_{FC}\zz_{C}\\
	\zz_{C}
	\end{array}\right)=\left(\begin{array}{c}
	\vec{0}_{F}\\
	\sc{\mdir}{F} \zz_{C}
	\end{array}\right)\,.
	\]
	and therefore $\xx^{\top} \mdir \xx=\zz^{\top}\sc{\mdir}{F}\zz$. Furthermore,
	note that
	\[
	\zz^{\top}\mU \zz=\zz_{F}^{\top}\mU_{FF}\zz_{F}+2\zz_{C}\mU_{CF}\zz_{F}+\zz_{C}^{\top}\mU_{CC}\zz_{C}
	\]
	and since $\mU_{FF}\succ0$ we have that $\zz^{\top}\mU \zz$ is minimized
	over $\zz_{F}$ when $\zz_{F}=-\mU_{FF}^{-1}\mU_{FC}\zz_{C}$ and thus $\yy^{\top}\mU \yy\leq \zz^{\top}\mU \zz$.
	Furthermore, 
	$
	\yy_{F}-\xx_{F} = \mdir_{FF}^{-1} [\mdir \yy]_{F}
	$, $\xx_{C}=\yy_{C}$, and \autoref{lem:complicated} yields
	\begin{align*}
	\norm{\xx-\yy}_{\mU}^{2} & =(\xx_{F}-\yy_{F})\mU_{FF}(\xx_{F}-\yy_{F})=\norm{[\mdir \yy]_{F}}_{\mdir_{FF}^{-1}\mU_{FF}\mdir_{FF}^{-1}}^{2}\,.\\
	& =\norm{\mdir \yy}_{\mm}^{2}\leq\alpha\norm{\mdir \yy}_{[\mdir^{\dagger}]^{\top}\mU\mdir^{\dagger}}^{2}\leq\alpha\norm \yy_{\mU}^{2}\,.
	\end{align*}
	Further, by the $\mU$-orthogonality of $\yy$ and $\xx-\yy$ (as $\xx-\yy$ is supported on $F$ and $\mU \yy$ is 0 on $F$),
	\[
	\norm \xx_{\mU}^{2}=\norm \yy_{\mU}^2+\norm{\xx-\yy}_{\mU} \leq (1+\alpha) \norm \yy_{\mU}^{2} ~.
	\]
	As $\norm \xx_{\mU}^{2}=\zz^{\top}\sc{\mdir}{F} \zz=\zz^{\top}\mU_{\sc{\mdir}{F}}\zz$
	and $\norm \yy_{\mU}^{2}\leq\norm \zz_{\mU}^{2}$ the result follows.
\end{proof}

\subsection{Schur Complements of Eulerian Laplacians}
\label{sec:schur_complement_eulerian}

Here we show how to apply the Schur complement bounds of the previous subsection to bound the increase in Schur complements for Eulerian Laplacians. In particular we bound the blowup of the Schur
complements as we pivot away $\alpha$-RCDD subsets
of vertices.The main result we prove is the following.

\begin{lemma}
\label{cor:rcdd_schur_bound}
  Suppose that $\mlap=\md-\ma^{\top}\in\R^{n\times n}$ is
an Eulerian Laplacian, and $F\subseteq [n]$ is an $\alpha$-RCDD
subset.
Then $\mU_{\sc{\mlap}{F}}\preceq(3+\frac{2}{\alpha} )\mU_{\mlap}$.
\end{lemma}

To prove this lemma, first we provide the following lemma, which is a self contained fact about Eulerian Laplacians that will allow us to leverage \autoref{lem:gen_schur_complement}.

\begin{lemma}
	\label{lem:small_complicated} Suppose that $\mlap=\md-\ma^{\top}\in\R^{n\times n}$
	is an Eulerian Laplacian associated with directed graph $G=(V,E,w)$. Then $\mlap^{\top}\md^{-1}\mlap\preceq2 \mU_{\mlap}$.
\end{lemma}
\begin{proof}
	Let $\xx\in\R^{n}$ be arbitrary and recall that
	\[
	[\mlap \xx]_{i}=\sum_{(i,j)\in E}w_{ij}(\xx(i)-\xx(j))\enspace\text{ and }\enspace\md_{ii}=\sum_{(i,j)\in E}w_{ij}\,.
	\]
	Furthermore, by Cauchy-Schwarz we have that 
	\[
	[\mlap \xx]_{i}^{2}=\left[\sum_{(i,j)\in E}w_{ij}(\xx(i)-\xx(j))\right]^{2}\leq\left[\sum_{(i,j)\in E}w_{ij}\right]\cdot\left[\sum_{(i,j)\in E}w_{ij}(\xx(i)-\xx(j))^{2}\right]\,.
	\]
	Consequently,
	\begin{align*}
	\xx^{\top}\mlap^{\top}\md^{-1}\mlap \xx & =\sum_{i\in[n]}\frac{[\mlap \xx]_{i}^{2}}{\md_{ii}}\leq\sum_{i\in[n]}\sum_{(i,j)\in E}w_{ij}(\xx(i)-\xx(j))^{2}=2\cdot \xx^{\top}\mU_{\mlap} \xx\,.
	\end{align*}
	
\end{proof}

Using \autoref{lem:gen_schur_complement} and \autoref{lem:small_complicated} we can prove the following, a key bound on the increase of Schur complements of Eulerian Laplacians.

\begin{lemma}
\label{lem:lap_schur_blowup}
Suppose that $\mlap=\md-\ma^{\top}\in\R^{n\times n}$ is
an Eulerian Laplacian,
let $\mU\defeq\mU_{\mlap}$ and let $F,C\subseteq[n]$
be a partition of $[n]$ such that 
$\mU_{FF}\succeq\frac{1}{\alpha}\md_{FF}$
then $\mU_{\sc{\mlap}{F}}\preceq(1+2\alpha)\mU$.
\end{lemma}

\begin{proof}[Proof of \autoref{lem:lap_schur_blowup}]
	First note that $\mlap_{FF}$ and $\mU_{FF}$ must be RCDD as $\mlap$
	is Eulerian, and since $\mU_{FF}\succeq\frac{1}{\alpha}\md_{FF}$ it
	is the case that $\mU_{FF}$ is invertible and so is $\mlap_{FF}$.
	Consequently, as $\mU_{FF}\preceq\mlap_{FF}^{\top}\mU_{FF}^{-1}\mlap_{FF}$
	from our general bounds on harmonic symmetrizations we have that
	$[\mlap_{FF}^{-1}]^{\top}\mU_{FF}[\mlap_{FF}^{-1}]\preceq\mU_{FF}^{-1}$.
	Furthermore, as $\mU_{FF}\succeq\frac{1}{\alpha}\md_{FF}$ we have
	that $\mU_{FF}^{-1}\preceq\alpha\md_{FF}^{-1}$ and therefore that
	\[
	\left[\begin{array}{cc}
	[\mlap_{FF}^{-1}]^{\top}\mU_{FF}\mlap_{FF}^{-1} & \mzero_{FC}\\
	\mzero_{CF} & \mzero_{CC}
	\end{array}\right]\preceq\left[\begin{array}{cc}
	\alpha\md_{FF}^{-1} & \mzero_{FC}\\
	\mzero_{CF} & \mzero_{CC}
	\end{array}\right]\preceq\alpha\md^{-1}\preceq\alpha \mlap^{\top}[\mlap^{\dagger}]^{\top}\md^{-1}[\mlap^{\dagger}] \mlap\,.
	\]
	As $[\mlap^{\dagger}]^{\top}\md^{-1}[\mlap^{\dagger}]\preceq2\cdot\mU$
	by \autoref{lem:small_complicated} the result then follows from
	\autoref{lem:gen_schur_complement}.
\end{proof}

Using \autoref{lem:lap_schur_blowup} we can now prove \autoref{cor:rcdd_schur_bound}.

\begin{proof}[Proof of \autoref{cor:rcdd_schur_bound}]
Since $F$ is an $\alpha$-RCDD subset $\frac{1}{1+\alpha} \md_{FF} -(\mU_{\ma})_{FF}$ is SDD and PSD. Consequently,  $\mU_{FF} = 
\frac{\alpha}{1+\alpha} \md_{FF}
+ \frac{1}{1+\alpha}
\md_{FF} -(\mU_{\ma})_{FF}  
\succeq 
\frac{\alpha}{1+\alpha} \md$ 
%
and the claim follows from \autoref{lem:lap_schur_blowup}.
\end{proof}

\subsection{Schur Complement Stability}
\label{sec:schur_stability}

Here, we show that Schur complements are robust to
small changes to the original matrix. Using this stability result \autoref{lem:schur_robust},  and the results of the previous subsection we prove the main claim of this section, \autoref{lem:RCDDRobustSchurBound}. 

\begin{lemma}
\label{lem:schur_robust}
Suppose $\mdir\in\R^{n\times n}$ satisfies $\mU\defeq\mU_{\mdir}\succeq\mzero$ and $\ker(\mdir) = \ker(\mdir^\top)$ and suppose $F,C\subseteq[n]$ is a partition of $[n]$ where $\mU_{FF}\succ\mzero$. If $\widetilde \mdir$ $\epsilon$-approximates $\mdir$ then $\mU_{\sc{\widetilde\mdir}{F}} \preceq \left ( \frac{1+\epsilon}{1-\epsilon} \right )^2 \mU_{\sc{\mdir}{F}}$.
\end{lemma}

\begin{proof}
As in the proof of \autoref{lem:gen_schur_complement}, given any
vector $\zz$ we define $\xx$ and $\xxtil$ such that $\xx_{C}=\xxtil_{C}=\zz_{C}$,
$\xx_{F}\defeq-\mdir_{FF}^{-1}\mdir_{FC}\zz_{C}$,
$\xxtil_{F} \defeq -\widetilde \mdir_{FF}^{-1} \widetilde \mdir_{FC}\zz_{C}$.
(Note $\mdir_{FF}$ is invertible by \autoref{lem:psd_invert} as $\mU_{FF} \succ \mzero$.)
 As in the proof of \autoref{lem:gen_schur_complement} this yields $\zz^T \mU_{\sc{\mdir}{F}} \zz = \xx^T \mU \xx$ and
$\zz^T \mU_{\sc{\widetilde \mdir}{F}} \zz = \xxtil^T \widetilde \mU \xxtil$. Furthermore, direct calculation reveals that
\begin{align*}
\xxtil_{F} - \xx_{F} = -\widetilde{\mdir}_{FF}^{-1} [\widetilde \mdir \xx]_{F} 
= -\widetilde \mdir_{FF}^{-1} [(\widetilde \mdir-\mdir) \xx]_{F} ~.
\end{align*}
By the assumption of $\widetilde \mdir$, we have
\begin{align*}
\norm{(\widetilde \mdir-\mdir) \xx}_{\widetilde \mU^\dagger} \leq \frac{1}{1-\epsilon} \norm{(\widetilde \mdir-\mdir) \xx}_{\mU^\dagger} 
\leq \frac{\epsilon}{1-\epsilon} \norm{\xx}_{\mU}.
\end{align*}
By \autoref{lem:complicated} and the fact that $[(\widetilde{N} - N)x]_{C} = \zero$ 
this gives 
\begin{align*}
\norm{\xxtil_{F}-\xx_{F}}_{{\widetilde \mU}_{FF}} = \norm{\widetilde \mdir_{FF}^{-1} [(\widetilde \mdir-\mdir) \xx]_{F}}_{{\widetilde \mU}_{FF}} 
\leq \norm{(\widetilde \mdir-\mdir) \xx}_{\widetilde \mU^\dagger} 
\leq \frac{\epsilon}{1-\epsilon} \norm{\xx}_{\mU} ~.
\end{align*}
Therefore we can write
\begin{align*}
\norm{\zz}_{\mU_{\sc{\widetilde \mdir}{F}}}^2 &= \norm{\xxtil}_{\widetilde \mU}^2 
\leq \left ( \norm{\xx}_{\widetilde \mU} + \norm{\xxtil-\xx}_{\widetilde \mU} \right )^2 
\leq \left ( (1+\epsilon) \norm{\xx}_{\mU} + \frac{\epsilon}{1-\epsilon} \norm{\xx}_{\mU} \right )^2 \\
&\leq \left ( \frac{1+\epsilon}{1-\epsilon} \right )^2 \norm{\xx}_{\mU}^2 
= \left ( \frac{1+\epsilon}{1-\epsilon} \right )^2 \norm{\zz}_{\mU_{\sc{\widetilde \mdir}{F}}}^2 ~.
\end{align*}
\end{proof}


\begin{proof}[Proof of \autoref{lem:RCDDRobustSchurBound}] By \autoref{cor:rcdd_schur_bound} and the fact that subsets of $\alpha$-RCDD subsets are $\alpha$-RCDD we have that 
$\mU_{\sc{\mvar{L}}{\widehat{J}}}\preceq O(1) \mU_{\mvar{L}}$.
and \autoref{lem:schur_robust}. The result then follows by applying \autoref{lem:schur_robust} since for Eulerian $\mlap$ we have $\ker(\mlap) = \ker(\mlap)^\top$.
\end{proof}


\section{Single Phase Analysis}
\label{sec:singlePhase}

In this section we prove \autoref{thm:InPhaseErrorAccumulation} which gives the guarantees of $\SinglePhase$ (\autoref{alg:singlePhase}), our main subroutine for eliminating blocks of vertices.

We start by analyzing the running time of $\SinglePhase$ in \autoref{alg:singlePhase}. The following \autoref{lem:singlePhaseRuntime} bounds both the number of times the entire graph is sparsified, i.e. $\textsc{SparsifyEulerian}$, as well as the total running time of the algorithm. To prove this lemma, we bound how much the sparsity of the graph increases as we eliminate vertices, i.e. call \textsc{SingleVertexElim} (\autoref{alg:selim}).

\begin{lemma}[$\SinglePhase$
(\autoref{alg:singlePhase}) Running Time]
\label{lem:singlePhaseRuntime}
Suppose $\dlap$ is an Eulerian Laplacian on $n$ vertices with $m$
non-zeros. In the \emph{for}-loop of $\SinglePhase$
(\autoref{alg:singlePhase}), $\textsc{SparsifyEulerian}$ is
called $O(P) = O(\eps^{-2}\log^2(1/\wprob))$ times. Consequently, the total running time for $\SinglePhase$ is
\[
\Otil(m + P (T + T P/n))
= \Otil(m + n \epsilon^{-8} \log^7(1/\delta) + \epsilon^{-10} \log^{9}(1/\delta)) ~.
\]
where here the $\otilde$ notation additionally hides $O(\log \log 1/\wprob)$ factors. 
\end{lemma}

\begin{proof}
Note that for all $k \in [0, k_{\max})$ the Eulerian Laplacian $\ms^{(k)}$ is non-zero on at most $n - k \geq n - k_{\max} \geq n/2$ vertices. Since in each iteration of the algorithm a vertex $v$ of at most twice the average degree and then adds $O(\deg(v)) = O(\nnz(\ms^{(k)}/n)$ edges are added to the
graph from each of the $P$ calls to \textsc{SingleVertexElim} and then sparsification only decreases sparsity we have that
\[
\nnz(\ms^{(k + 1)}) 
\leq \left(1 + O\left(\frac{P}{n}\right)\right) \nnz(\ms^{(k)})
\leq \nnz(\ms^{(k)}) \exp(O(P/n)) ~.
\]
Consequently, for all $t > 0$ we have that $\nnz(\ms^{(k + t)}) \leq \exp(O(Pt/n)) \nnz(\ms^{(k)})$. Consequently, it takes at least $t = \Omega(n/P)$ iterations for the sparsity of $\nnz(\ms^{(k)})$ to double and we have that $\SparsifyEulerian$ is called at most $O(k_{\max} / (n/P)) = O(P)$ times.  

With the exception of the first one, each invocations of $\SparsifyEulerian$ is on a graph with sparsity is at most $(1 + O(P/n)) T$. Consequently, the running time for all the sparsification calls combined (ignoring $O(\log \log 1/\wprob)$ factors) is
\[
\Otil(m + P \cdot (T + T P/n))
= \Otil(m + n \epsilon^{-8} \log^7(1/\delta) + \epsilon^{-10} \log^{9}(1/\delta)) ~.
\]
This also upper bounds the running time required for vertex
eliminations, as \textsc{SingleVertexElim} (\autoref{alg:selim}) can
be implemented to run in time $O(d \log d)$, where $d$ is the
combinatorial degree of the vertex being eliminated.
With probability $1-O(\delta)$, this also upper bounds
the running time required for performing the random vertex selections
of low degree vertices, which can be implemented using a simple
rejection sampling approach.
\end{proof}

To prove \autoref{alg:singlePhase} it only remains to show that  in $\SinglePhase$ it is the case that 
\begin{equation}
\Pr\left[
\norm{\u{\mvar{L}}^{\pseudoRoot}
	\left( \mvar{L}^{(k_{\max})} - \mvar{L} \right)
\u{\mvar{L}}^{\pseudoRoot}}_2 > \epsilon
\right]
\leq O(\wprob)
\label{eq:singlePhaseErrProb}
\end{equation}
where $\mvar{L}$ is the input to $\SinglePhase$
and $\mvar{L}^{(k_{\max})}$ is the output. To do this, we set up a matrix martingale as follows. For the inner loops (inside $k$, inside the $t$ loops) of $\SinglePhase$.
We define the change in each step to be:
\[
\mvar{X}^{(k,t)}
\defeq \mvar{S}^{(k,t)} -\mvar{S}^{(k,t-1)}.
\]
Each $\mvar{X}^{(k,t)}$ has zero expectation, hence we can define the following zero-mean martingale
sequence
\[
\mvar{M}^{(k,t)} 
=
\sum_{\hat{k}=1}^{k} \sum_{\hat{t} = 1}^{
	\substack{\text{$P$ if $\hat{k} < k$}\\ \text{$t$ if $\hat{k} = k$}}
}
\mvar{X}^{(\hat{k}, \hat{t})}
= \sum_{(\hat{k}, \hat{t}) \leq (k, t)}
\mvar{X}^{(\hat{k},\hat{t})}.
\]
Here and for the remainder of this section we overload the $\leq$ and $<$ notation to handle pair of variables in the
lexicographical sense, e.g. $(\hat{k}, \hat{t}) \leq (k, t)$ if and only if $\hat{k} < k$ or $\hat{k} = k$ and $\hat{t} \leq t$.
%

Note that this martingale does not include the changes introduced by calls to
$\textsc{SparsifyEulerian}$ and therefore it may be the case that 
$
\mvar{L}^{(k)}
\neq \mvar{L} + \mvar{M}^{(k, P)}
$.
To track the changes caused by $\textsc{SparsifyEulerian}$,
we further define changes from the sparsification steps.
We let
\[
\mvar{Z}^{(0)} \defeq \mvar{S}^{(0)} - \mvar{S}
\text{ and }
\mvar{Z}^{(k)}
\defeq
\mvar{S}^{(k)} - \mvar{S}^{(k, P)}
\text{ for all } 
k > 0 
\]
and with this notation define the output matrix $\mvar{L}^{(k)}$ and intermediate matrices $\mvar{L}^{(k,t)}$ by 
\[
\mvar{L}^{(k)}
= \mvar{L}
+ \mvar{M}^{(k,P)}
+ \sum_{\hat{k} = 0}^{k} \mvar{Z}^{(\hat{k})}
\text{ and } 
\mvar{L}^{(k,t)} =
\mvar{L} + \mvar{M}^{(k, t)}
+ \sum_{\hat{k} = 0}^{k - 1} \mvar{Z}^{(\hat{k})}.
\]

Now we wish to analyze this martingale conditioned on certain high probability events holding which make the martingale \emph{safe} or stable for analysis. We defined martingale safety as follows. 

\newcommand{\safe}{\textbf{SAFE}}

\begin{definition}[Martingale Safety]
\label{def:safe}
We let  $\safe^{(k, t)}$ denote the event that the \emph{martingale is safe until
$(k, t)$, for ${t \in \setof{1,\ldots,P+1}}$} which we define as the following two conditions holding:
\begin{enumerate}
\item All calls to \textsc{SparsifyEulerian} strictly before the $k$th elimination
have been successful,
e.g. 
$\normInline{\u{\mvar{S}^{(\hat{k}, P)}}^{\pseudoRoot}
\mvar{Z}^{(\hat{k})}
\u{\mvar{S}^{(\hat{k}, P)}}^{\pseudoRoot}}_2 \leq \epsilon'
$ for all $\hat{k} < k$.
\item
\label{enu:safemartingalepart}
For all indices $(\hat{k}, \hat{t}) < (k,
  t)$ we had
$
\normInline{\u{\mvar{L}}^{\pseudoRoot}
\mvar{M}^{(\hat{k}, \hat{t})}
\u{\mvar{L}}^{\pseudoRoot}}_2\leq \epsilon /2
$.
\end{enumerate}
\end{definition}
%
%

With this notation of Martingale safety established we define the following truncated Martingale as one where the steps incur
no additional error once it fails. Formally, we let
\begin{align}
\mvar{\overline{X}}^{(k,t)}
\defeq \begin{cases}
\mvar{X}^{(k,t)} & \text{if $\safe^{(k,t)}$}\\
\zero & \text{otherwise}.
\end{cases}
\end{align}
Note that the following sequence of sums of $\mvar{\overline{X}}^{(k,t)}$ is another zero
mean martingale:
\[
\mvar{\overline{M}}^{(k,t)} 
= \sum_{(\hat{k}, \hat{t}) \leq (k, t)}
\mvar{\overline{X}}^{(\hat{k},\hat{t})}.
\]



%
To analyze this martingale we first establish the follow nice consequences of 
$\safe^{(k,t)}$.

\begin{lemma}
\label{lem:SafeImplyLandUApprox}
\noindent
If $\safe^{(k,t)}$ holds,
\begin{itemize}
\item 
then for all $(\hat{k}, \hat{t}) < (k, t)$,
\begin{equation}
\normInline{
\u{\mvar{L}}^{\pseudoRoot}
(\mvar{L}^{(\hat{k}, \hat{t})} - \mvar{L} )
\u{\mvar{L}}^{\pseudoRoot}
}_2 \leq \epsilon
\text{ and }
\normInline{
\u{\mvar{L}}^{\pseudoRoot}
(\mvar{L}^{(\hat{k})} - \mvar{L} )
\u{\mvar{L}}^{\pseudoRoot}
}_2 \leq \epsilon 
\label{eq:safeglobal}
\end{equation}
\item
  and we have for all $\hat{k} < k$
  \begin{equation}
\u{\mvar{S}^{(\hat{k})}} \preceq
O\left( 1 \right) \cdot \u{\mvar{L}} ~.\label{eq:safeschur}
\end{equation}
\end{itemize}
\end{lemma}



\begin{proof}[Proof of \autoref{lem:SafeImplyLandUApprox}]
We consider the ordering $(k,1) < (k,2) < \ldots < (k,P) < (k,P+1) < (k+1,1) < \ldots$ and prove by induction on this ordering, that \autoref{eq:safeschur} holds as well as 
\begin{equation}
\sum_{\hat{k} < k} 
\norm{
\u{\mvar{L}}^{\pseudoRoot}
\mvar{Z}^{(\hat{k})}
\u{\mvar{L}}^{\pseudoRoot}
}_2
\leq
\frac{C \eps}{P}
\cdot  N^{(k,t)}
\label{eq:strongsafeglobal}
\end{equation}
Where 
$N^{(k,t)}$ is defined as the number of calls $\textsc{SparsifyEulerian}$ before 
$(k,t)$ and $C$ is defined as a constant such that guarantee of
\autoref{lem:singlePhaseRuntime}
that $N^{(k,t)} = O(P)$ ensures $C \eps N^{(k,t)} / P \leq \eps/2$. Since $
\mvar{L}^{(k, t)} - \mvar{L}
= \mvar{M}^{(k, t)} 
+ \sum_{\hat{k} < k} \mvar{Z}^{(\hat{k})}
$, by triangle inequality and Part~\ref{enu:safemartingalepart} of \autoref{def:safe} this suffices to prove the result.

As our base case, consider the index $(1,1)$.
We call $\textsc{SparsifyEulerian}(\mvar{L} , \wprob/P,
\epsilon') $ to compute $\mvar{S}^{(0)}$, and $\safe^{(1,1)}$ guarantees
this call succeeded.
So \autoref{thm:order_n_sparsifier} immediately tells us
that $\normInline{
\u{\mvar{L}}^{\pseudoRoot}
\mvar{Z}^{(0)}
\u{\mvar{L}}^{\pseudoRoot}
}_2 
\leq \eps' \leq \frac{ C\eps}{P}$, establishing
\autoref{eq:strongsafeglobal} for this index.
\autoref{eq:safeglobal} follows from a triangle
inequality combined with Part~\ref{enu:safemartingalepart} of \autoref{def:safe} and 
Equation~\eqref{eq:safeschur} follows from \autoref{lem:RCDDRobustSchurBound}.

Next, we consider proving the inductive statements when
$\safe^{(k,t+1)}$ holds, assuming the induction hypothesis holds for $\safe^{(k,t)}$.
In this case, the condition in Equation~\eqref{eq:safeschur}
remains unchanged, so it follows immediately from the induction
hypothesis for $\safe^{(k,t)}$.
The sum $\sum_{\hat{k} < k} 
\normInline{
\u{\mvar{L}}^{\pseudoRoot}
\mvar{Z}^{(\hat{k})}
\u{\mvar{L}}^{\pseudoRoot}
}_2$ and upper bound we want for it also remain unchanged, so again we
get \autoref{eq:strongsafeglobal}.
This gives \autoref{eq:safeglobal}, from a triangle
inequality combined with Part~\ref{enu:safemartingalepart} of
\autoref{def:safe}. 
Now we consider proving the inductive statements when
$\safe^{(k+1,1)}$ holds, assuming the induction hypothesis holds for $\safe^{(k,P+1)}$.
From the induction hypothesis for $\safe^{(k,P+1)}$, we have that
$
\normInline{
\u{\mvar{L}}^{\pseudoRoot}
\left(\mvar{L}^{(k,P)}  - \mvar{L} \right)
\u{\mvar{L}}^{\pseudoRoot}
}_2 \leq \epsilon
$ and 
from \autoref{lem:RCDDRobustSchurBound}, we then get $\u{\mvar{S}^{(k,P)}} \preceq
O\left( 1 \right) \cdot \u{\mvar{L}}.$
This ensures that if a call to $\textsc{SparsifyEulerian}$ was made at
the end of $k$th elimination, then since $\safe^{(k+1,t)}$ guarantees
the call succeeded, we have
\[
 \norm{
\u{\mvar{L}}^{\pseudoRoot}
\mvar{Z}^{(k)}
\u{\mvar{L}}^{\pseudoRoot}
}_2 
\leq O(1)
\norm{
\u{\mvar{S}^{(k,P)}}^{\pseudoRoot}
\mvar{Z}^{(k)}
\u{\mvar{S}^{(k,P)}}^{\pseudoRoot}
}_2 
\leq
O(1)
\eps'
\leq 
\frac{ C\eps}{P}
.
\]
This then proves \autoref{eq:strongsafeglobal} for
$\safe^{(k+1,1)}$, which gives \autoref{eq:safeglobal} from a triangle
inequality combined with Part~\ref{enu:safemartingalepart} of \autoref{def:safe}.
Finally, by \autoref{lem:RCDDRobustSchurBound}, we then get 
${\u{\mvar{S}^{(k)}} \preceq O\left( 1 \right) \cdot \u{\mvar{L}},}$
which proves \autoref{eq:safeschur} for $\safe^{(k+1,1)}$.
\end{proof}

Note that this lemma shows that if we can prove $\Pr[\neg \safe^{(n+1,1)} ]
\leq O( \wprob )$, it will imply \autoref{eq:singlePhaseErrProb} and prove \autoref{thm:InPhaseErrorAccumulation}. To prove this, note that
$
\normInline{\u{\mvar{L}}^{\pseudoRoot}
\mvar{M}^{(k,t)}
\u{\mvar{L}}^{\pseudoRoot}}_2 > \epsilon /2
$
implies
$
\normInline{\u{\mvar{L}}^{\pseudoRoot}
\mvar{\overline{M}}^{(k,t)}
\u{\mvar{L}}^{\pseudoRoot}}_2 > \epsilon /2.
$
Hence, when upper bounding the probability of $\Pr[\neg \safe^{(k,t)}
]$, we can instead consider the higher probability event
\[
\neg
\left(
\left(
\forall \hat{k} \leq k \enspace
\norm{\u{\mvar{S}^{(\hat{k}, P)}}^{\pseudoRoot}
\mvar{Z}^{(\hat{k})}
\u{\mvar{S}^{(\hat{k}, P)}}^{\pseudoRoot}}_2 \leq \epsilon'
\right)
\wedge
\left(
\forall (\hat{k}, \hat{t}) \leq (k, P)
\enspace
\norm{\u{\mvar{L}}^{\pseudoRoot}
\mvar{\overline{M}}^{(\hat{k}, \hat{t})}
\u{\mvar{L}}^{\pseudoRoot}}_2\leq \epsilon /2
\right)
\right)
\]
To bound this we use the following rectangular
matrix martingale result from \cite{Tropp11:arxiv}.

\begin{lemma}[Matrix Freedman (from Cor 1.3. of~\cite{Tropp11:arxiv})]
\label{lem:martingale}
Let $\mvar{E}^{(1)} \ldots \mvar{E}^{(N)}$ be a sequence of matrices and let
$
\expec{j - 1}{\mvar{E}^{(j)}}
$
\sidford{Should we use $<$ notation here, i.e. change this to $\expec{< j}{\mvar{E}^{(j)}}$ for consistency with the rest?}
denote the expectation of $\mvar{E}^{(j)}$
conditioned on $\mvar{E}^{(j - 1)}, \mvar{E}^{(j - 2)}, \ldots, \mvar{E}^{(1)}$. If $
\expec{i}{\mvar{E}^{(i)}}  = 0
$
and
$\norm{\mvar{E}^{(i)}}_2 \leq \rho$ with probability 1 for all $i$ then for any error $t$ we have:
\begin{align*}
\prob{}{\exists k \geq 0 ~s.t.~\norm{\sum_{j \leq k} \mvar{E}^{(j)}}_2 \geq t
~\text{AND}~
\norm{\sum_{j \leq k} \expec{j - 1}{\mvar{E}^{(j)} \mvar{E}^{(j)\top}
			+ \mvar{E}^{(j)\top} \mvar{E}^{(j)}}}_2 \leq
                    \sigma^2}
\\
\leq n\cdot \exp\left( \frac{-t ^2}{100 \left(\sigma^2 + t \rho \right)}\right)
~.
\end{align*}
\end{lemma}
As our bounds normalize by $\u{\mvar{L}}$ for simplicity we can define
rescaled quantities:
\begin{align*}
\mvar{\widehat{M}}^{(k, t)}
 \defeq  \u{\mvar{L}}^{\pseudoRoot}
\mvar{\overline{M}}^{(k, t)}
\u{\mvar{L}}^{\pseudoRoot}
\text{ and }
\mvar{\widehat{X}}^{(k,t)}
\defeq  \u{\mvar{L}}^{\pseudoRoot}
\mvar{\overline{X}}^{(k, t)}
\u{\mvar{L}}^{\pseudoRoot}
\end{align*}
Together with a union bound, this allows us to bound $\Pr[\neg \safe^{(k_{\max}+1,1)} ]$ by:
\begin{align}
& \Pr[\neg \safe^{(k_{\max}+1,1)} ]\nonumber \\
& \leq 
\sum_{k} \Pr\left[ 
\norm{\u{\mvar{S}^{(k, P)}}^{\pseudoRoot}
\mvar{Z}^{(k)}
\u{\mvar{S}^{(k, P)}}^{\pseudoRoot}}_2 
> \eps'
\right]
\label{eq:SparseEulError}\\
& + \Pr \left[ \exists~(k, t)~s.t.~\norm{\mvar{\widehat{M}}^{(k,t)}}_2 \geq s
~\text{AND}~
\norm{\sum_{(\hat{k}, \hat{t}) \leq (k, t)}
	\expec{< (\hat{k}, \hat{t})}{\mxhat^{(\hat{k}, \hat{t})} \mxhat^{(\hat{k}, \hat{t})\top}
				+ \mxhat^{(\hat{k}, \hat{t})\top} \mxhat^{(\hat{k}, \hat{t})}}}_2
  \leq \sigma^2 \right]
\label{eq:SmallError}  \\
& + \Pr\left[{\exists (k, t)~s.t.~\norm{\sum_{(\hat{k}, \hat{t}) \leq (k, t)}
		\expec{< (\hat{k}, \hat{t}) }{\mxhat^{(\hat{k}, \hat{t})} \mxhat^{(\hat{k}, \hat{t})\top}
+ \mxhat^{(\hat{k}, \hat{t})\top} \mxhat^{(\hat{k}, \hat{t})}}}_2
  > \sigma^2} \right]
\label{eq:SmallVariance}.
\end{align}

Each call to \textsc{SparsifyEulerian} is made with error probability
parameter $\wprob/P$ and by
\autoref{lem:singlePhaseRuntime}, we call the routine at most
$O(P)$ times, so by a union bound and
\autoref{thm:order_n_sparsifier}, the probability at some call fails
is at most $O(\wprob)$, which bounds the term~\eqref{eq:SparseEulError}.

The other two terms rely on properties of the truncated Martingales,
which in turn rely on the condition $\safe^{(k, t)}$. To bound these first we provide the following simple lemma which uses 
\autoref{lem:SafeImplyLandUApprox} to bound
the error of $\mxhat^{(k, t)}$.
\begin{lemma}
\label{lem:StepErrorUpper}
We have
$
\normInline{ \mxhat^{(k, t)} }_2 \leq O\left( 1 / P \right)
$
over the entire support of $\mxhat^{(k, t)}$ unconditionally.
\end{lemma}
\begin{proof}
Note that if $\safe^{(k, t)}$ no longer holds, the truncation
process sets $\mxhat{(k, t)} = 0$.
Otherwise, since
by~\autoref{lem:SafeImplyLandUApprox} we have
$
\u{local^{(k)}}
\preceq \u{\mvar{S}^{(k - 1)}}
\preceq O\left( 1 \right) \cdot \u{\mvar{L}},
$
and therefore
\[
\norm{\mxhat^{(k, t)}}_2
= 
\norm{\u{\mvar{L}}^{\pseudoRoot}
\mvar{X}^{(k, t)} \u{\mvar{L}}^{\pseudoRoot}}_2
\leq O\left( 1 \right)
\norm{\u{local}^{\pseudoRoot}
\mvar{X}^{(k, t)} \u{local}^{\pseudoRoot}}_2
\leq O \left(\frac{1}{ P} \right).
\]
\rp{Is there a pointer of how increasing norm
can only increase this error?}
Here the last condition follows from the rescaling
factor of $1/ P$, and the bounds on the error
of the single vertex elimination algorithm given
by \autoref{lem:SElim} Part~\ref{part:SElimError}.
\end{proof}

\autoref{lem:StepErrorUpper} implies the steps of
$\mvar{\widehat{M}}^{(k,t)}$ have norm bounded by $ O \left( 1 / P \right)$. Consequently, \autoref{lem:martingale} yields that for
$\sigma^2 = \frac{\Theta(\log(1/\wprob))}{P} $
and $s = \eps^2 $,
with $P = \Theta(\eps^{-2} \log^2(1/\wprob))$ and
${\log(1/\wprob) \geq \Omega(\log n)}$, 
 the probability
\autoref{eq:SmallError} is
upper bounded by
\begin{equation}
n \exp\left( \frac{-t ^2}{100 \left(\sigma^2 + t \rho
    \right)}\right)
\leq 
O(\wprob)
\label{eq:dirMartingaleErrProb}
.
\end{equation}
Consequently, all that remains is to bound \eqref{eq:SmallVariance} which we do with the following lemma.


\begin{lemma}
\label{lem:VarianceHalf} For $\sigma^2 = \frac{\Theta(\log(1/\wprob))}{P} $ we have
\[
\prob{}{\exists (k, t)~s.t.~\norm{\sum_{(\hat{k}, \hat{t}) \leq (k, t)}
		\expec{< (\hat{k}, \hat{t}) }{\mxhat^{(\hat{k}, \hat{t})} \mxhat^{(\hat{k}, \hat{t})\top}
+ \mxhat^{(\hat{k}, \hat{t})\top} \mxhat^{(\hat{k}, \hat{t})}}}_2
  > \sigma^2}
\leq 
O(\wprob).
\] 
\end{lemma}

\begin{proof}
  \begin{align}
    \prob{}{\exists (k, t)~s.t.~\norm{\sum_{(\hat{k}, \hat{t}) \leq (k, t)}
    \expec{< (\hat{k}, \hat{t}) }{\mxhat^{(\hat{k}, \hat{t})} \mxhat^{(\hat{k}, \hat{t})\top}
    + \mxhat^{(\hat{k}, \hat{t})\top} \mxhat^{(\hat{k}, \hat{t})}}}_2
    > \sigma^2}
\nonumber
\\    
\leq
    \prob{}{\exists (k, t)~s.t.~\norm{\sum_{(\hat{k}, \hat{t}) \leq (k, t)}
    \expec{< (\hat{k}, \hat{t}) }{
    \mxhat^{(\hat{k}, \hat{t})\top} \mxhat^{(\hat{k}, \hat{t})}}}_2
    > \sigma^2 /2 }
\label{eq:FirstVarMartingale}
\\
+
    \prob{}{\exists (k, t)~s.t.~\norm{\sum_{(\hat{k}, \hat{t}) \leq (k, t)}
    \expec{< (\hat{k}, \hat{t}) }{\mxhat^{(\hat{k}, \hat{t})} \mxhat^{(\hat{k}, \hat{t})\top}
    }}_2
    > \sigma^2 /2 }
\label{eq:SecondVarMartingale}
  \end{align}
We will bound each of these two terms by $O(\wprob)$, giving the desired
result. Since the proofs for bounding each of the terms, \eqref{eq:FirstVarMartingale} and \eqref{eq:SecondVarMartingale} are
essentially identical, we only bound \eqref{eq:FirstVarMartingale}.

When  $\safe^{(k, t)}$ does not hold $\mxhat^{(\hat{t}, \hat{k})} = \matzero$. However, when  $\safe^{(k, t)}$ holds, then by \autoref{lem:SafeImplyLandUApprox}, 
\[
\u{local^{(k)}}
\preceq \u{\mvar{S}^{(k - 1)}}
\preceq O\left( 1 \right) \u{\mvar{L}},
\]
in which case 
\[
\mvar{X}^{(k, t) \trp}
\u{\mvar{L}}^{\dag}
\mvar{X}^{(k, t) \trp}
\preceq
O(1)
\mvar{X}^{(k, t) \trp}
\u{local^{(k)}}^{\dag}
\mvar{X}^{(k, t) \trp} ~.
\]
This implies that
\begin{align*}
\mxhat^{(k, t) \trp } \mxhat^{(k, t)}
&=
\u{\mvar{L}}^{\pseudoRoot}
\mvar{\overline{X}}^{(k, t) \trp}
\u{\mvar{L}}^{\dag}
\mvar{\overline{X}}^{(k, t)}
\u{\mvar{L}}^{\pseudoRoot}
\\
&\preceq
O(1)
\u{\mvar{L}}^{\pseudoRoot}
\u{local}^{1/2}
\u{local}^{\pseudoRoot}
\mvar{X}^{(k, t) \trp}
\u{local^{(k)}}^{\dag}
\mvar{X}^{(k, t)}
\u{local}^{\pseudoRoot}
\u{local}^{1/2}
\u{\mvar{L}}^{\pseudoRoot}
\\
&\preceq
\frac{O(1)}{P^2}
\u{\mvar{L}}^{\pseudoRoot}
\u{local}
\u{\mvar{L}}^{\pseudoRoot}
\end{align*}
where in the last line we used that $\normInline{\u{local}^{\pseudoRoot} \mvar{X}^{(k, t)} \u{local}^\pseudoRoot}_2 = O(1/P)$ by \autoref{lem:SElim}.
This bound holds unconditionally, so clearly we also have
\[
\expec{< (k, t) }{ \mxhat^{(k, t) \trp } \mxhat^{(k, t)}
} \preceq
\frac{O(1)}{P^2}
\u{\mvar{L}}^{\pseudoRoot}
\u{local}
\u{\mvar{L}}^{\pseudoRoot}
.
\]
%
%
%
%
Summing for a fixed $k$ over the $P$ samples made in one round of
elimination, we get
\[
\sum_{t} \expec{< (k, t) }{ \mxhat^{(k, t) \trp } \mxhat^{(k, t)}}
\preceq
\frac{O(1)}{P}
\u{\mvar{L}}^{\pseudoRoot}
\u{local}
\u{\mvar{L}}^{\pseudoRoot}
.
\]
We define $\WW_{0} \defeq \matzero$ and
$\WW_k \defeq \sum_{\hat{k} \leq k} \sum_{t} \expec{< (\hat{k}, t) }{
  \mxhat^{(\hat{k}, t) \trp } \mxhat^{(\hat{k}, t)}}$.
This gives $ \matzero \preceq \WW_k - \WW_{k-1} $.
Starting from \autoref{lem:StepErrorUpper}, direct algebraic manipulations then imply that if
$\safe^{(k,t)}$ holds, then
\[
\norm{\mxhat^{(k,t) \top} \mxhat^{(k,t)}}_2
\leq O\left( 1 / P^2\right),
\]
and when $\safe^{(k,t)}$ does not hold, we get this bound trivially
from $\mxhat^{(k,t)} = \matzero.$
Triangle inequality then gives
\begin{equation}
\norm{\WW_k - \WW_{k-1}} \leq O(P/P^2) = O(1/P)
\label{eq:dirvarnorm}
\end{equation}
and 
if we let $\expec{<k}{\cdot}$ denote expectation of $\WW_k$
over the random choice of vertex to eliminate, conditional on all the
random choices of the algorithm until the $k$th elimination, then by
\autoref{eq:exptlocalglobal},

\[
\expec{<k}{\WW_k - \WW_{k-1}} \preceq \frac{O(1)}{P(n-k)}
\lapid
.
\]
%
%
where $lapid$ is the projection matrix orthogonal to the all ones vector. Summing over all $\hat{k} \leq k$, and using $k \leq n / 2$
gives
\begin{equation}
\sum_{\hat{k} \leq k} 
\expec{<\hat{k} }{\WW_k - \WW_{k-1}} 
\preceq O\left( \frac{1}{P} \right)  \lapid.
\label{eq:dirvarexpec}
\end{equation}

We now construct a zero mean martingale which we will use to bound
the probability in term~\eqref{eq:FirstVarMartingale}, by an application of
\autoref{lem:martingale}.
%
Let
\[
\VV_{k} \defeq \WW_{k} - \WW_{k-1} - \expct_{<k}\left[ \WW_{k} -
  \WW_{k-1} \right] = \WW_{k} - \expct_{<k}\left[ \WW_{k}\right]
~.
\]
$\VV_{k}$ is zero-mean conditional on the random choices up to step $k$
and so $\RR_{k} = \sum_{j=1}^{k} \VV_{j}$ is a zero-mean martingale.
\begin{align*}
\RR_{k}
=
\sum_{j=1}^{k} \WW_{j} - \WW_{j-1} - \expct_{<j}\left[ \WW_{j} -
  \WW_{j-1} \right]
=
\WW_{k} - \sum_{j=1}^{k} \expct_{<j}\left[ \WW_{j} -\WW_{j-1} \right]
~.
\end{align*}
Let ${\mvar{H}_{k} = \sum_{j \leq k } \expct_{<j}
  \VV_{j}\VV_{j}^{\trp} +  \VV_{j}^{\trp}\VV_{j}= \sum_{j \leq k }
  \expct_{<j} 2\VV_{j}^2}$.
Note that unconditionally
\begin{align}
\mvar{H}_{k} 
&= 2\sum_{j \leq k } \expct_{<j} \VV_{j}^2
= 2\sum_{j \leq k } \expct_{<j} 
\left(
\WW_{j} - \WW_{j-1} - \expct_{<j}\left[ \WW_{j} -
  \WW_{j-1} \right] 
\right)^2
\preceq
2\sum_{j \leq k } \expct_{<j} 
\left(
\WW_{j} - \WW_{j-1}
\right)^2
\nonumber
\\
&\preceq
2\sum_{j \leq k } \expct_{<j} 
\left(
\WW_{j} - \WW_{j-1}
\right)
\norm{\WW_{j} - \WW_{j-1}}
\preceq
\frac{O(1)}{P^2} \lapid
\label{eq:dirvarvarnorm}
.
\end{align}

Note, in the line above, we are using that $\WW_{j} - \WW_{j-1} \succeq
\matzero$, i.e. this difference is PSD, and 
we used that for any symmetric matrix $\AA$,
$\expct (\AA - \expct \AA)^2 = \expct \AA^2 - (\expct \AA)^2$,
and when $\AA$ is PSD, we hence get
$\expct (\AA - \expct \AA)^2 \preceq \expct \AA^2$.


%

Let $\omega^2 = \frac{C}{P^2} $, for some absolute constant 
$C$ chosen such that 
${\Pr[\exists i :\lambda_{\max}( \mvar{H}_{i} )> \omega^2  ] = 0}$
by  \autoref{eq:dirvarvarnorm}.
Now, by \autoref{eq:dirvarexpec}
we get 
\begin{align*}
  \Pr[ \exists k :
\lambda_{\max}(\WW_{k}) > \sigma^2 ]
&=
  \Pr\left[ \exists k :
\lambda_{\max}\left(\RR_{k}+
\sum_{j=1}^{k} \expct_{<j}\left[ \WW_{j} -\WW_{j-1}
  \right]\right) 
> \sigma^2 \right]
\\
&\leq
  \Pr\left[ \exists k :
\lambda_{\max}\left(\RR_{k}\right) 
> \sigma^2
  - \frac{O(1)}{P}
\right]
\\
&
=
 \Pr\left[ \exists k : \lambda_{\max}(\RR_{k}) \geq \sigma^2 
- \frac{O(1)}{P}
\text{ and }
  \lambda_{\max}( \mvar{H}_{k} ) \leq \omega^2  \right]
.
\end{align*}
We now want to apply \autoref{lem:martingale}.
By \autoref{eq:dirvarnorm}
\begin{align*}
\norm{\VV_{k}}
&= \norm{\WW_{k} - \WW_{k-1} - \expct_{<k}\left[ \WW_{k} -
  \WW_{k-1} \right] } \\
&\leq \max\setof{ \norm{\WW_{k} - \WW_{k-1}} ,
  \norm{ \expct_{<k}\left[ \WW_{k} - \WW_{k-1} \right]}}
& \text{(since both terms are PSD)}
\\
&\leq \frac{O(1)}{P}
,
\end{align*}
which gives us a value for the norm control parameter $R$.
Thus by \autoref{lem:martingale}, 
and using $\sigma^2 = \frac{\Theta(\log(1/\wprob))}{P}$, 
$\log(1/\wprob) \geq \Omega(\log n)$,
and $\omega^2 =
\frac{O(1)}{P^2}$,
we get for an appropriate choice of constants that
\begin{align*}
\Pr\left[ \exists i : \lambda_{\max}(\RR_{i}) \geq \sigma^2 - \frac{O(1)}{P}
\text{ and }
  \lambda_{\max}( \mvar{H}_{i} ) \leq \omega^2  \right]
\leq
n\exp
\left( 
-\frac{ \left(\sigma^2 - \frac{O(1)}{P} \right)^2  }
{\sigma^2 +
  \frac{O(1)}{P^2}}
\right)
\leq
\wprob
~.
\end{align*}

This completes the bound on the probability
term~\eqref{eq:FirstVarMartingale}, and similarly, we can show the
term~\eqref{eq:SecondVarMartingale} is bounded by $\wprob$.
\end{proof}

We now how everything we need to prove \autoref{thm:InPhaseErrorAccumulation}.

\begin{proof}[Proof of \autoref{thm:InPhaseErrorAccumulation}]
The running time guarantees we need were established in \autoref{lem:singlePhaseRuntime}.
Based on \autoref{lem:SafeImplyLandUApprox}, we observed earlier that
${\Pr[\neg \safe^{(n+1,1)} ] \leq O(\wprob) }$ implies
Equation~\eqref{eq:singlePhaseErrProb}.
Our bounds on each
of the terms~\eqref{eq:SparseEulError}, \eqref{eq:SmallError} (see Equation~\eqref{eq:dirMartingaleErrProb})
and~\eqref{eq:SmallVariance} (see \autoref{lem:VarianceHalf}))
establish this, hence proving the theorem.

\end{proof}

\section{Bounding Error Accumulations}
\label{sec:ErrorAccumulation}

In this section we study of the overall accumulation of errors
resulting from the single phases, and prove \autoref{lem:CumulativeErrorNew}.
In particular, we're interesting in the operator $\mvar{F}$ obtained by summing
the undirectifications of the Eulerian Laplacians obtained at the end of each
phase.

\todo{would it be useful to restate the setup from Lemma~\ref{lem:CumulativeErrorNew} here again?}

\rp{If we do stick with that setup, there should not be any
$L$ in this portion: should have just $M$s.}

We first show the overall error accumulation
from Part~\ref{part:CumulativeError}.
For this we need to invoke Lemma~B.2 of~\cite{CohenKPPRSV16},
which we state below for completeness.

\begin{lemma}(Lemma~B.2 of~\cite{CohenKPPRSV16})
\label{lem:SingularValueOpt}
For all $\AA \in \Re^{n \times n}$ and symmetric PSD $\MM, \NN
\in \Re^{n \times n}$ such that
$\ker(\MM) \subseteq \ker(\AA^{\top})$
and
$\ker(\NN) \subseteq \ker(\AA)$
we have
\[
\norm{\MM^{-1/2} \AA \NN^{-1/2}}_2
=
\max_{\xx, \yy \neq \zero}
\frac{\xx^{\top} \AA \yy}{\sqrt{\left( \xx^{\top} \MM \xx \right)
	\left( \yy^{\top} \NN \yy \right)}}
=
2 \cdot
\max_{\xx, \yy \neq \zero}
\frac{\xx^{\top} \AA \yy}{\xx^{\top} \MM \xx + \yy^{\top} \NN \yy}.
\]
where in each of the maximization problems we define $0/0$ to be 0.   
\end{lemma}

\begin{proof}[Proof of \autoref{lem:CumulativeErrorNew}
Part~\ref{part:CumulativeError}]

Recall that the assumption fo error per phase gives:
\[
\norm{\u{\mvar{S}^{\left(i_p\right)}}^{\dag 1/2}
	\left(\mvar{M}^{\left(i_p\right)} - \mvar{M}^{\left(i_{p + 1}\right)} \right)
\u{\mvar{S}^{\left(i_p\right)}}^{\dag 1/2}}_2
\leq \theta_p \epsilon,
\]
for any $0 \leq p < p_{\max}.$.

By \autoref{lem:SingularValueOpt}, this implies for every $p$
	$$2 \xx \left(\mvar{M}^{\left(i_p\right)} - \mvar{M}^{\left(i_{p + 1}\right)} \right) \yy 
	\leq \theta_p \epsilon  \pr{\xx^{\top} \u{\mvar{S}^{\left(i_p\right)}} \xx + \yy^{\top} \u{\mvar{S}^{\left(i_p\right)}} \yy}.$$
	  So summing over these gives
	  $$2 \xx \left(\mvar{M}^{\left(0\right)} - \mvar{M}^{\left(i\right)} \right) \yy 
	\leq  \epsilon \pr{\xx^{\top} \mvar{F} \xx + \yy^{\top} \mvar{F}\yy}.$$
	Again by \autoref{lem:SingularValueOpt}, this gives
	\[
	\norm{\mvar{F}^{\dag 1/2}
		\left( \mvar{M}^{\left(i\right)} - \mvar{M}^{\left(0\right)} \right)
		\mvar{F}^ {\dag 1/2}}
	\leq \epsilon.
	\]
	
\end{proof}

We now turn our attention to the additional condition on $\mvar{F}$
outlined in Part~\ref{part:FNormBounds}.
For this proof, we define a new matrix $\widehat{\mvar{F}}$,
which is made up of the various
Schur complements of the final matrix $\mvar{M}^{(n)}$ onto the
corresponding intermediate spaces.

Let $I_p = \left \{i_p \ldots n \right\}$ be the set of vertices remaining after first $p$ phases.
Let
\begin{equation}
\Fhat
\defeq \sum_{0 \leq p < p_{\max}}
\theta_p \cdot \u{\sc{\mvar{M}^{\left(n\right)}}{\left \{i_p \ldots n \right\} }}
\label{eq:Fhat}
\end{equation}
where $\sum_p \theta_p = 1$ and $\theta_p \geq 1/O(p_{\max}) = 1/O(\log n)$.

Rewriting $\Fhat$ as differences between consecutive steps
shows that it is in fact close to $\mvar{F}$.
\begin{lemma}
	\label{lem:FhatAndF}
The matrices $\mvar{F}$ and $\Fhat$ as defined
in Equations~\ref{eq:F} and~\ref{eq:Fhat} respectively satisfy:
	\[
	\Fhat \approx_{O(p_{\max})}  \mvar{F},
	\]
	where $p_{\max} = O(\log n)$ is the number of invocations of \textsc{SinglePhase} (\autoref{alg:singlePhase}) by $\eulLU$ (\autoref{alg:multiPhase}). 
\end{lemma}

\begin{proof}
		The key observation is that because the Schur complement
	steps after step $i_p$ are completely contained among
	the vertices $\{i_p, \ldots n\}$,
	the difference between $\mvar{F}$ and $\widehat{\mvar{F}}$
	can be bounded using the discrepancies at the steps.
	
	Formally, the choice of pivots means we have
	\[
	\u{\sc{\mvar{M}^{\left(n\right)}}{\left \{i_p \ldots n \right\} }}
	= \u{\mvar{S}^{(i_p)}}
	+ \sum_{p' \geq p} \left( \u{\mvar{M}^{\left(i_{p'}\right)}} 
	- \u{\mvar{M}^{\left(i_{p' + 1}\right)}}\right),
	\]
	which when substituted into the formula for $\widehat{\mvar{F}}$ gives:
	\[
	\widehat{\mvar{F}}
	= \sum_{0 \leq p < p_{\max}}
	\theta_p \u{\mvar{S}^{(i_p)}}
	+ \sum_{0 \leq p < p_{\max} }\sum_{p' \geq p} \theta_p \left( \u{\mvar{M}^{\left(i_{p'}\right)}} 
	- \u{\mvar{M}^{\left(i_{p' + 1}\right)}}\right).
	\]
	Collecting the terms related to $\mvar{F}$, and reversing
	the summation on the $p'$s turns this into:
	\[
	\widehat{\mvar{F}}
	= \mvar{F} 
	+ \sum_{p' \geq p} \left( \sum_{p \leq p'} \theta_{p}\right)
	\left( \u{\mvar{M}^{\left(i_{p'}\right)}} 
	- \u{\mvar{M}^{\left(n\right)}}\right).
	\]
	By triangle inequality we then get:
	\[
	\norm{\mvar{F}^{\dag 1/2} \left( \widehat{\mvar{F}} - \mvar{F} \right) \mvar{F}^{\dag 1/2}}_2
	\leq \sum_{p'} \left( \sum_{p \leq p'} \theta_{p}\right)
	\norm{\mvar{F}^{\dag 1/2} \left(  \u{\mvar{M}^{\left(i_{p'}\right)}} 
	- \u{\mvar{M}^{\left(n\right)}} \right) \mvar{F}^{\dag 1/2} }_2
	\]
	\todolow{can probably prove a tighter bound by invoking Lemma~something?}
	Since $\sum \theta_p =1$, the above is at most $\epsilon p_{\max}$
	provided the maximum error over any consecutive sequences of phases is $\epsilon$, which happens $1 - O(\delta)$ by Part~\ref{part:CumulativeError} of \autoref{lem:CumulativeErrorNew} shown above. Since the $\epsilon$ argument to $\eulLU$ (\autoref{alg:multiPhase}) is required to be $\leq 1$, the desired result follows.
\end{proof}


\begin{lemma}
\label{lem:FhatBound}
Let $\mvar{M} = \mvar{M}^{(0)}$ be a (possibly asymmetric) matrix,
$\mvar{M}^{(1)}, \mvar{M}^{(2)}, \ldots \mvar{M}^{(n)}$
be the intermediate elimination states with errors
defined in \autoref{lem:CumulativeErrorNew}
using $I_0 = V, \ldots ,I_{p_{\max} - 1}$ that are
nested  subsets of indices, i.e.,
	$I_0 \subseteq I_{p + 1} \subseteq \ldots \subseteq I_{p_{\max}-1}$
	and $c_0,c_1, \ldots ,c_p$ be constants.
Then the matrix $\widehat{\mvar{F}}$ as defined above in
\autoref{eq:Fhat} and
$\widetilde{\mvar{M}} \defeq \mvar{M}^{(n)}$
satisfy
	\[
	\widehat{\mvar{F}}
	\preceq  
	\widetilde{\mvar{M}}^{\top}
	\widehat{\mvar{F}}^{\dag}
	\widetilde{\mvar{M}},
	\]
\todo{are these definitions needed? can we point to a global
set of dfns above?}
 \sidford{How are $L^{(n)}$ and $L$ connected when arbitrary?}
\end{lemma}

\newcommand{\Smat}{\SS} 

Throughout this section, we will frequently use the following
definition, which allows us to extend the view of Schur complements as
inverses of coordinate restrictions of the inverse of a matrix
($\sc{\MM}{C} = (\MM^{-1})_{CC}^{-1}$) to the setting pseudo-inverses.
Such characterization requires this definition of
restricting pseudo-inverses to a subset of coordinates,
\begin{definition}[Projected coordinate restriction]
\label{def:projcoordres}
  Consider any $\MM \in \rea^{[n] \times [n]}$,
  and $C,F$ a partition of $[n]$,  where  $\MM_{FF}$ is invertible.
  Let  $\SS = \sc{\MM}{C}$.
  We define the projected coordinate restriction of $\MM^{\dag}$ to $C$ as
\[
\MM^{\dagger}[C] \defeq  \PP_{\SS} (\MM^{\dagger})_{CC} \PP_{\SS^{\top}}
.\]
\end{definition}

Now, let $\mvar{Z}$ be any matrix with $\mvar{Z} + \mvar{Z}^T \succeq 0$ and define \sidford{use the U notation for simplicity?}
\[
\c{\mvar{Z}} \defeq \mvar{Z}^{\top}
\left( \frac{ \mvar{Z} + \mvar{Z}^{\top}}{2} \right)^{\dag} \mvar{Z}
= \left(\frac{\mvar{Z}^{\dag} + \mvar{Z}^{\dag \top}}{2} \right)^{\dag}.
\]


\begin{lemma}\label{lem:complicated_matrix_SC}
If $\MM\in \R^{n\times n}$
such that $\ker(\MM) = \ker(\MM^{\top})$,
 and $I$ is a subset of indices, $\bar{I} = [n] \setminus I$,
such that $\MM_{\bar{I} \bar{I}}$, the principal minor of $\MM$
on the indices outside of $I$, is invertible, then
\begin{enumerate}[label=(\roman*)]
	\item  $\c{\sc{\MM}{I} } =\sc{\c{\MM}}{I}$,
\label{part:SchurComplementRepresentation}
	\item  $\u{\sc{\MM}{I}} \preceq \c{\sc{\MM}{I}}
	\preceq \c{\MM}$.
\label{part:ComplicatedIneq}
\end{enumerate} 
\end{lemma}
\begin{proof}
To prove $(i)$, we invoke the characterization of Schur
complements as minors of inverses
(which we formalize in \autoref{sec:pinv-facts},
specifically \autoref{lem:SchurPinv})
to write the Schur complement as $\sc{\MM}{I}=\left(\MM^\dag[I]\right)^\dag$.
It follows that
\begin{align*}
\c{\sc{\MM}{I} } 
&=\left(\frac{\sc{\MM}{I} ^{\dag} + \sc{\MM}{I}^{\dag \top}}{2} \right)^{\dag}
=\left(\frac{{\MM}^{\dag}[I] + {\MM}^{\dag \top}[I]}{2} \right)^{\dag}
=\left(\c{\MM}^\dag[I]\right)^\dag =\sc{\c{\MM}}{I},
\end{align*}
where to see ${\MM}^{\dag}[I] + {\MM}^{\dag \top}[I] =
  \c{\MM}^\dag[I]$, we need the fact that left and right kernels of
  the involved matrices agree, from which it also follows that
  $(\c{\MM})_{\bar{I} \bar{I}}$ is invertible (otherwise some
  nonzero vector in
  $\ker((\c{\MM})_{\bar{I} \bar{I}})$ would also lie in
  $\ker(\MM_{\bar{I} \bar{I}})$, and the latter matrix is invertible).
This ensures  $\c{\MM}^\dag[I]$ is well-defined.

The first inequality in $(ii)$ is an immediate consequence of \autoref{lem:complicated}.
For the second inequality in $(ii)$, we note that the monotonicity property of the Schur complement stated in \autoref{lem:SchurLess} implies that $\sc{\c{\MM}}{I}\preceq\c{\MM}$.  
Combining this with $(i)$ yields the asserted bound.
\end{proof}

\begin{lemma}\label{lem:nested_schur_ineq} 
Let $\MM\in\R^{n\times n}$ be a matrix 
such that $\ker(\MM) = \ker(\MM^{\top})$,
with subset of indices
$I\subseteq J \subseteq [n]$ such that
the principal minors of $\MM$ on $[n] \setminus I$
and $[n] \setminus J$ are both invertible.
Then
\begin{equation}\label{eq:nested_schur_ineq1}
\left(\MM^\dag \u{\sc{\MM}{I}} \MM^{\dag \top}\right)[J]
\preceq 
\u{\sc{\MM}{J}}^\dag,
\end{equation}
and
\begin{equation}\label{eq:nested_schur_ineq2}
\left(\MM^\dag \u{\sc{\MM}{J}} \MM^{\dag \top}\right)[I]
\preceq 
\u{\sc{\MM}{I}}^\dag.  
\end{equation}
\end{lemma}
\begin{proof}
By \autoref{lem:SchurPinvCompose},  $\sc{ \sc{\MM}{J} }{I} =  \sc{\MM}{I}$.
By \autoref{lem:complicated_matrix_SC}
Part~\ref{part:ComplicatedIneq} and the assumption that $I\subseteq J$, 
$\u{\sc{\MM}{I}}\preceq \c{\sc{\MM}{J}}$, so
\[
\left(\MM^\dag \u{\sc{\MM}{I}} \MM^{\dag \top}\right)[J]\preceq \left(\MM^\dag \c{\sc{\MM}{J}} \MM^{\dag \top}\right)[J]
\]
Since the kernels agree, and $\c{\sc{\MM}{J}}$ is supported on the submatrix with row and column indices in $J$, we can replace $\MM^\dag$ and $\MM^{\dag \top}$ with their respective projected coordinate restrictions $\MM^\dag[J]=\sc{\MM}{J}^\dag$ and $\MM^{\dag\top}[J]=\sc{\MM}{J}^{\dag\top}$.
We thus have
\begin{align*}
	\left(\MM^\dag \u{\sc{\MM}{I}} \MM^{\dag \top}\right)[J]
	&\preceq \sc{\MM}{J}^\dag \c{\sc{\MM}{J}} \sc{\MM}{J}^{\dag \top}\\
	&= \sc{\MM}{J}^\dag \left(\sc{\MM}{J} \u{\sc{\MM}{J}}^\dag \sc{\MM}{J}\right) \sc{\MM}{J}^{\dag \top}\\
	&= \u{\sc{\MM}{J}}^\dag,
\end{align*}
which proves \autoref{eq:nested_schur_ineq1}.  

To prove \autoref{eq:nested_schur_ineq2}, we can restrict $\MM^\dag$ and $\MM^{\dag\top}$ to indices in $J$ as above to rewrite its left-hand side as
\[
\left(\MM^\dag \u{\sc{\MM}{J}} \MM^{\dag \top}\right)[I]=\left(\sc{\MM}{J}^\dag \u{\sc{\MM}{J}} \sc{\MM}{J}^{\dag \top}\right)[I]=\c{\sc{\MM}{J}}^\dag [I]
=\c{\sc{\MM}{I}}^\dag.
\]
We have $\c{\sc{\MM}{I}}\succeq \u{\sc{\MM}{I}}$ by \autoref{lem:complicated_matrix_SC} , so $\c{\sc{\MM}{I}}^\dag \preceq \u{\sc{\MM}{I}}^\dag$, which yields \autoref{eq:nested_schur_ineq2}.
\end{proof}

\begin{proof}[Proof of \autoref{lem:FhatBound}]

Because $\widehat{\mvar{F}}$ is a convex combination
of the ${\sc{\widetilde{\Smat}}{I_p}}$ over the $p$s
(see \autoref{eq:Fhat}),
it suffices to show
\[
\u{\sc{\widetilde{\Smat}}{I_p}}
\preceq  \widetilde{\Smat}^{\top} \widehat{\mvar{F}}^{\dag} \widetilde{\Smat}
\]
for each $p$.
Also, because of the (operator) monotonicity of Schur 
complements given in \autoref{lem:SchurLess},
we can instead show the stronger condition:
	\[
	\u{\sc{\widetilde{\Smat}}{I_p}} \preceq 
	\sc{\widetilde{\Smat}^{\top} \widehat{ \mvar{F} }^{\dag} \widetilde{\Smat}}{I_p}
	= \left( \left( \widetilde{\Smat}^{\top} \widehat{ \mvar{F} }^{\dag} \widetilde{\Smat}\right)^{\dag} \left[I_p\right] \right)^{\dag}.
	\]
	Inverting both sides then reduces it to
	$
	 \u{\sc{\widetilde{\Smat}}{I_p}}^{\dag}
	\succeq
	\left( \widetilde{\Smat}^{\dag} \widehat{\mvar{F}} \widetilde{\Smat}^{\dag \top} \right) \left[I_p\right].
	$
	Using the definition of $\widehat{\mvar{F}}$ gives
	\[
	\left( \widetilde{\Smat}^{\dag} \widehat{\mvar{F}} \widetilde{\Smat}^{\dag \top} \right) \left[I_p\right]
	= \sum_{0 \leq p' < p_{\max}} \theta_{p'} 
	\left( \widetilde{\Smat}^{\dag} \u{\sc{\widetilde{\Smat}}{I_{p'}}} \widetilde{\Smat}^{\dag \top} \right) \left[I_p\right].
	\]
	We now consider the terms separately and note that, by \autoref{lem:nested_schur_ineq}, $\left( \widetilde{\Smat}^{\dag} \u{\sc{\widetilde{\Smat}}{I_{p'}}} \widetilde{\Smat}^{\dag \top} \right) \left[I_p\right] \preceq \u{\sc{\widetilde{\Smat}}{I_p}}^{\dag}$ for every $p'$.   
	Taking a convex combination of these inequalities thus completes the proof.
\end{proof}

We can now conclude things formally.
\begin{proof}[Proof of \autoref{lem:CumulativeErrorNew}
Part~\ref{part:FNormBounds}]
	\todolow{Tighten the parameters of the above pieces.}
	By \autoref{lem:FhatBound}, $\widehat{\mvar{F}} \preceq \left(\mvar{L}^{(n)}\right)^{\top} \widehat{\mvar{F}}^{\dag} \mvar{L}^{(n)}$. By \autoref{lem:FhatAndF}, we have with probability $1-O(\delta)$ that $\Fhat \approx_{O(\log(n))}  \mvar{F}$. Thus, we have
$	
1/O(\log^2 n) \cdot \mvar{F} \preceq \left(\mvar{L}^{(n)}\right)^{\top} \mvar{F}^{\dag} \mvar{L}^{(n)}
$.
\end{proof}


\bibliographystyle{plain}
\bibliography{ref-dirlap,ref-tree,ref-simple,ref-lshard,ref}
\appendix

\section{Finding an $\alpha$-RCDD Block}
\label{sec:rcdd-analysis}

Here provide a basic result showing we can find large $\alpha$-RCDD blocks of vertices efficiently. This results is a natural analog of a result in \cite{LeePS15}. The main result of this section is the following theorem analyzing \autoref{alg:rcdd}.

\begin{theorem}\label{thm:rcdd-correctness}
Given a directed graph $G$, the function
$\textsc{FindRCDDBlock}(G,\alpha)$ (\autoref{alg:rcdd}) outputs
an $\alpha$-RCDD set of vertices of size at least $\frac{n}{16
  (1+\alpha)}$ in time $O(m \log(1/\delta))$ with probability at least
$1-O(\delta)$.
\end{theorem}

\begin{proof}
The proof of this theorem follows immediately from \autoref{lem:getRCDD} which in turn uses \autoref{lem:expected_size} which are proved in the remainder of this section.
\end{proof}

\begin{algorithm}								
\caption{$\textsc{FindRCDDBlock}(G,\alpha)$}
\label{alg:rcdd}
\SetAlgoVlined
					
\KwIn{a directed graph $G$ and a parameter $\alpha$}
\KwOut{an $\alpha$-RCDD set of vertices $F$ of size at least $n/16$}						

$F \gets \emptyset$

\While{$|F| < \frac{n}{16(1+\alpha)}$}{
Let $F$ be $k=\frac{n}{8(1+\alpha)}$ vertices sampled uniformly at random.

Remove from $F$ each vertex that is not $\alpha$-RCDD with respect to $F$.
}
\end{algorithm}


\begin{lemma}
	\label{lem:expected_size} Let $\mlap\in\R^{n\times n}$ be an Eulerian
	Laplacian and let $F\subseteq V$ be a random subset of size $k$.
	Then the expected number of $i\in F$ such that $\sum_{j\in F,j\neq i}|\mlap_{ij}|\geq\frac{1}{1+\alpha}|\mlap_{ii}|$
	is at most $k^{2}(1+\alpha)/n$.\end{lemma}
\begin{proof}
	We have that for all $j\neq i$
	\[
	\Pr\left[j\notin F\,|\,i\in F\right]=\prod_{i\in[k-1]}\left(1-\frac{1}{n-i}\right)=\prod_{i\in[k-1]}\left(\frac{n-i-1}{n-i}\right)=\frac{n-k}{n-1}
	\]
	and therefore $\Pr\left[j\in F\,|\,i\in F\right]=\frac{k-1}{n-1}$.
	Since $\sum_{j\in F:j\neq i}|\mlap_{ij}|=\mlap_{ii}$ we have that
	$\E\left[\sum_{j\in F,j\neq i}|\mlap_{ij}|\,|\,\right]=\left(\frac{k-1}{n-1}\right)\mlap_{ii}$
	and that by Markov's inequality
	\[
	\Pr\left[\sum_{j\in F,j\neq i}|\mlap_{ij}|\geq\frac{1}{1+\alpha}\mlap_{ii}\,|\,i\in F\right]\leq\left(\frac{k-1}{n-1}\right)(1+\alpha)
	\]
	and since 
	\[
	\Pr\left[i\notin F\right]=\prod_{i\in[k]}\left(1-\frac{1}{n+1-i}\right)=\prod_{i\in[k]}\frac{n-i}{n+1-i}=\frac{n-k}{n}=1-\frac{k}{n}
	\]
	we have that $\Pr[i\in F]=\frac{k}{n}$ and the expected number $i\in F$
	such that $\sum_{j\in F,j\neq i}|\mlap_{ij}|\geq\frac{1}{1+\alpha}$
	is at most 
	\begin{align*}
	\sum_{i\in[n]}\Pr\left[i\in F,\sum_{j\in F,j\neq i}|\mlap_{ij}|\geq\frac{1}{1+\alpha}\right] &=\sum_{i\in[n]}\Pr\left[i\in F\right]\Pr\left[\sum_{j\in F,j\neq i}|\mlap_{ij}|\geq\frac{1}{1+\alpha}\mlap_{ii}\,|\,i\in F\right] \\
	&\leq k\left(\frac{k-1}{n-1}\right)(1+\alpha)
	\end{align*}
	Since $k\leq n$ we have $(k-1)/(n-1)\leq k/n$ and the result follows.\end{proof}

\begin{lemma}
\label{lem:getRCDD}
	Let $\mlap\in\R^{n\times n}$ be an Eulerian Laplacian, let $F\subseteq V$
	be a random subset of size $k$ and let $F'\subseteq V$ be the elements
	$i\in F$ for which $\sum_{j\in F,j\neq i}|\mlap_{ij}|\leq\frac{1}{1+\alpha}|\mlap_{ii}|$
	and $\sum_{j\in F,j\neq i}|\mlap_{ji}|\leq\frac{1}{1+\alpha}|\mlap_{ii}|$
	then with probability at least $1/2$ we have
	\[
	|F'|\geq k\left[1-\frac{4k}{(1+\alpha)n}\right]
	\]
	and therefore for $k=\frac{n}{8(1+\alpha)}$, $\mlap_{F'F'}$ is
	$\alpha$-RCDD with $|F'|\geq\frac{n}{16(1+\alpha)}$ with probability
	at least $1/2$.\end{lemma}
\begin{proof}
	Applying Lemma~\ref{lem:expected_size} to $\mlap$ and $\mlap^{\top}$
	we see that the expected number of elements $i\in F$ for which $\sum_{j\in F,j\neq i}|\mlap_{ij}|\geq\frac{1}{1+\alpha}|\mlap_{ii}|$
	is at most $k^{2}(1+\alpha)/n$ and the expected number of elements
	$i\in F$ for which $\sum_{j\in F,j\neq i}|\mlap_{ij}|\geq\frac{1}{1+\alpha}|\mlap_{ii}|$
	is at most $k^{2}(1+\alpha)/n$ consequently an expected $2k^{2}(1+\alpha)/n$
	are removed from $F$ to get $F'$. Consequently, by Markov's inequality
	with probability at least $1/2$ at most $2k^{2}(1+\alpha)/n$ are
	removed from $F$ to get $F'$
\end{proof}

\section{Matrix Facts}
\label{sec:mat_facts}

Here we provide some general matrix facts we use throughout the paper. 

\begin{lemma}\label{lem:psd_invert}
If $\mdir$ is a square matrix with $\mU_\mdir \succ \mzero$ then $\mdir$ is invertible.
\end{lemma}

\begin{proof}
If $d \neq \vzero$ with $\mdir d = \vzero$ then $0 = d^\top \mdir d = d^\top \mU_{\mdir} d$ contradicting $\mU_{\mdir} \succ \mzero$.
\end{proof}



\begin{lemma}\label{lem:complicated}
If $\mdir$ is a matrix with $\ker(\mdir) = \ker(\mdir^\top) = \ker(\mU_{\mdir})$ and $\mU_{\mdir} \succeq \mzero$ then $\mU_{\mdir} \preceq \mdir^\top \mU_{\mdir}^\dagger \mdir$.
\end{lemma}

\begin{proof}
This was previously shown with slightly different hypotheses in Equation 2.1 and Theorem~2.2 in~\cite{mathias92} and~Lemma 13 in~\cite{cohen2016faster}, and in the current form in Lemma~B.9 in~\cite{CohenKPPRSV16}.  For completeness, we include the proof from~\cite{CohenKPPRSV16}, which we reproduce here almost verbatim.

Let $\mdir=\mU+\mvar{V}$, where 
\[
\mU:=\mU_{\mdir}=(\mdir+\mdir^\top)/2\text{ and }\mvar{V}=(\mdir-\mdir^\top)/2.
\]
We have 
\[
\mdir^\top \mU^\dagger \mdir 
=(\mU+\mvar{V})^\top \mU^\dagger (\mU + \mvar{V})
=\mU^\top \mU^\dagger \mU 
+\mU^\top \mU^\dagger \mvar{V}
+\mvar{V}^\top \mU^\dagger \mU 
+\mvar{V}^\top \mU^\dagger \mvar{V}.
\]
Our kernel assumptions imply that $\mU^\top \mU^\dagger \mvar{V}=\mvar{V}$ and $\mvar{V}^\top \mU^\dagger \mU =\mvar{V}^\top$, 
and $\mvar{V}^\top=-\mvar{V}$, so we obtain
\begin{align*}
\mdir^\top \mU^\dagger \mdir 
&=\mU^\top \mU^\dagger \mU 
+\mvar{V}
+\mvar{V}^\top 
+\mvar{V}^\top \mU^\dagger \mvar{V}
\\
&=\mU 
+\mvar{V}^\top \mU^\dagger \mvar{V}\succeq \mU,
\end{align*}
where the final inequality used the assumption that $\mU\succeq 0$ to guarantee that $\mvar{V}^\top \mU^\dagger \mvar{V}\succeq 0$.
\end{proof}

\begin{lemma}\label{lem:approx-implies-lfl}
Let $\widetilde{\mvar{L}},\mvar{L},\mvar{F}$ be arbitrary matrices with $\ker(\widetilde{\mvar{L}})=\ker(\widetilde{\mvar{L}}^\intercal)=\ker(\mvar{L})=\ker(\mvar{L}^\intercal)=\ker(\mvar{F})=\ker(\mvar{F}^\intercal)$. If $\|\mvar{F}^{+/2} (\mvar{L} - \widetilde{\mvar{L}}) \mvar{F}^{+/2}\| \leq \epsilon$ and $\gamma \mvar{F} \preceq \widetilde{\mvar{L}}^\intercal \mvar{F}^+ \widetilde{\mvar{L}}$ then $\mvar{L}^\intercal \mvar{F}^+ \mvar{L} \approx_{O\left(\frac{\epsilon}{\sqrt{\gamma}} + \frac{\epsilon^2}{\gamma} \right)} \widetilde{\mvar{L}}^\intercal \mvar{F}^+ \widetilde{\mvar{L}}$. 
\end{lemma}

\begin{proof}
We have
\begin{align*}
\|\mvar{F}^{+/2} (\mvar{L} - \widetilde{\mvar{L}}) \mvar{F}^{+/2}\| &\leq \epsilon \\
\|(\mvar{L} - \widetilde{\mvar{L}}) x \|_{\mvar{F}^+} &\leq \epsilon \cdot \|x\|_{\mvar{F}} & \forall x \\
\left| \|\mvar{L}x\|_{\mvar{F}^+} - \|\widetilde{\mvar{L}} x \|_{\mvar{F}^+} \right| &\leq \epsilon / \sqrt{\gamma} \cdot \|\widetilde{\mvar{L}} x\|_{\mvar{F}^+} & \forall x \\
\left| x^\intercal \mvar{L}^\intercal \mvar{F}^+ \mvar{L} x - x^\intercal\widetilde{\mvar{L}}^\intercal \mvar{F}^+ \widetilde{\mvar{L}} x \right| &\leq O\left(\frac{\epsilon}{\sqrt{\gamma}} + \frac{\epsilon^2}{\gamma} \right) \cdot x^\intercal\widetilde{\mvar{L}}^\intercal \mvar{F}^+ \widetilde{\mvar{L}} x & \forall x, \\
\end{align*}
which is one definition of the desired condition.
\end{proof}

\begin{lemma}
\label{lem:SchurLess}
For any positive semi-definite matrix $\mvar{P} \succeq 0$
and any subset of variables $I$, we have
\[
\sc{\mvar{P}}{I} \preceq \mvar{P}.
\]
\end{lemma}

\begin{proof}
This follows from the optimization definition
of Schur complements:
\[
\xx^{\top} \sc{\mvar{P}}{I} \xx
\defeq
\min_{\xxhat: \xxhat_{I} = \xx_{I}}
\xxhat^{\top} \mvar{P} \xxhat.
\]
A formal proof of this fact can be found in Lemma B.2
(proven in Appendix C) of~\cite{MillerP13}.
\end{proof}

\begin{lemma}
\label{lem:dlocalulocal}
Suppose $\aa$ is a vector with positive entries and $\DD$ is the
diagonal matrix with $\aa$ on the diagonal, and $d = \vecone^{\top}
\aa$, then $\UU = \DD - \frac{1}{d} \aa \aa^{\top}$ is satisfies
\[
  \UU^{\dagger} \preceq \DD^{-1}.
\]
\end{lemma}

To prove this lemma, we will use a standard fact about
pseudo-inverses:
\begin{fact}
  Suppose $\AA$ is a symmetric matrix and $\XX$ is a non-singular
  matrix, and that $\PP$ is the projection onto the image of
  $\XX^{\top} \AA \XX$.
  Then, 
  \[
    (\XX^{\top} \AA \XX)^{\dagger}  = 
    \PP \XX^{-1} \AA^{\dagger} (\XX^{-1})^{\top} \PP
   \]
\end{fact}

\begin{proof}[Proof of Lemma~\ref{lem:dlocalulocal}]
  Note that one can check that $\UU$ is in fact the undirected
  Laplacian of a weighted complete graph and so kernel $\UU$ is
  exactly the span of $\vecone$, and $\UU$ is PSD.
  Let $\PP = \II - \vecone\vecone^{\top}$ denote the projection onto
  the image of $\UU$.
  Let $\vv = \DD^{-1/2} \frac{1}{d^{1/2}} \aa$, and $\VV = \II  -
  \vv\vv^{\top}$.
  Note that $\UU = \DD^{1/2} \VV \DD^{1/2}$ and
  $\VV^{\dagger} = \VV$.
  Hence 
  \[
    \UU^{\dagger} = (\DD^{1/2} \VV \DD^{1/2})^{\dagger}
    =  \PP \DD^{-1/2} \VV^{\dagger} \DD^{-1/2} \PP
    =  \PP \DD^{-1/2} \VV \DD^{-1/2} \PP
    \preceq \PP \DD^{-1} \PP
    \preceq \DD^{-1}.
   \]
\end{proof}



\section{Pseudo-Inverses and Schur Complements}
\label{sec:pinv-facts}

In this appendix, we formally justify our view of Schur complements
as taking inverses of minors of inverses in the setting of pseudo-inverses.
Such characterization requires Definition~\ref{def:projcoordres} which
defines the appropriate way to 
restrict pseudo-inverses to a subset of coordinates.
Our main equivalence statement is:
\begin{lemma}
\label{lem:SchurPinv}
Consider any $\MM \in \rea^{[n] \times [n]}$,
and $F,C$ a partition of $[n]$,  where  $\MM_{FF}$ is invertible.
Then we have
\[
 \sc{\MM}{C}
 =\left(\MM^\dag\left[C\right]\right)^\dag.
\]
\end{lemma}
Using this characterization, we can also show the next lemma, which
tells us that a Schur complement onto a set can be computed by
blockwise elimination, or by first eliminating a some variables and
then eliminating more. This is a well-known result for invertible
matrices, or symmetric matrices, but we extend it to the case of
singular, asymmetric matrices.
\begin{lemma}
\label{lem:SchurPinvCompose}
Consider any $\MM \in \rea^{[n] \times [n]}$,
and $F,C$ a partition of $[n]$,  
where  $\MM_{FF}$ is invertible.
Let $F_1,F_2$ be a partition of $F$, and let $C_1 = [n] \setminus F_1$.
Suppose $\MM_{F_1 F_1}$ is invertible.
Then we have
\[
 \sc{ \sc{\MM}{C_1} }{C} =  \sc{\MM}{C}
\]
and 
\[
\left(\MM^\dag[C_1]\right)\left[C\right]
=
\left(\MM^\dag\left[C_1\right]\right).
\]
\end{lemma}

\subsection{Pseudo-Inverse of a Product}

Given a real matrix $\MM \in \rea^{m \times n}$
with kernel $\ker(\MM)$,
we let $\PP_{\MM}$ denote the orthogonal projection onto
its columns pace $\ker(\MM)^{\bot}$, and
$\QQ_{\MM} = \II_{n \times n} - \PP_{\MM}$ denote the
orthogonal projection onto $\ker(\MM)$.
Recall that such orthogonal projections are symmetric matrices.
Note also that $\MM \PP_{\MM} = \MM$, and so $\MM \QQ_{\MM} = \mzero$.
Similarly, we can show $\QQ_{\MM^{\top}} \MM  = \mzero$.

In this subsection, we prove the following helpful lemma that
characterizes the pseudo-inverse of a product.
\begin{lemma}
\label{lem:pinvasymproduct}
  Consider real matrices
  $\AA \in \rea^{m \times m}$,
  $\BB \in \rea^{m \times n}$,
  and $ \CC \in \rea^{n \times n}$,
  where $\AA$ and $\CC$ are invertible.
  Let $\MM = \AA \BB \CC$
  Then $\MM^{\dagger} = \PP_\MM \CC^{-1} \BB^{\dagger} \AA^{-1} \PP_{\MM^{\top}}$.
\end{lemma}

Before proving this lemma we recall a standard fact about
pseudo-inverses (e.g. see \cite{horn1990matrix}, 2nd
ed. p. 453).

\begin{fact}
\label{fac:pinvbysymtest}
  The pseudo-inverse of $\MM$, denoted by $\MM^{\dagger}$, is the unique operator satisfying 
  \begin{enumerate}
  \item  \label{enu:mcancel}
    $\MM = \MM \MM^{\dagger}\MM$.
  \item 
\label{enu:pinvcancel}
$\MM^{\dagger} = \MM^{\dagger}\MM \MM^{\dagger}$.
\item
\label{enu:pisym}
$(\MM^{\dagger} \MM)^{\top} = \MM^{\dagger} \MM = \PP_{\MM}$.
\item
\label{enu:pitrpsym}
$(\MM \MM^{\dagger})^{\top} = \MM \MM^{\dagger} = \PP_{\MM^{\top}}$.
  \end{enumerate}
\end{fact}

To prove \autoref{lem:pinvasymproduct},
we also need a simple observation about the projection operations
related to $\MM$ and $\BB$.
\begin{claim}
\label{clm:projzero}
\noindent
  \begin{enumerate}
  \item 
    \label{enu:projzero}
    $ \PP_\MM \CC^{-1} \QQ_{\BB} = \mzero$.
  \item 
    \label{enu:projzerotrp}
    $ \QQ_{\BB^{\top}} \AA^{-1}\PP_{\MM^{\top}}= \mzero$.
  \end{enumerate}
\end{claim}
\begin{proof}
  \[ 
\PP_\MM \CC^{-1} \QQ_{\BB} = \MM^{\dagger} \MM \CC^{-1} \QQ_{\BB} 
= \MM^{\dagger} \AA \BB \CC \CC^{-1} \QQ_{\BB} = \MM^{\dagger} \AA \BB
\QQ_{\BB} = \MM^{\dagger} \AA \mzero = \mzero.
\]
and 

  \[ 
\QQ_{\BB^{\top}} \AA^{-1}\PP_{\MM^{\top}}
=
\QQ_{\BB^{\top}} \AA^{-1} \AA \BB \CC  \MM^{\dagger} 
=
\QQ_{\BB^{\top}} \BB \CC  \MM^{\dagger} 
=
\mzero \CC  \MM^{\dagger} 
=
\mzero.
\]
\end{proof}

\begin{proof}[Proof of \autoref{lem:pinvasymproduct}.]
Let $\NN = \PP_\MM \CC^{-1} \BB^{\dagger} \AA^{-1} \PP_{\MM^{\top}}$.
We want to show $\NN = \MM^{\dagger}$.
We prove this by verifying the four conditions of Fact~\ref{fac:pinvbysymtest}.
First we verify Condition~\ref{enu:mcancel}:
\[
\MM \NN \MM = 
\MM
\PP_\MM \CC^{-1} \BB^{\dagger} \AA^{-1}
\PP_{\MM^{\top}} \MM
=
\MM
\CC^{-1} \BB^{\dagger} \AA^{-1}
\MM
=
\AA \BB \CC
\CC^{-1} \BB^{\dagger} \AA^{-1}
\AA \BB \CC
=\AA \BB \BB^{\dagger} \BB \CC
=\MM
.
\] 
Second, we similarly verify Condition~\ref{enu:pinvcancel}:
\[
\NN \MM \NN= 
\PP_\MM \CC^{-1} \BB^{\dagger} \AA^{-1}
\PP_{\MM^{\top}}
\MM
\PP_\MM \CC^{-1} \BB^{\dagger} \AA^{-1}
\PP_{\MM^{\top}}
=
\PP_\MM
\CC^{-1} \BB^{\dagger} \AA^{-1}
\PP_{\MM^{\top}}
=
\NN
.
\] 
Third, we verify Condition~\ref{enu:pisym}.
First, observe that 
\[
\NN \MM 
=
\NN \MM \PP_\MM
=
\PP_\MM \CC^{-1} \PP_{\BB} \CC \PP_\MM
=
\PP_\MM \CC^{-1} (\PP_{\BB} + \QQ_{\BB}) \CC \PP_\MM
= 
\PP_\MM
,
\]
where to obtain the last equality we used Claim~\ref{clm:projzero}, Part~\ref{enu:projzero}.
Hence $(\NN \MM)^{\top} =\PP_\MM^{\top} = \PP_\MM = \NN \MM$, which
establishes the condition.
Condition~\ref{enu:pitrpsym} can be verified similarly.
\end{proof}

\subsection{Pseudo-Inverses and Schur Complements}
We now utilize the above characterization of projections
to prove the full characterization of Schur complements
as pseudoinverses as stated in \autoref{lem:SchurPinv}
Throughout the rest of this section, we will use
$C$ and $F$ to denote the partition of variables:
\[
\MM =
\left[\begin{array}{cc}
 \MM_{FF} & \MM_{FC}\\
 \MM_{CF} & \MM_{CC}
\end{array}
\right].
\]
and furthermore assume $\MM_{FF}$ is invertible.
Note that this assumption allows us to invoke
$\MM_{FF}^{-1}$, and can write the Schur complement as:
\[
\sc{\MM}{C} = \MM_{CC} - \MM_{CF} \MM_{FF}^{-1}\MM_{FC}.
\]

\begin{lemma}
\label{lem:schurkernel}
Under the assumptions at the start of this subsection,
\[
\ker(\MM) = \setof{
  \begin{pmatrix}
    -\MM_{FF}^{-1} \MM_{FC} \bb \\ 
    \bb
  \end{pmatrix}
\mid
\sc{\MM}{C} \bb = \mzero}
\]
\end{lemma}
\begin{proof}
Consider a vector $[\aa; \bb]$ in the null space of $\MM$
where $\aa$ is on the $F$ coordinates and $\bb$
is on the $C$ coordinates.
Invoking the block-wise characterization of $\MM$ gives:
\begin{align*}
\MM_{FF} \aa + \MM_{FC}\bb &= \mzero,\\
\MM_{CF} \aa + \MM_{CC}\bb &= \mzero.
\end{align*}
Since $\MM_{FF}$ is invertible, the first condition is
equivalent to
\[
\aa  = -\MM_{FF}^{-1}\MM_{FC}\bb
\]
and substituting this,
as well as the characterization of Schur complement,
into the second condition gives
\[
\sc{\MM}{C} \bb
= \left(-\MM_{CF} \MM_{FF}^{-1}\MM_{FC}
  + \MM_{CC}\right)\bb
= \MM_{CF} \aa + \MM_{CC}\bb = \zero
.
\]
\end{proof}

\begin{lemma}
\label{lem:composeproj}
Consider the assumptions from the start of this subsection.
For convenience of notation, let $\SS = \sc{\MM}{C}$,
and $ \mzero_{FF} $, $\mzero_{FC}$ and $\mzero_{CF}$
denote the $|F| \times |F|$, $|F| \times |C|$, and
$|C| \times |F|$ all-zeros matrices respectively.
Then 
  \[
    \PP_{\MM} 
\left[\begin{array}{cc}
 \mzero_{FF} & \mzero_{FC}\\
\mzero_{CF} & \PP_{\SS}
\end{array}\right]
=
\left[\begin{array}{cc}
 \mzero_{FF} & \mzero_{FC}\\
\mzero_{CF} & \PP_{\SS}
\end{array}\right]
    \]
\end{lemma}

\begin{proof}
    
Recall that the column space of $\MM$ can also be characterized
as the vectors $\yy$ orthogonal to the null space of $\MM$,
or formally
\[
\vv^{\top} \yy = \mzero
\qquad
\text{ for all $\vv$ s.t. $\MM \vv = \vzero$}.
\]
By \autoref{lem:schurkernel}, such vectors $\vv$
can be written as
\[
\vv =   \begin{pmatrix}
    -\MM_{FF}^{-1} \MM_{FC} \bb \\ 
    \bb
  \end{pmatrix}
  \]
where $\sc{\MM}{C} \bb = \vzero$.
  
Now consider a vector $\xx$ whose partition into coordinates
in $F$ and $C$ we denote as $[\xx_{F}; \xx_{C}]$.
We have
  \[
    \PP_{\MM} 
\left[\begin{array}{cc}
 \mzero_{FF} & \mzero_{FC}\\
\mzero_{CF} & \PP_{\SS}
\end{array}\right]
  \begin{pmatrix}
    \xx_{F} \\ 
    \xx_{C}
  \end{pmatrix}
=
    \PP_{\MM} 
  \begin{pmatrix}
    \vzero \\ 
    \PP_{\SS} \xx_{C}
  \end{pmatrix},
  \]
while by definition of $\PP_{\SS}$,
\[
\bb^{\top}\PP_{\SS} \xx_{C}=\vzero
\]
for any $\bb$ in the null space of $\sc{\MM}{C}$.
So for any  $\vv$ s.t. $\MM \vv = \vzero$,
we have 
\[
\vv^{\top} \begin{pmatrix}
    \vzero \\ 
    \PP_{\SS} \xx_{C}
  \end{pmatrix}
=
\bb^{\top}  \PP_{\SS} \xx_{C} = \vzero.
\]
Thus $\PP_{\MM}   \begin{pmatrix}
    \vzero \\ 
    \PP_{\SS} \xx_{C}
  \end{pmatrix} = \begin{pmatrix}
    \vzero \\ 
    \PP_{\SS} \xx_{C} \end{pmatrix}$.

This means for any $\xx$,
  \[
    \PP_{\MM} 
\left[\begin{array}{cc}
 \mzero_{FF} & \mzero_{FC}\\
\mzero_{CF} & \PP_{\SS}
\end{array}\right]
\xx
=
    \PP_{\MM} 
  \begin{pmatrix}
    \vzero \\ 
    \PP_{\SS} \xx_{C}
  \end{pmatrix}
=
  \begin{pmatrix}
    \vzero \\ 
    \PP_{\SS} \xx_{C}
  \end{pmatrix}
=
\left[\begin{array}{cc}
 \mzero_{FF} & \mzero_{FC}\\
\mzero_{CF} & \PP_{\SS}
\end{array}\right]
\xx
,
    \]
and the claim follows.

\end{proof}

\begin{proof}(of \autoref{lem:SchurPinv})
Let  $\SS = \sc{\MM}{C}$.
  Recall the standard factorization
  \begin{align}
  \MM =
  \begin{pmatrix}
    \II & \mzero \\
    \MM_{CF}\MM_{FF}^{-1} & \II
  \end{pmatrix}
  \begin{pmatrix}
  \MM_{FF} & \mzero \\
  \mzero & \SS
  \end{pmatrix}
  \begin{pmatrix}
    \II & \MM_{FF}^{-1}\MM_{FC} \\
    \mzero & \II
  \end{pmatrix}
.
\end{align}
By \autoref{lem:pinvasymproduct}
\begin{align*}
  \MM^{\dag}
 &=
\PP_{\MM}
   \begin{pmatrix}
    \II & \MM_{FF}^{-1}\MM_{FC} \\
    \mzero & \II
  \end{pmatrix}^{-1}
 \begin{pmatrix}
  \MM_{FF}^{-1} & \mzero \\
  \mzero & \SS^{\dag}
  \end{pmatrix}
  \begin{pmatrix}
    \II & \mzero \\
    \MM_{CF}\MM_{FF}^{-1} & \II
  \end{pmatrix}^{-1}
\PP_{\MM^{\top}}
\end{align*}
One can show (e.g. simple multiplication or by applying the formula for blockwise inversion)  that 
\begin{align*}
  \begin{pmatrix}
    \II & \mzero \\
    \BB & \II
  \end{pmatrix}^{-1}
=
\begin{pmatrix}
\II
&
\mzero
\\
-\BB
&
\II
\end{pmatrix} 
\text{ and } 
  \begin{pmatrix}
    \II & \BB \\
    \mzero & \II
  \end{pmatrix}^{-1}
=
\begin{pmatrix}
\II
&
-\BB
\\
\mzero
&
\II
\end{pmatrix} 
\end{align*}
So
\begin{align*}
\MM^{\dag} &=
\PP_{\MM}
\begin{pmatrix}
\II
&
-\MM_{FF}^{-1}\MM_{FC} 
\\
\mzero
&
\II
\end{pmatrix}
  \begin{pmatrix}
  \MM_{FF}^{-1} & \mzero \\
  \mzero & \SS^{\dag}
  \end{pmatrix}
\begin{pmatrix}
\II
&
\mzero
\\
-\MM_{CF}\MM_{FF}^{-1}
&
\II
\end{pmatrix} 
\PP_{\MM^{\top}}
.
\end{align*}
Now, by applying \autoref{lem:composeproj} to get rid of $\PP_{\MM}$
and $\PP_{\MM^{\top}}$, we get
\begin{align}
\label{eq:productschurres}
&
\left[\begin{array}{cc}
 \mzero_{FF} & \mzero_{CF}\\
\mzero_{CF} & \PP_{\SS}
\end{array}\right] 
\MM^{\dag} 
\left[\begin{array}{cc}
 \mzero_{FF} & \mzero_{FC}\\
\mzero_{CF} & \PP_{\SS^{\top}}
\end{array}\right]
\\
\nonumber
&=
\left[\begin{array}{cc}
 \mzero_{FF} & \mzero_{FC}\\
\mzero_{CF} & \PP_{\SS}
\end{array}\right] 
\begin{pmatrix}
\II
&
-\MM_{FF}^{-1}\MM_{FC} 
\\
\mzero
&
\II
\end{pmatrix}
  \begin{pmatrix}
  \MM_{FF}^{-1} & \mzero \\
  \mzero & \SS^{\dag}
  \end{pmatrix}
\begin{pmatrix}
\II
&
\mzero
\\
-\MM_{CF}\MM_{FF}^{-1}
&
\II
\end{pmatrix} 
\left[\begin{array}{cc}
 \mzero_{FF} & \mzero_{FC}\\
\mzero_{CF} & \PP_{\SS^{\top}}
\end{array}\right]
\\
\nonumber
&=
\left[\begin{array}{cc}
 \mzero_{FF} & \mzero_{FC}\\
\mzero_{CF} & \PP_{\SS} \SS^{\dag}\PP_{\SS^{\top}}
\end{array}\right]
\\
\nonumber
&=
\left[\begin{array}{cc}
 \mzero_{FF} & \mzero_{FC}\\
\mzero_{CF} & \SS^{\dag}
\end{array}\right]
,
\end{align}
where in the last equality, we used the fact that
$\SS^{\dag}$ has the same kernel as $\SS^{\top}$,
and $(\SS^{\dag})^{\top}$ has the same kernel as $\SS$
(which are in turn consequences of Fact~\ref{fac:pinvbysymtest}).
But, directly computing the matrix product also tells us that 
\begin{align}
\label{eq:productcoordres}
&
\left[\begin{array}{cc}
 \mzero_{FF} & \mzero_{FC}\\
\mzero_{CF} & \PP_{\SS}
\end{array}\right] 
\MM^{\dag} 
\left[\begin{array}{cc}
 \mzero_{FF} & \mzero_{FC}\\
\mzero_{CF} & \PP_{\SS^{\top}}
\end{array}\right]
\\
\nonumber
&=
\left[\begin{array}{cc}
 \mzero_{FC} & \mzero_{FC}\\
\mzero_{CF} & \PP_{\SS} (\MM^{\dag})_{CC} \PP_{\SS^{\top}}
\end{array}\right].
\end{align}
By comparing Equations~\eqref{eq:productschurres}
and~\eqref{eq:productcoordres}, we arrive at the desired conclusion,
after noting that by \autoref{def:projcoordres}, we have
$\MM^{\dagger}[C] 
=\PP_{\SS} (\MM^{\dag})_{CC} \PP_{\SS^{\top}}$,
and $\AA = \BB$ if and only if $\AA^{\dag} = \BB^{\dag}$.
\end{proof}

\begin{claim}
\label{clm:twostepschurinvertible}
Consider any $\MM \in \rea^{[n] \times [n]}$,
and $F,C$ a partition of $[n]$,  
where  $\MM_{FF}$ is invertible.
Let $F_1,F_2$ be a partition of $F$, and let $C_1 = [n] \setminus F_1$.
Suppose $\MM_{F_1 F_1}$ is invertible.
Then $(\sc{\MM}{C_1})_{F_2 F_2}$ is invertible.
\end{claim}
\begin{proof}
The key observation we need is that 
\[
(\sc{\MM}{C_1})_{F_2 F_2} = 
\sc{\MM_{FF}}{F_2}.
\]
This holds because 
\[
(\sc{\MM}{C_1})_{F_2 F_2} 
=(\MM_{C_1 C_1} - \MM_{C_1 F_1 } \MM_{F_1 F_1 }^{-1} \MM_{F_1 C_1
})_{F_2 F_2} 
=
\MM_{F_2 F_2} - \MM_{F_2 F_1 } \MM_{F_1 F_1 }^{-1} \MM_{F_1 F_2}
=
\sc{\MM_{FF}}{F_2}
.
\]
By a standard factorization, we have
\[
  \MM_{FF} =
  \begin{pmatrix}
    \II & \mzero \\
    \MM_{F_2F_1}\MM_{F_1F_1}^{-1} & \II
  \end{pmatrix}
  \begin{pmatrix}
  \MM_{F_1F_1} & \mzero \\
  \mzero & \sc{\MM_{FF}}{F_2}
  \end{pmatrix}
  \begin{pmatrix}
    \II & \MM_{F_1F_1}^{-1}\MM_{F_1F_2} \\
    \mzero & \II
  \end{pmatrix}
,
\]
from which we conclude that unless both $\MM_{F_1F_1}$ and
$\sc{\MM_{FF}}{F_2}$ are full rank, $\MM_{FF}$ is not full rank.
As $\MM_{FF}$ is invertible and hence full rank, we conclude
$\sc{\MM_{FF}}{F_2}$ is full rank and hence invertible.
\end{proof}

\begin{proof}[Proof of \autoref{lem:SchurPinvCompose}.]
Note that by \autoref{lem:SchurPinv}, 
\[
 \sc{ \sc{\MM}{C_1} }{C} =  \sc{\MM}{C}
\]
is equivalent to  
\[
\left(\MM^\dag[C_1]\right)\left[C\right]
=
\left(\MM^\dag\left[C\right]\right),
\]
so it suffices to show the latter.
We assume that $\MM_{FF}$ and $\MM_{F_1F_1}$ are both invertible. 
Hence by Claim~\ref{clm:twostepschurinvertible}, 
$(\sc{\MM}{C_1})_{F_2 F_2}$ is invertible, which ensure the Schur
complement $\sc{ \sc{\MM}{C_1} }{C}$ is well-defined.
Let $\SS = \sc{\MM}{C}$ and  $\SS_1 = \sc{\MM}{C_1}$, and  $\TT =
\sc{\sc{\MM}{C_1}}{C} =\sc{\SS_1}{C} $.
Next we observe that
\begin{align*}
(\MM^\dag[C_1])\left[C\right]= \SS_1^{\dag} \left[C\right]
=
\PP_{\TT} (\SS_1^{\dag})_{C C}\PP_{\TT^{\top}}
&=
\PP_{\TT} (
\PP_{\SS_1} (\MM^{\dag})_{C_1 C_1} \PP_{\SS_1^{\top}}
)_{C C}\PP_{\TT^{\top}}
\\
&=
\left(
\begin{pmatrix}
 \mzero_{F_2 F_2} & \mzero_{F_2 C}\\
\mzero_{C F_2} & \PP_{\TT}
\end{pmatrix}
\PP_{\SS_1} (\MM^{\dag})_{C_1 C_1} \PP_{\SS_1^{\top}}
\begin{pmatrix}
 \mzero_{F_2 F_2} & \mzero_{F_2 C}\\
\mzero_{C F_2} & \PP_{\TT^{\top}}
\end{pmatrix}
\right)_{C C}
\end{align*}
But, by \autoref{lem:composeproj}, 
\[
\PP_{\SS_1^{\top}}
\begin{pmatrix}
 \mzero_{F_2 F_2} & \mzero_{F_2 C}\\
\mzero_{C F_2} & \PP_{\TT^{\top}}
\end{pmatrix}
=
\begin{pmatrix}
 \mzero_{F_2 F_2} & \mzero_{F_2 C}\\
\mzero_{C F_2} & \PP_{\TT^{\top}}
\end{pmatrix},
\text{ and } 
\PP_{\SS_1}
\begin{pmatrix}
 \mzero_{F_2 F_2} & \mzero_{F_2 C}\\
\mzero_{C F_2} & \PP_{\TT}
\end{pmatrix}
=
\begin{pmatrix}
 \mzero_{F_2 F_2} & \mzero_{F_2 C}\\
\mzero_{C F_2} & \PP_{\TT}
\end{pmatrix}
.
\]
So
\begin{align*}
(\MM^\dag[C_1])\left[C\right]
&=
\left(
\begin{pmatrix}
 \mzero_{F_2 F_2} & \mzero_{F_2 C}\\
\mzero_{C F_2} & \PP_{\TT}
\end{pmatrix}
(\MM^{\dag})_{C_1 C_1} 
\begin{pmatrix}
 \mzero_{F_2 F_2} & \mzero_{F_2 C}\\
\mzero_{C F_2} & \PP_{\TT^{\top}}
\end{pmatrix}
\right)_{C C}
&=
\PP_{\TT} (\MM^{\dag})_{C C} \PP_{\TT^{\top}}
\end{align*}
We can show by \autoref{lem:schurkernel}, that 
$\PP_{\TT} = \PP_{\SS}$, 
because $\ker(\TT) = \ker(\SS)$ as both kernel arise as coordinate
restrictions of $\ker(\MM)$.
Formally $\ker(\TT) = \setof{ \bb_{C} : \bb \in \ker(\SS_1) }$
, and $\ker(\SS_1)= \setof{ \bb_{C_1} : \bb \in \ker(\MM) }$, so
$\ker(\TT) = \setof{ \bb_{C} : \bb \in \ker(\MM) } = \ker(\SS)$.
Similarly, we get $\PP_{\TT^{\top}} = \PP_{\SS^{\top}}$.
Thus 
$(\MM^\dag[C_1])\left[C\right] = \PP_{\SS} (\MM^{\dag})_{C C}
\PP_{\SS^{\top}} = \MM^\dag[C]$.
\end{proof}


\end{document}